\documentclass[journal-jpcc, manuscript=letter,layout=onecolumn]{achemso}

\setkeys{acs}{chaptertitle = true}

\usepackage{geometry}
\usepackage{enumerate}
\usepackage{graphicx,subfigure,mciteplus}
\usepackage{amssymb, amsmath}
\usepackage[flushleft]{threeparttable}
\usepackage{comment}
\usepackage{epstopdf, booktabs}
\usepackage{multirow,bigdelim}
\usepackage{ragged2e}
\makeatletter
\g@addto@macro\TPT@defaults{\footnotesize}
\makeatother
\usepackage{bm}
\usepackage{dcolumn}
\usepackage{bm}
\usepackage{longtable}
\usepackage{tabularx}
\usepackage{multirow}
\usepackage{ulem}
\usepackage[version=3]{mhchem}
\usepackage{color}
\usepackage[latin1]{inputenc}
\usepackage{array}
\usepackage{float}
\usepackage{leftidx}
\usepackage{placeins}
\usepackage[latin1]{inputenc}
\usepackage[T1]{fontenc}
\usepackage{multirow}
\usepackage{booktabs}
\usepackage{siunitx}
\usepackage{footnote}
\usepackage{dcolumn}
\usepackage{catchfilebetweentags}
\usepackage{xr-hyper}
\usepackage[nounderscore]{syntax}
\usepackage[version=3]{mhchem}
\usepackage{dcolumn}
\usepackage{siunitx}
\usepackage{amsmath}
\usepackage{amsthm}
\newtheorem{theorem}{Theorem}[section]

\newtheorem{lemma}{Lemma}
\usepackage{hyperref}
\hypersetup{colorlinks,
	citecolor=blue
}
\usepackage{cleveref}

\DeclareMathOperator*{\argmin}{arg\,min}
\mciteErrorOnUnknownfalse

\usepackage{geometry}
\usepackage{graphicx,subfigure,mciteplus}
\usepackage{amssymb, amsmath}
\usepackage[flushleft]{threeparttable}
\usepackage{comment}
\usepackage{epstopdf, booktabs}
\usepackage{mathtools}
\usepackage{multirow,bigdelim}
\usepackage{ragged2e}
\makeatletter
\g@addto@macro\TPT@defaults{\footnotesize}
\makeatother
\usepackage{bm}
\usepackage{dcolumn}
\usepackage{bm}
\usepackage{longtable}
\usepackage{tabularx}
\usepackage{multirow}
\usepackage{ulem}
\usepackage[version=3]{mhchem}
\usepackage{color}
\usepackage{amsthm}
\usepackage[latin1]{inputenc}
\usepackage{float}
\usepackage{leftidx}
\usepackage{placeins}
\usepackage[latin1]{inputenc}
\usepackage{catchfilebetweentags}
\usepackage{xr-hyper}
\usepackage[nounderscore]{syntax}
\usepackage[version=3]{mhchem}
\usepackage{dcolumn}
\usepackage{siunitx}
\usepackage{hyperref}
\usepackage{amsmath}
\usepackage{hyperref}
\hypersetup{colorlinks,
	citecolor=blue
}
\usepackage{cleveref}
\mciteErrorOnUnknownfalse


\newcolumntype{L}{D{.}{.}{2,5}}

\newcolumntype{L}{D{.}{.}{2,5}}

\title[An \textsf{achemso} demo]
  {Quantum Machine-Learning for Eigenstate Filtration in Two-Dimensional Materials}

\author{Manas Sajjan}
\affiliation{Department of Chemistry, Purdue University,
West Lafayette, IN-47907, USA}

\author{Shree Hari Sureshbabu}
\affiliation{Elmore Family School of Electrical and Computer Engineering, Purdue University, West Lafayette, IN-47907, USA}

\author{Sabre Kais}
\affiliation{Department of Chemistry, Department of Physics and Astronomy, and Purdue Quantum Science and Engineering Institute, Purdue University, West Lafayette, IN-47907, USA}
\email{kais@purdue.edu}

\SectionNumbersOn
\begin{document}

\begin{abstract}

Quantum machine learning algorithms have emerged to be a promising alternative to their classical counterparts as they leverage the power of quantum computers. Such algorithms have been developed to solve problems like electronic structure calculations of molecular systems and spin models in magnetic systems. However, the discussion in all these recipes {\color{black} focuses} specifically on targeting the ground state. Herein we demonstrate a quantum algorithm that can filter any energy eigenstate of the system based on either symmetry properties or a predefined choice of the user. The workhorse of our technique is a shallow neural network encoding the desired state of the system with the amplitude computed by sampling the Gibbs - Boltzmann distribution using a quantum circuit and the phase information obtained classically from the non-linear activation of a separate set of neurons. We show that the resource requirements of our algorithm are strictly quadratic. To demonstrate its efficacy, we use state-filtration in monolayer transition metal-dichalcogenides which are hitherto unexplored in any flavor of quantum simulations. We implement our algorithm not only on quantum simulators but also on actual IBM-Q quantum devices and show good agreement with the results procured from conventional electronic structure calculations. We thus expect our protocol to provide a new alternative in exploring band-structures of exquisite materials to usual electronic structure methods or machine learning techniques that are implementable solely on a classical computer.

\end{abstract}

\section{Introduction}

Machine learning concerned with identifying and utilizing patterns within a data set has gained tremendous importance within the last decade. Even though the germinal idea can be traced back to the 1950s \cite{COHEN20211}, it is safe to say that the domain has become a pioneering field of research within the last few years due to escalation in computational prowess and data availability, and have metamorphosed several disciplines including autonomous driving \cite{Auto_drive}, image-recognition\cite{he2016deep}, speech recognition\cite{sak2015learning}, natural language processing \cite{Hoy2018}, computer games \cite{Silver2016}, and even refugee integration\cite{Bansak2018}. Consequently the integration of the technique in solving problems of physico-chemical interest \cite{RevModPhys.91.045002} have also been explored with remarkable success whether in predicting ground-state density functionals\cite{PhysRevLett.108.253002, Brockherde2017}, self-energy in Dynamical Mean-Field Theory (DMFT) for {\color{black} the} Anderson model\cite{Arsenault2014}, atomistic potentials and forcefields for molecular dynamics\cite{PhysRevLett.104.136403, Li2015} or unsupervised learning of phases of {\color{black} the} 2D-Ising Hamiltonian\cite{Wang2016a}. Similar advancements have also been made in the field of Deep Learning\cite{Lecun2015} and Artificial Neural Networks (ANN) which has been successfully used to learn phase transition parameters\cite{Carrasquilla_2017, ChNg2017} or in quantum phase recognition \cite{Cong_2019}. Among the various architectures in this category, Restricted Boltzmann machine (RBM) based generative models being a universally powerful approximator for any probability density\cite{Roux_RBM, Melko2019a} have particularly {\color{black} gained} attention. RBMs have been successfully used to reconstruct quantum states in tomography from measurement statistics\cite{torlai2018neural}. Carleo and Troyer showed how a neural network encoding a shallow RBM ansatz requires fewer parameters than certain {\color{black} kinds of matrix product states} 
and can predict the ground state energy and unitary dynamical evolution of simple spin models with high accuracy \cite{carleo2017solving}.

However, all the algorithms discussed above have trained machine learning or deep-learning models on a classical computer to effectively recreate either a quantum state or its essential features. The past decade {\color{black} has} also witnessed unprecedented development in  quantum computing as a new paradigm which is fundamentally different than its classical counterpart in processing and storing data and performing logical operations\cite{preskill2018quantum} harnessing the power of quantum superposition and non-classical correlations like entanglement. A natural question that has spawned is whether such quantum machines can interpret and produce statistical patterns in data which are either difficult for classical machine learning algorithms or the performance {\color{black} of machine learning algorithms on quantum computer can outperform the classical variants in efficiency }\cite{Biamonte2017}. This has naturally motivated the development of a host of quantum -machine learning algorithms like Quantum Principal Component Analysis (PCA) \cite{Lloyd2014}, Quantum Support Vector Machines (QSVM)\cite{rebentrost2014quantum}, Quantum Reinforcement Learning \cite{Dunjko2016}, quantum supervised and unsupervised learning \cite{lloyd2013quantum}, {\color{black} kernel design for Gaussian processes\cite{Gauss_kernel}, Gaussian process regression\cite{PhysRevA.99.052331}}, quantum classifier\cite{neven2009training} or a plethora of linear algebra routines like HHL\cite{PhysRevLett.103.150502}, QSVD\cite{PhysRevA.97.012327}, qBLAS\cite{QBlas} which forms the backbone of the quantum versions of many other machine learning algorithms. Each of these methods has reported theoretical speedup over the best-known classical algorithm under certain specific circumstances \cite{Ciliberto2017}. Similar investigations have also been undertaken for artificial neural networks to discover any unforeseen quantum advantage. For instance, Amin and co-workers have demonstrated a Quantum Boltzmann Machine \cite{Amin2018} by adding an off-diagonal transverse field to the training model thereby making it more expressive to treat larger classes of problems \cite{PhysRevA.96.062327}. Weibe \textit{et al} have shown how sampling from a Gibbs distribution as is required for training an RBM can be distinctly accelerated using a quantum processor\cite{wiebe2014quantum}.

Motivated by such recent developments, Xia and Kais~\cite{Xia2018} proposed an actual quantum circuit using polynomial resources to correctly learn the amplitude of the RBM ansatz encoded within a neural network representing the state of a quantum system. The work also extended the neural network to  three layers to learn the sign of the various components of the encoded wavefunction. The algorithm was benchmarked by showing the evaluation of ground states on simple molecular systems like $\rm{H}_2$, $\rm{LiH}$, etc thereby formally extending the efforts mentioned above to actual electronic structure calculations which are considered to be powerful applications of near-term quantum devices. {\color{black}
Indeed, interesting algorithmic advances have been made recently that can capture both the ground and excited state of such electronic structure problems with good accuracy \cite{Alan_2005, PhysRevLett.126.070504, doi:10.1002/9781118742631.ch01, Peruzzo2014, Hardware_eff_Nat,  DASKIN201887, VQE_review}.}
Kanno {\it et al.}~\cite{kanno2019manybody} modified the above method to encompass the complex phase of each component of the wavefunction by adding an additional neuron to the third layer. {\color{black} However both the work simulated the performance of the algorithm for ground states only on noiseless classical devices. In fact, due to conditional dependence on the sequence of measurements of the ancilla register, straight-forward implementation of the algorithm on a present-day actual NISQ device is difficult.

The main contributions of this manuscript are thus the following:- 1) Unlike previous efforts we focus our attention beyond just the ground state and devise a quantum machine learning algorithm with a three-layered RBM being trained to learn any arbitrary state of the system retaining the quadratic resource requirements.  To train the network with the RBM ansatz, we employ a hardware-implementable version of the above quantum circuit which as we shall discuss explicitly makes our algorithm require quadratic resources in all fronts like circuit width, circuit depth, parameter count.
2) A generic lower bound for the successful sampling of the quantum circuit in the algorithm is derived in terms of the parameters of the network. The performance of the lower bound is thoroughly characterized and 
specific limiting cases leading to known bounds are formally deduced and discussed. Based on this we also discuss in detail the measurement statistics of our algorithm for systems studied in this report and in general.
3) We also present a simple yet formal proof of feasibility of the cost function used to train the network. Even though such functions have been used in classical algorithms and are beginning to being used in other quantum algorithms beyond the precincts of quantum machine learning, a formal proof is lacking in literature which we supply here for completeness.
4) Furthermore unlike most reports on quantum machine learning and quantum computing in general which have studied molecular systems only, we apply our algorithm on important 2D materials like monolayer transition metal di-chalcogenides (TMDCs) which are hitherto unexplored on a quantum computer using any algorithm let alone quantum machine learning. These materials have been shown to possess tunable band-gap for many novel applications\cite{Choi2017, https://doi.org/10.1002/nano.202000047, Lv2015, Manzeli2017}. We make a comprehensive study of such materials by showing how our algorithm can not only learn the true band gap but also by resolving finer yet important features like trigonal warping and spin-orbit coupling (SOC) which dictates the low-energy physics near the K-valley. The importance of understanding excited states beyond just the valence band for such periodic materials underlies its function in photovoltaics\cite{C8CS00067K, ALHARBI20151073}. If applied to other systems, excited states can be an insightful resource like in elucidating the reaction pathways across conical {\color{black}intersections} arising in processes like vision\cite{Palczewski2012, PhysRevLett.98.023001}, photosynthesis\cite{Cerullo2395, C8CP05535A}, magneto-reception\cite{Rodgers353, PhysRevE.90.042707}, and even bio-chemistry of luciferin \cite{Chung2008} to name a few.
5) We further demonstrate in a unified way how a user can sieve any desired state in such materials using not only energy as in point (4) but other inherent symmetries of the Hamiltonian too.
6) All numerical experiments are implemented on not only quantum simulator (Qiskit) but on actual NISQ devices using the quantum processors at IBM~\cite{aleksandrowicz2019qiskit}. The performance of the algorithm is benchmarked thoroughly in each case using quantifiers like energy errors of the target state, state composition, constraint violation, infidelity with the target state learnt by the neural network etc. The usage of certain kinds of error-mitigation techniques and role of initial parameterization, enhancing the capacity of the model through additional spins in the network is thoroughly discussed. We have also included results from a molecular example wherein multi-reference correlation is important due to geometric distortion. To the best of our knowledge, all of these are first of its kind in any flavor of quantum-machine learning. We show that the performance of our algorithm is in excellent agreement with the exact value in each case.}

The organization of the paper is as follows. In section \ref{Sec:Theory} we discuss the theoretical underpinning of our algorithm with an original proof of the feasibility of our cost function employed for training the network. In section \ref{Sec:Algo} we elaborate on the geometry of the network and the details of the algorithm required for learning the desired feature with the {\color{black} associated} resource requirements and implementation techniques. We prove an explicit lower bound on the probability of successful {\color{black} events} on our algorithm (see Section 2 Supplementary Information). In section \ref{Sec:results} we discuss the application of the algorithm in simulating excited states or any arbitrary states in two important TMDCs - $\rm{MoS}_2$ and $\rm{WS}_2$ based on user-defined constraints. We conclude in section \ref{Sec:conclusion} with a brief discussion of possible future extensions. 

\section{Theory}
\label{Sec:Theory}

Our objective is to develop an efficient algorithm to train a neural network to perform the following minimization in a $d$-dimensional space
\begin{align} \label{Const_main}
    &\min_{\forall \psi \in S}\:\:\: \langle \psi| \hat{H} |\psi \rangle \nonumber \\
    S &= \{|x\rangle\:\: |\:\:\: \hat{O}|x\rangle = \omega|x\rangle\:\: \forall\: |x\rangle \in \mathbb{C}^d\}
\end{align}
where $\hat{H} \in \mathbb{C}^{d\times d}$ is the hermitian Hamiltonian defining the problem. Similarly $\hat{O} \in \mathbb{C}^{d\times d}$ is the user-defined hermitian operator. $\omega$ is the eigenvalue (real-valued) of the operator  
$\hat{O}$ and $|x\rangle$ is the corresponding eigenvector. {\color{black} The set $S$ is the collection of all such eigenvectors with a specific eigenvalue $\omega$.}  
The operators $\hat{O}$ which we shall discuss will generally have more than one element in set $S$ due to degeneracy in the eigenspace labelled by $\omega$. By construction, the form of the algorithm shall always normalize the state $|\psi\rangle$ and hence normalization as a further constraint is unnecessary. We will return to this point later. {\color{black}
The primary goal of the network is to then encode a normalized state-vector $|\psi\rangle$ which is a formal solution to Eq.\ref{Const_main}. The corresponding state so obtained is from the eigenspace of $\hat{O}$ with eigenvalue $\omega$. If several such choices exist, the network learns the one with minimum energy.}

To solve the quadratic minimization problem with quadratic constraint in Eq. \ref{Const_main} we will define a penalty procedure as 

\begin{align} \label{cost_fn}
    F(|\psi\rangle, \hat{H}, \hat{O}, \lambda) = \langle \psi| \hat{H} |\psi \rangle + \lambda \langle \psi| (\hat{O}-\omega)^2 |\psi \rangle 
\end{align}
where $\lambda \ge 0$ is the penalty parameter. We provide a formal and original proof of equivalence of Eq. \ref{cost_fn} with respect to Eq. \ref{Const_main} based on the following Theorem.

\begin{theorem} \label{main_thm}
Let $\{\lambda_i\}_{i=1}^{\infty}$ be a sequence in the penalty parameter such that $\lambda_1 \le \lambda_2 \le \lambda_3..... \lambda_{\infty} \rightarrow \infty$.
Also let $P=\{|\psi_i\rangle\}_{i=1}^{\infty}$ such that $\forall$ $|\psi_i\rangle$ $\in$ $P$ the following is true.
\begin{align}
    |\psi_i\rangle &= \argmin_{\psi} {\rm{F}(\lambda_i, \hat{H}, \hat{O}, |\psi\rangle)}
\end{align}
In other words, $P$ is the set of minimizers for Eq. \ref{cost_fn} for each penalty parameter $\lambda \in \{\lambda_i\}_{i=1}^{\infty}$.
If $|\psi^*\rangle \in P$  is a limit-point of the convergent sequence $\{\psi_i\}_{i=1}^\infty$ in $P$ i.e $|\psi^*\rangle = \lim_{i \to \infty} |\psi_i\rangle$ then $|\psi^*\rangle$ $\in$ $S$
\end{theorem}

An original proof of Theorem \ref{main_thm} is in Section 1 of the Supporting Information based on the fact that both the 1st and 2nd term in Eq. \ref{cost_fn} are quadratic forms. An intuitive explanation can be provided that would suffice to appreciate the discussion in this report. One can note that in the cost function defined in Eq. \ref{cost_fn} the term $\langle\psi|\hat{H}|\psi\rangle$ imposes the minimization of energy as required in Eq. \ref{Const_main}. The second term i.e. $\langle\psi|(\hat{O}-\omega)^2|\psi\rangle$ is the variance of the operator $\hat{O}$ with the mean being the eigenvalue $\omega$ and is non-negative by construction. For {\color{black}large values of the penalty parameter} $\lambda$, the minimization of the overall cost function is afforded if the variance term is pinned to zero i.e. the state $|\psi^*\rangle$ so chosen is an eigenstate of the operator $\hat{O}$ with eigenvalue $\omega$. The space of such states is defined by the set $S$ in Eq. \ref{Const_main}. If several such choices exist, the role of the 1st term kicks in to guarantee optimality in energy.

While penalized optimization schemes with cost function of the kind in Eq. \ref{cost_fn} {\color{black}has been employed in classical algorithms like in Density Matrix Renormalization Group (DMRG)\cite{doi:10.1146/annurev-conmatphys-020911-125018}, in Quantum Monte-Carlo methods in the past \cite{PhysRevLett.60.1719} and even recently \cite{doi:10.1063/5.0030949}} are also {\color{black}} beginning to gain attention in recent literature on quantum algorithms beyond quantum-machine learning i.e. in algorithms using Unitary-Coupled Cluster Ansatz (UCC) of variational quantum eigensolver (VQE) \cite{Kuroiwa2021}, yet a formal proof is lacking. Besides a more popular choice that has been studied in some detail is constraining the average value of the operator $\langle\psi|\hat{O}|\psi\rangle$ \cite{Ryabinkin2019, Greene-Diniz2021} with the required eigenvalue instead of penalizing the variance as in Eq. \ref{cost_fn}. However, this recent study \cite{Kuroiwa2021} shows Eq. \ref{cost_fn} is a better penalty procedure in terms of feasibility and final error than restraining the average without providing a formal proof of equivalence between Eq. \ref{cost_fn} and Eq. \ref{Const_main}. Ref\cite{Kuroiwa2021} also implemented the same to target symmetry operators on molecular systems using UCC-VQE using Qulacs\cite{suzuki2020qulacs} which is an ideal simulator of a real quantum computer. However, in this report, we shall use Eq. \ref{cost_fn} to develop and train a shallow neural network using a quantum machine learning algorithm with quadratic resource requirements in terms of the size of qubit register, number of gates and parameter counts. The ansatz which the neural network would encode for the quantum state $|\psi\rangle$ would correspond to a probability density represented by
RBM. We benchmark our algorithm on important 2D periodic materials like transition metal di-chalcogenides (TMDCs) and show implementations not only on quantum simulators but on actual NISQ devices (IBM-Q). TMDCs have never been studied before using any quantum algorithm. In the next few sections, we shall show how to filter any specific state of these 2D materials using either symmetry operators of the Hamiltonian or user-defined constructions of operator $\hat{O}$ in a unified manner using the same algorithm. Such an attempt to the best of our knowledge is the first of its kind in QML as all previous reports have focused exclusively on targeting the ground state of the system alone 
\cite{kanno2019manybody,Sureshbabu2021}.

\subsection{Filter for specific excited states} \label{exc_state_section}

To target the first excited state of the system, one can use a user-defined operator ($\hat{O}= |g\rangle \langle g|, \omega=0$) where $|g\rangle$ is the ground state of the system obtained by training the network in a previous computation with $\lambda=0$ in Eq. \ref{cost_fn}. In essence, we require the neural-network to return a state-vector in the null space of operator $|g\rangle \langle g|$. Since the null-space is $d-1$ dimensional, the minimum energy criterion as enforced by the 1st term in Eq. \ref{cost_fn} guarantees the first excited state. This method using the penalty program in Eq. \ref{cost_fn} is formally equivalent to deflation technique if one recognizes the idempotency of $\hat{O}= |g\rangle \langle g|$. Deflation has been the cornerstone of many classical algorithms in the past for obtaining excited states\cite{Miranda-Quintana2013, Shayan_2018} and even a quantum algorithm as well with UCC-VQE \cite{Higgott2019}. But the formal reduction of our penalty procedure to deflation in Eq. \ref{cost_fn} based on Theorem \ref{main_thm} offers a slightly different perspective. Moreover as we shall see shortly, the penalty program in Eq. \ref{cost_fn} is more general and can be used to sieve any state based on arbitrary operator $\hat{O}$. For higher excited states (say the $t$ th) one can add similar terms to Eq. \ref{cost_fn} with the set $\hat{\{O_i\}}_{i=i}^{t-1}$ which forms a set of commuting operators with progressively refined null-space. For the choice of the penalty parameter $\lambda$ in Eq. \ref{cost_fn}, one can choose any number greater than the spectral range of the Hamiltonian $\hat{H}$ as that would always work. The spectral range can be computed from the knowledge of the ground state and $||\hat{H}||_2$.

\subsection{Filter for arbitrary states using symmetry operators}

Eq. \ref{cost_fn} can be used to solve a more general problem with any symmetry operator of the system $\hat{O}$ (by definition such operators satisfy $[\hat{O}, \hat{H}]=0$ and hence share the same eigenspace). The corresponding user-desired eigenvalue $\omega$ labels the symmetry sector (set $S$ in Eq. \ref{Const_main}). Unlike in the previous case in section \ref{exc_state_section}, usual symmetry operators need not satisfy idempotency and hence relaxation to deflation is impossible. To demonstrate our point, here we shall use $\hat{O}=L^2$ where $L^2$ is the squared-orbital angular momentum operator, a symmetry for 2D materials. $\omega$ would be set to the desired eigenvalue of $L^2$ . We shall see that the network will always learn the lowest energy eigenstate correctly despite multiple-fold degeneracy. To sieve other states from the entire degenerate subspace one can use a combination filter of $\hat{O_1}=L^2$ and $\hat{O_2}=|v\rangle\langle v|$ where $|v\rangle$ is the lowest energy state in the symmetry subspace obtained from the RBM. The penalty parameter $\lambda$ can be chosen using the prescription in \cite{Kuroiwa2021}.

\section{Algorithm} \label{Sec:Algo}

\begin{figure}[!htb]
    \centering
    \includegraphics[width=1.05\textwidth]{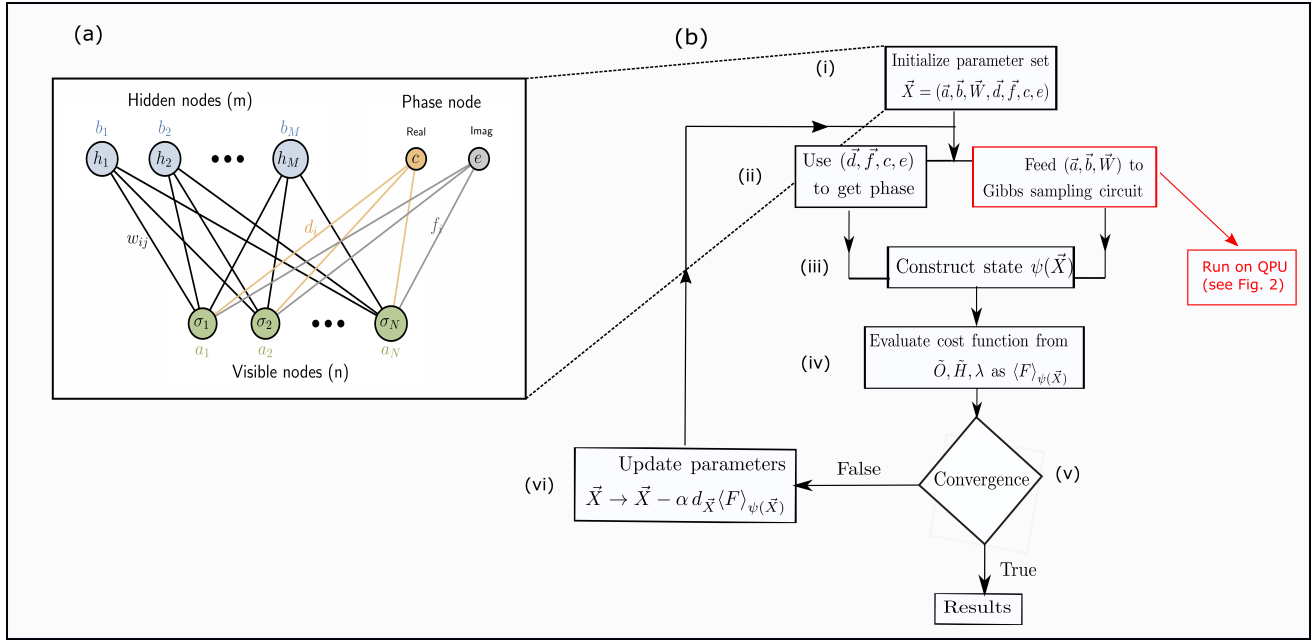}
    \caption{{\color{black}(a) The RBM architecture used in this work. The visible node contains $n$ neurons (green), the hidden node has $m$ neurons (blue) and the phase node contains 2 neuron, one to model the real part(orange) of the phase of the wavefunction and the other to model the imaginary part (grey). The weights and biases of the respective units are displayed. The RBM ansatz for the required state is defined from the Boltzmann distribution over the state-space of the visible-hidden units (b) The QML algorithm used to perform the variance penalized optimization. The part of step (ii) marked within the red box is performed on a quantum processor (QPU). All other steps are performed on a classical computer. Each step is marked with a roman numeral. We follow each of these roman numerals for discussing the algorithm in section \ref{Algo_out}}}
    \label{Fig:RBM_algo}
\end{figure}

\subsection{The Model} \label{The_Model}

{\color{black}
In the early 1980s, Hopfield networks \cite{Hopfield2554} defined a probability distribution over a set of random variables which is encoded within the nodes of a unidirected graph using the physical notion of energy of interaction between the nodes. Boltzmann machines (BM) are extensions of such a network that categorizes the node space into visible/physical layer and hidden/latent layers maintaining all to all connectivity\cite{ACKLEY1985147}. Restricted Boltzmann Machine (RBM) \cite{10.1162/089976602760128018, Fischer2014, Hinton504, Mehta2019, RevModPhys.91.045002, Melko2019a} is a practically useful sub-category of BM which permits interaction only between the visible layer and hidden layer.
The energy function
used in the RBM model is thus that of a partially connected classical Ising network and the ansatz for the probability distribution is the corresponding thermal distribution. The ansatz is optimized to mimic the underlying probability distribution of the given data using free parameters called $\textit{weights}$ and $\textit{biases}$\cite{Hinton504, Fischer2014,Mehta2019, Roux_RBM, doi:10.1080/00018732.2017.1341604,Graph_reg_2021,Graph_reg_2018}. The goal of this paper is to use the RBM distribution to encode the amplitude field of an arbitrary quantum state $|\psi\rangle$ which is a solution to Eq. \ref{Const_main}. Such neural-network quantum states (NQS) have been successfully employed in a variety of problems recently \cite{carleo2017solving, RevModPhys.91.045002, Melko2019a, PhysRevB.97.035116} by training the \textit{weights} and $\textit{biases}$ using a classical computer. Herein we shall train the network by constructing the RBM distribution using a quantum circuit and discuss the quantum advantages.

The RBM network we use in this report consists specifically of three layers each having multiple neurons. The schematic of the network architecture is presented in Fig.
\ref{Fig:RBM_algo}(a). The first layer is the visible node consisting of $n$ neurons, the second layer is the hidden node consisting of $m$ neurons and the last layer is a phase node consisting of two neurons. While the $n$ neurons are responsible for encoding the actual state, the purpose of the hidden neurons $m$ is to add more controllable parameters to make the joint probability distribution (to be defined in Eq. \ref{rbm_dist} soon) more expressive and induce higher order correlation among $n$ neurons \cite{Melko2019a}. Variables encoded by the visible node neurons (henceforth denoted by $\{\sigma_i\}_{i=1}^n$) and those by the hidden node neurons (henceforth denoted by $\{h_j\}_{j=1}^m$) are both binary random variables as $\sigma_i$ and $h_j$ $\in \{1, -1\}$.
As depicted in Fig. \ref{Fig:RBM_algo} (a), the bias vector of the visible neurons is denoted as $\Vec{a} \in \mathbb{R^{\rm{n}}}$, bias vector of hidden neurons is denoted as $\Vec{b} \in \mathbb{R^{\rm{m}}}$, the interconnecting weights of the visible and hidden neurons are denoted as $\Vec{W} \in \mathbb{R^{\rm{n} \times \rm{m}}}$. The joint RBM distribution\cite{Fischer2014, Melko2019a, Mehta2019, Hinton504,RevModPhys.91.045002} $P(\Vec{a}, \Vec{b}, \Vec{W}, \vec{\sigma},\vec{h})$ defined over the variables ($\vec{\sigma}, \vec{h}$) is 
\begin{align}
P(\Vec{a}, \Vec{b}, \Vec{W}, \vec{\sigma},\vec{h}) &= \frac{e^{\sum_{i}a_i\sigma_i + \sum_{j}b_j h_j + \sum_{ij}w_{ij}\sigma_i h_j}}
{\sum_{\{\sigma h\}}e^{\sum_{i}a_i\sigma_i + \sum_{j}b_j h_j + \sum_{ij}w_{ij}\sigma_i h_j}} \label{rbm_dist}
\end{align}

For an electronic Hamiltonian with $r$ spin-orbitals and $N$ electrons, a naive Jordan Wigner mapping (JW) \cite{JW_vs_BK} would make $n=r$ or (or $n \approx O(log_2 (r))$ for Bravi-Kitaev mapping \cite{JW_vs_BK}). However, it is well understood now that qubit requirements can be tapered by using additional symmetries like $Z_2$ \cite{Setia_Z2}. Chemically inspired process of reducing qubit cost  like using an active space \cite{PhysRevX.10.011004} (wherein number of physical qubits required is still $r$ but logical qubits required are much less as some qubits have frozen occupation/eigenvalue with $Z$ operator) or using point-group symmetry or angular-momentum symmetry of the required state\cite{Setia_Z2, PhysRevResearch.3.013039} are also being recently employed. Thus a direct relationship of $n$ with $r$ would depend on the specification of the mapping and tapering used. Whatever may be the method, if the final $\hat{H}$ matrix is $\mathbb{C}^{d\times d}$ (as used in Section \ref{Sec:Theory}) then it is safe to say that in our model $n=\rm{log}_2(d)$. The number of hidden units $m$ in our model is user-defined (for almost all data in this manuscript we have used $n=m$) but the hidden node density $\alpha = \frac{m}{n}$ can be tuned to enhance the final accuracy desired. We shall return to this point later. Neurons in the phase node are always 2 in number.

The purpose of the neurons in the phase node is to account for complex values and capture the phase of the wavefunction \cite{kanno2019manybody} unlike in conventional two-layer RBM networks\cite{torlai2018neural} which faithfully recovers only the amplitude.
As shown in Fig.\ref{Fig:RBM_algo}(a), for the phase node, the biases are denoted by $\{c, e\} \in \mathbb{R^{\rm{2}}}$ where $c$ is the bias for the neuron capturing the real part of the phase and $e$ is the bias for the neuron encoding the corresponding imaginary part. The phase node shares interconnections with the visible node only and is defined by $\Vec{d} \in \mathbb{R^{\rm{n}}}$ for the real part of the phase and $\Vec{f} \in \mathbb{R^{\rm{n}}}$ for the associated imaginary part. The corresponding phase function for the quantum state $|\psi\rangle$ defined using these nodes is

\begin{align}
s(\Vec{d}, \Vec{f}, c, e, \vec{\sigma}) &= \tanh\left[(c + \sum_{i}d_i\sigma_i) + i(e + \sum_{i}f_i\sigma_i)\right]  \label{phase_encod}
\end{align}

Together the set $\Vec{X} = (\Vec{a}, \Vec{b}, \Vec{W}, \Vec{d}, \Vec{f}, c, e)$ thus defines the complete set of trainable parameters of the model which the network shall learn iteratively to mimic the coefficients of the quantum state $|\psi\rangle$ in the chosen basis. We shall discuss the algorithm to do this in the next section.

}

\subsection{Outline of the Method} \label{Algo_out}
The entire algorithm is schematically depicted in Fig. \ref{Fig:RBM_algo}(b). It goes as follows. 

\begin{enumerate}[(i)]
\item {\color{black}The}  first step is to initialize the parameters in the parameter vector $\Vec{X} = (\Vec{a}, \Vec{b}, \Vec{W}, \Vec{d}, \Vec{f}, c, e)$ on a classical computer. All parameters are randomly initialized in the parameter range [-0.02, 0.02] to avoid the vanishing gradient of the activation function for the phase node \cite{Xia2018}. Sometimes if random initialization returns a poorly converged result, we use the initial parameter set of a converged point in a similar problem as the starting guess, a process known as warm optimization. 

\item {\color{black} In} the second step the set $(\Vec{a}, \Vec{b}, \Vec{W})$ is fed into a quantum circuit for Gibbs sampling shown in Fig. \ref{Fig:RBM_circ}. {\color{black} This step is performed on a quantum computer}. The circuit requires $n+m$ qubits to encode the visible node and the hidden node respectively and additionally $m\times n$ ancillary qubits. The entire register is initialized to $|0\rangle$. The purpose of the circuit is to sample a bit string $(\vec{\sigma}, \vec{h}) \in \{1,-1\}^{m+n}$ from the RBM distribution $P(\Vec{a}, \Vec{b}, \Vec{W}, \vec{\sigma},\vec{h})$ defined in before in Eq. \ref{rbm_dist} \cite{Sureshbabu2021}.
In reality the circuit actually draws a sample $(\vec{\sigma}, \vec{h})$ from 
\begin{align}
     Q(\Vec{a}, \Vec{b}, \Vec{W}, \vec{\sigma},\vec{h}) &= \frac{e^{\frac{1}{k}(\sum_{i}a_i\sigma_i + \sum_{j}b_j h_j + \sum_{ij}w_{ij}\sigma_i h_j)}}{\sum_{\{\sigma h\}}e^{\frac{1}{k}(\sum_{i}a_i\sigma_i + \sum_{j}b_j h_j + \sum_{ij}w_{ij}\sigma_i h_j)}}
 \end{align}
and then reconstruct $P(\Vec{a}, \Vec{b}, \Vec{W}, \vec{\sigma},\vec{h}) \propto Q(\Vec{a}, \Vec{b}, \Vec{W}, \vec{\sigma},\vec{h})^k$. The real-valued parameter $k$ will be discussed shortly.
\begin{figure}[!htb]
    \centering
    \includegraphics[width=0.70\textwidth]{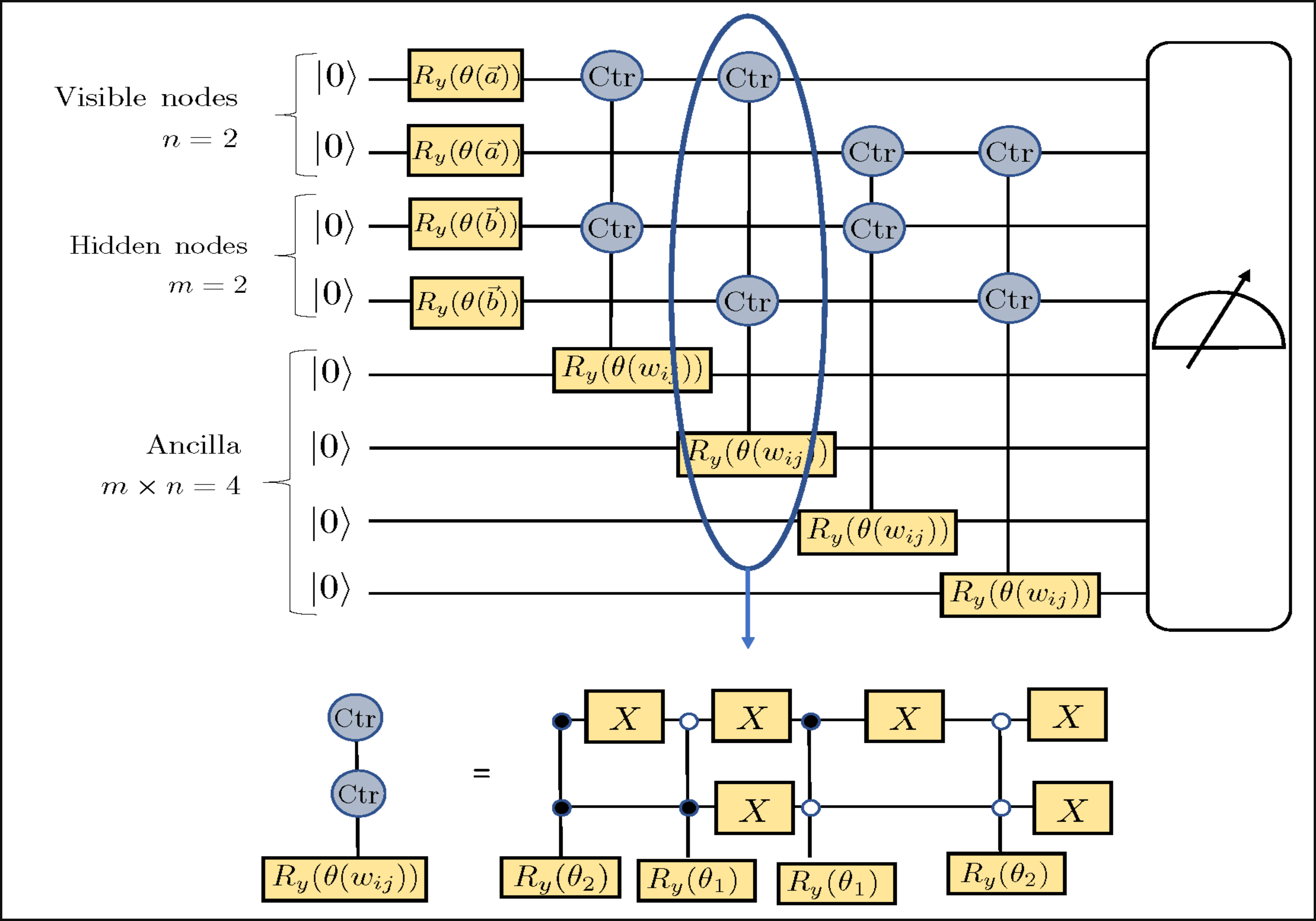}
    \caption{{\color{black} The Gibbs sampling quantum circuit used to create the Boltzmann distribution in Fig. \ref{Fig:RBM_algo}(b) (highlighted within the red box in Fig. \ref{Fig:RBM_algo}(b) step(ii)) for the case of $n=m=2$. The circuit contains single-qubit $R_y$ gates parameterized by biases ($\vec{a}$, $\vec{b}$) of hidden and visible neurons and $C-C-R_y$ gates parameterized by weights $\vec{W}$ between the hidden and visible neurons. Each $C-C-Ry$ gate is conditioned to rotate by different angles $\theta_1$ and $\theta_2$ for different choices of configurations of the control qubits. This can be implemented by use of $X$ gates as illustrated at the bottom. The open circles show a node in state $|0\rangle$ and the closed circles show a node in state $|1\rangle$. At the end of the circuit all qubits are measured and configurations wherein the ancilla qubits are all in state $|1\rangle$ are post-selected (see text for details). For $(n+m)$ visible and hidden neurons, there will be $(n+m)$ visible and hidden qubits and also $(n+m)$ single $R_y$ gates as there are that many biases. However since the $C-C-R_y$ gates are always controlled by 1 visible and 1 hidden qubit, there will be $m\times n$ such possibilities each of which targets one ancilla thereby making the size of the ancilla register $m \times n$. Thus there will be $O(n\times m)$ gates and number of qubits in the circuit. We discuss this further in section \ref{res_req}.}}
    \label{Fig:RBM_circ}
\end{figure}


The state of the visible node qubits and hidden node qubits are denoted henceforth as $|\sigma_i\rangle$ and $|h_j\rangle$ respectively. Note when $\sigma_i\:\: (\rm{or}\:\: h_j) = -1$,  $|\sigma_i\rangle\:\: (\rm{or}\:\: |h_j\rangle) = |0\rangle$ and $|1\rangle$ otherwise. In the circuit shown in Fig. \ref{Fig:RBM_circ} the single-qubit $R_y$ gates acting only on the visible and hidden units have rotation angles parameterized by $(\Vec{a}, \Vec{b})$ and are responsible for creating the non-interacting part of the distribution in $Q(\Vec{a}, \Vec{b}, \Vec{W},\vec{\sigma},\vec{h})$ while the interaction terms $\{\sum_{i,j} w_{ij}\sigma_ih_j\}$ are turned on through using $C-C-R_y$ gates acting on ancilla register as the target. The rotation angles of these doubly-controlled $R_y$ gates are parameterized by $\Vec{W}$ and are different for different configurations of the control qubits (always 1 hidden and 1 visible). Various such choices can be realized by using $X$ gates as shown in Fig. \ref{Fig:RBM_circ}. After all such operations, we measure all the $(m+n+m\times n)$ qubits and post-select the results wherein the ancilla qubits have collapsed to state $|1111...1_{mn}\rangle$ only. We show that the probability of such a successful event {\color{black}has} a generic lower bound determinable in terms of the parameters of the network  $(\Vec{a}, \Vec{b}, \Vec{W})$ (for details of the derivation of the generic bound refer to Section 2 of Supplementary Information). This master lower bound generalizes the previously noted one\cite{Xia2018} as a special case. The role of the real-valued parameter $k$ kicks in here. It serves as a regulator and is chosen in simulation to make the aforesaid lower bound a constant value (see Section 2 in Supplementary Information). After the post-selection, the corresponding states of the visible and hidden units are equivalent to all possible bit strings sampled from the distribution $Q(\Vec{a}, \Vec{b}, \Vec{W}, \vec{\sigma},\vec{h})$ from which the desired distribution $P(\Vec{a}, \Vec{b}, \Vec{W}, \vec{\sigma},\vec{h})$ is constructed.
{\color{black} The primary quantum advantage in our algorithm comes at this step where the full RBM distribution is constructed. Indeed we shall elaborate in Section \ref{res_req}, that there exist no polynomial-time classical algorithms for the construction of full RBM distribution. In our case, we can access the full distribution using quadratic resources by leveraging a quantum computer. The physical reason for this advantage is rooted in quantum parallelism which before a projective measurement allows the general state of the $(m+n+m\times n)$ qubits to be a superposition of all possible bit-strings with the coefficients sampled from the full RBM distribution. Many such measurements are necessary to construct the RBM distribution encoding the target state as post-measurement we can retrieve only one such bit-string. As explained above, the $k$ parameter in our model is useful here as it can be adaptively chosen by the user to control the measurement statistics (see Section 2 in Supplementary Information). Besides, for all systems primarily treated in this manuscript, we shall show that the chances of the ancilla register collapsing in the favorable state are naturally high even for modest values of the $k$ parameter. (see Section 5 in Supplementary Information) } 
With $P(\Vec{a}, \Vec{b}, \Vec{W}, \vec{\sigma},\vec{h})$ constructed, one can now compute the marginal distribution over the state space of the visible units only as $\Tilde{p}(\Vec{a}, \Vec{b}, \Vec{W}, \vec{\sigma})$ where $\Tilde{p}(\Vec{a}, \Vec{b}, \Vec{W}, \vec{\sigma}) = \sum_{h} P(\Vec{a}, \Vec{b}, \Vec{W}, \vec{\sigma},\vec{h})$. Now $\sqrt{\Tilde{p}(\Vec{a}, \Vec{b}, \Vec{W}, \vec{\sigma})}$ defines the amplitude of wavefunction over basis states of the visible units i.e. $|\sigma_1\sigma_2....\sigma_n\rangle$. The phase of each component of the wavefunction is now constructed classically using $(\Vec{d}, \Vec{f}, c, e)$  and $tanh$ activation of neurons in the phase node as defined before in Eq. \ref{phase_encod}

\item  {\color{black} With} the two information from step (ii), the target wavefunction can now be constructed classically as 
\begin{align}
    \psi(\Vec{X}) = \sum_{\sigma}  \sqrt{\Tilde{p}(\Vec{a}, \Vec{b}, \Vec{W}, \vec{\sigma})}s(\Vec{d}, \Vec{f}, c, e, \vec{\sigma})|\sigma_1\sigma_2...\sigma_n\rangle  \label{wavefn}
\end{align}

\item  {\color{black} With} the wavefunction, the cost function in Eq.
\ref{cost_fn} can now be constructed classically with the $(\hat{H}, \hat{O}, \lambda)$ from the user where $\hat{H}$ and $\hat{O}$ are the Hamiltonian and filter operator for the system being investigated respectively and $\lambda$ is the penalty parameter.

\item  {\color{black} The} next step is to check for convergence criterion or maximum number of iterations (to be discussed later). If either of the criterion is satisfied, results are printed

\item{\color{black} If} either of the criterion from the previous step is not satisfied then the parameter set $\Vec{X}$ is updated using steepest - descent algorithm with a learning rate (set to 0.005 in all our calculations). The updated parameter vector $\Vec{X}$ is fed into step (ii) for the next iteration of the algorithm. We have also used {\color{black} the} ADAM {\color{black}optimizer}\cite{Kingma2015} but there is no significant change in convergence for the systems treated in this report.
{\color{black} It must be emphasized that unlike in classical supervised deep learning models, the learning of our network does not require prior training against a pre-assigned labeled data-set. The network learns the target eigenstate directly through minimization of the cost function (see Eq. \ref{cost_fn}) using the optimizer of choice (gradient descent in this case).}
\end{enumerate}

\subsection{Resource Requirements}\label{res_req}

{\color{black}The power of an RBM ansatz even though underutilized in material science is beginning to gain attention in many areas of fermionic and bosonic physics \cite{Melko2019a, PhysRevB.97.035116, Batista_RBM}}. Using $n$ visible neurons and $m$ hidden neurons, a recent study \cite{PhysRevB.101.195141} {\color{black}has} shown explicitly how a shallow RBM ansatz ($\alpha = \frac{m}{n} =1$) like ours already captures several orders of perturbation theory and is a good approximant to the exact state. Classically, constructing such a full RBM distribution will require tracking  amplitudes from a $2^{m+n}$ dimensional state space and hence has exponential resource requirements in preparation. Ref\cite{Long2010} formalizes and consolidates this statement by proving that a polynomial-time algorithm for classically simulating or constructing a full RBM distribution is not only absent now but is unlikely to exist even in the future as long as the polynomial hierarchy remains uncollapsed. However, such analysis does not preclude the existence of efficient quantum algorithms such as the one considered in this work. The quantum circuit in our algorithm (see Fig. \ref{Fig:RBM_circ}) {\color{black} uses} $m+n+m \times n$ qubits only for constructing the state indicating an $O(m\times n)$ scaling in qubit resource {\color{black} which if expressed in terms of hidden node density $\alpha = \frac{m}{n}$ is $O(\alpha n^2)$}  The gate-set comprising single-qubit $R_y$ gates scales as $m+n$ too, {\color{black} one for each of the bias terms $(\vec{a}, \vec{b})$} of the visible and hidden node qubits. {\color{black} Each $C-C-R_y$ gate in the circuit mediate a single interaction term within the $\Vec{W}$ matrix between a spin of the visible layer $\sigma_i$ and a spin of the hidden layer $h_j$. Since there are $m\times n$ such terms, the number of $C-C-R_y$ gates are $m\times n$ too, with the targets being each qubit in the ancilla register}. Toggling between the various configurations of the control qubits (1 visible + 1 hidden) would require 6 $R_x$ gates additionally in each $C-C-R_y$ (see Fig. \ref{Fig:RBM_circ}) and hence the total number of such $R_x$ gates is $6mn$. This indicates the total gate requirements of our sampling circuit is also $O(m\times n)$ {\color{black} which is equivalent to $O(\alpha n^2)$}. {\color{black} The number of variational parameters in our algorithm for amplitude encoding using RBM is $m+n$ for the biases of the two nodes and $m \times n$ for the $\Vec{W}$ matrix. For the phase encoding, the variables are two $n$-dimensional vectors $(\Vec{d}, \Vec{f})$ and two scalars $(c,e)$. Thus the total number of variational parameters is $m \times n+m+3n+2 = \alpha n^2+\alpha n+3n+2$ which is also quadratic. The upshot is then, our algorithm for an RBM ansatz uses $O(\alpha n^2)$ qubits (circuit width), gate-set (circuit depth), and variational parameters to encode any arbitrary quantum state of $n$ qubits 
in a $d = 2^n$-dimensional Hilbert space. Removing redundancy in global phase and normalization, a general such state would require $2(2^{n}-1)$ parameters. One must know, in the RBM construction circuit no specific structure or sparsity has been assumed in the $\Vec{W}$ matrix which if present may lower the requirements further}. Quantum advantages have also been observed in supervised learning using {\color{black} the} RBM distribution \cite{wiebe2014quantum}. The study indicated that for the data-set of size $N$, a quantum circuit with amplitude amplification reduces the complexity of the algorithm from the conventional $O(N)$ to $O(\sqrt{N})$, a quadratic {\color{black}boost}. {\color{black} It must also be emphasized that all the results in this manuscript are primarily treated for the case of $\alpha=1$ as that suffices for the description of the system we study. We show how the results change for changing hidden node density $\alpha$ in Section 6 of the Supplementary Information. Even though $\alpha=1$ is good for systems in this report, for the case where the state is highly entangled, the user may be required to enhance the hidden node density as that increases the number of variational parameters and make the ansatz more expressive \cite{Melko2019a}. That may also be the case for molecular systems under geometric distortion wherein multi-reference correlation is important (we explore this point briefly in Section 10 of Supplementary Information). In this work all our results are compared against exact diagonalization as it affords the best accuracy in a given basis. The exact diagonalization results are obtained using `Numpy' package\cite{Harris2020} in python 3.0 with LAPACK routine.} 


\subsection{Implementation Methods}
We implement the algorithm in three flavors of computation. The first flavor henceforth designated as \textbf{`RBM-cl'} involves implementing the entire gate set of the Gibbs sampling circuit on a classical computer. This computation returns to us the exact state after the termination of the circuit. The second flavor is henceforth designated as \textbf{`RBM-qasm'}. This has been implemented by simulating the Gibbs sampling circuit using Qiskit which stands for IBM's Quantum Information Software Kit (Qiskit)~\cite{aleksandrowicz2019qiskit}. We specifically used the {\it qasm\_simulator} at Aer provider (hence the name \textbf{RBM-qasm}) which is a quantum computer simulator  
and hence can mimic calculations performed on a noisy-intermediate scale quantum computing device even using a classical computer with options to incorporate customizable noise models. Unlike in \textbf{`RBM-cl'} where the exact state is returned, in \textbf{`RBM-qasm'}, the Gibbs sampling circuit in Fig. \ref{Fig:RBM_circ} is interrogated multiple times to build measurement statistics. From the observed bit-strings, the measurement probabilities $P(\vec{a}, \vec{b}, \vec{W} ,\vec{\sigma},\vec{h})$ {\color{black} are} computed and hence the results are subjected to statistical fluctuations due to finite sampling errors. No noise model was used during the simulation in \textbf{`RBM-qasm'}. Finally to see the effect of noise we also investigated the performance of our algorithm on real IBM-Q quantum computers using the Qiskit interface. We used IBM-Q Sydney\cite{IBM_Sydney} and IBM-Q Toronto\cite{IBM_Toronto} interchangeably both of which are 27 qubit machines and hence suitable for our case studies. Calculations of this flavor are henceforth referred to as \textbf{`RBM-IBMQ'}. To reduce the effect of noise on the sampling probabilities we {\color{black} employ Measurement Error Mitigation (MEM) \cite{Barron2020}} directly implementable on Qiskit. {\color{black} We show in this report that MEM alone guarantees smooth and clear self-convergence in training (see Section 4 in Supplementary Information). The final accuracy of the results is affected by both MEM and warm-starting. We have seen without warm-starting convergence can not only be slow but sometimes the network can even be trapped in a local minima. It is in general difficult to assess apriori when the need for warm-starting can arise without a knowledge of the optimization surface as the objective function being optimized for the amplitude and the phase are non-convex in the arguments (see Eq.\ref{rbm_dist} and Eq.\ref{phase_encod}). It has been noted that the algorithm converges better without the need for warm-starting near optima (symmetry points for the system being treated in this report as discussed later).} For the \textbf{`RBM-qasm'} and \textbf{`RBM-cl'} simulations, the maximum number of iterations within which well-converged results to be discussed below were obtained is $\le$ 30,000 either with a warm-start or randomly initialized parameter set depending on the case. {\color{black} The \textbf{`RBM-IBMQ'} simulations were performed by breaking into two sessions/runs with the maximum iteration $\le$ 700 for each session to reduce the job queue . Normally most calculations converged well before 700 iterations were reached within the first run as warm-starting and MEM has been used as described above. For the few that did not, the final parameter set of the first run is punched for initializing the second session to ensure one continuous run.}
{\color{black} It must be emphasized that the entire code-base for training the network is home-built in Python 3.0 using standard packages like Numpy \cite{Harris2020}. As mentioned before, we have extensively used Qiskit though as an interface to communicate with the IBMQ hardware and with {\it qasm\_simulator.}}


\section{Results and Discussion}
\label{Sec:results}

As a test of our method, we target state filtration of energy eigenstates of two well-established transition-metal dichalcogenides (TMDCs) - monolayer Molybdenum di-Sulfide ($\rm{MoS}_2$) and monolayer Tungsten di-Sulfide ($\rm{WS}_2$). Monolayer TMDCs have so far eluded attention in quantum simulations even though it is imperative to study their electronic structures to understand novel properties\cite{nano8070463, Manzeli2017} like high carrier mobility, high photoluminescence due to the direct band-gap, lack of inversion symmetry leading to large spin-orbit coupling and intra-valley transport etc. Indeed such features have made them  attractive candidates 
for applications in Field-Effect Transistors\cite{doi:10.1063/1.4789365}, supercapacitors\cite{choudhary_patel_park_sirota_choi_2016}, spintronics\cite{Ahn2020}, opto-electronics\cite{PhysRevLett.105.136805, https://doi.org/10.1002/advs.201700231}, valleytronics\cite{Mak2018}. We first show how the entire conduction band (CB) in such materials can be simulated using an appropriate choice of operator $\hat{O}$ as the ground state projector as discussed before and then later show how to `sieve' eigenstates based on angular momentum symmetry. In all cases, we implement our algorithm on three flavors of RBM calculations - RBM-cl, RBM-qasm, RBM-IBMQ as discussed.

\subsection{Filter for target excited states - Simulation of low energy bands in $\rm{MoS}_2$ and $\rm{WS}_2$ and effect of Spin-Orbit Coupling}\label{exc_state}

\begin{figure}[!htb]
    \centering
    \includegraphics[width=0.90\textwidth]{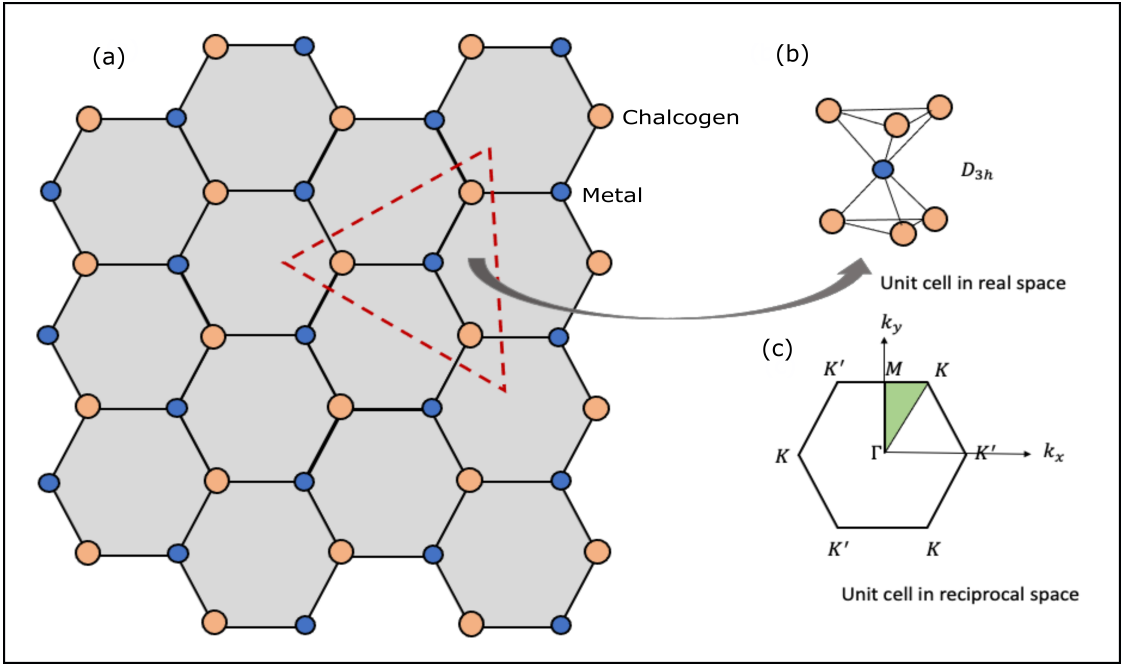}
    \caption{{\color{black} (a) The top view of the TMDC monolayer as studied in this report. The orange atoms are a chalcogen whereas the blue atoms are the metal centre. (b) The real-space trigonal prismatic unit cell highlighting $D_{3h}$ symmetry. This shows that in the TMDC monolayer unlike in graphene, the constituent atoms have a non-coplanar arrangement. (c) The unit cell in reciprocal space showing the important symmetry points $(\Gamma, K, M, K^\prime)$. We shall investigate the energy and other properties within the sector marked in green following the usual $\Gamma-K-M-\Gamma$ path as in \cite{Liu2013}. The co-ordinates of the symmetry points as $(k_x, k_y)$ are : $\Gamma = (0,0), K =(\frac{4\pi}{3a_0}, 0), M = (\frac{\pi}{a_0}, \frac{\pi}{\sqrt{3}a_0})$ where $a_0$ is the metal-chalcogen bond length. For systems studied in this report the metal centre is $\rm{Mo}, \rm{W}$ and the chalcogen is $\rm{S}$}}
    \label{Fig:TMDC_pic}
\end{figure}

The geometrical structure of monolayer TMDCs like $\rm{MoS}_2$ or $\rm{WS}_2$ indicates the presence of a trigonal prismatic real space unit cell \cite{Manzeli2017} with $D_{3h}$ point group symmetry {\color{black} as shown in Fig. \ref{Fig:TMDC_pic}}. The transition metal is at the centre and the sulfur atoms are at the six corners of the triangular prism ({\color{black}see Fig.
\ref{Fig:TMDC_pic}(b)}).  Consequently, the best orbital decomposition to evaluate the band structure of this periodic material should involve not only the $s, p, d$-orbitals of the central metal atom but also of the surrounding sulfur atoms. Indeed several reports exist which treats the electronic structure of such materials using a tight-binding description obtainable from a 5 band, 7-band or an 11-band model using varying degree of inclusion of the orbital set of the metal and the chalcogen\cite{Ridolfi_2015,doi:10.1063/1.4804936, SHAHRIARI2018169,PhysRevB.92.205108}. However, recently a 3-band parameterization has been demonstrated to yield remarkable accuracy in energy over the entire Brillouin zone\cite{Liu2013}. A tight binding Hamiltonian in this description is obtained by fitting the energy curves against DFT calculations (with GGA and LDA functionals) employing  the $d_{z^2}$, $d_{xy}$ and $d_{x^2 -y^2}$ orbitals of the metal centre\cite{Liu2013} only. This choice is based on the fact that for trigonal prismatic coordination, the d-orbital set of the metal splits into three groups-$A_1^\prime$ containing $d_{z^2}$ orbital only, $E^\prime$ containing $d_{xy}, d_{x^2 -y^2}$ and $E^{\prime\prime}$ containing $d_{xz}, d_{yz}$ orbitals. However, reflection symmetry of $D_{3h}$ restricts inter-coupling between the orbitals of $E^{\prime\prime}$ set with the remaining two groups. Indeed $E^{\prime\prime}$ contributes exclusively to higher energy bands and has no role to play in the low-energy physics of the valence and conduction band which is considered in this work. The absence of chalcogen $p$-orbitals is definitely an approximation albeit a good one as seen from Ref \cite{Liu2013}. We shall return to this point shortly.

We use a tight-binding model comprising of third-nearest neighbor (TNN) metal-metal hopping\cite{Liu2013} of the aforesaid three band Hamiltonian for all our calculations henceforth. The parameters of the model are obtained from the more accurate GGA calculation set \cite{Liu2013}. Section 3 of Supplementary Information enlist details of the Hamiltonian and parameters for completeness and brevity. Our working Hamiltonian, for both the systems are thus a $3\times3$ Hermitian matrix. For {\color{black}qubitization} we convert it into a $4\times4$ Hermitian matrix by padding an additional $1\times1$ block with a diagonal entry chosen to be $\ge$ spectral range of $H_{3\times3}$ as that would keep the low-lying eigenvalue structure of the resultant matrix undisturbed for the training to successfully proceed. Thus for both the systems, our neural network comprises of  a visible node with 2 neurons to encode the state, two hidden neurons and additional 2 neurons for the phase node too. For the Gibbs sampling circuit in Fig.
\ref{Fig:RBM_circ}, we thus need 2 qubits to represent the entire visible layer and 2 qubits for the hidden layer. In addition we need 4 ancilla qubits to serve as targets for $(C-C-R_y)$ rotation {\color{black} thereby requiring 8 qubits in total}. For the circuit in Fig. \ref{Fig:RBM_circ}, we use 4 single qubit Rotation gates $(R_y)$, 4 Controlled-Controlled Rotation gates ($C-C-R_y$), and also 24 Bit-flip (X) gates.  The optimization in each case starts from a randomly initialized parameter set. In case if the accuracy is poor, we re-start the algorithm by feeding the initial parameter from the results of a nearby converged $k$-point as a warm start. We see the results are in excellent agreement with the exact diagonalization when a such a warm start is employed along with MEM as described before. For IBMQ implementation we have used `IBM-Sydney' and `IBM-Toronto' both of which are 27 qubit machines.
To reduce the operational time on the actual quantum device for job queue and isolate the effect of gate-infidelity, IBMQ simulations for each $k$-point were often warm-started with an initial parameter set obtained from the initial parameters of the {\it qasm} simulation of a nearby but non-identical $k$-point.

\begin{figure}[H]
    \centering
    \includegraphics[width=1.0\textwidth]{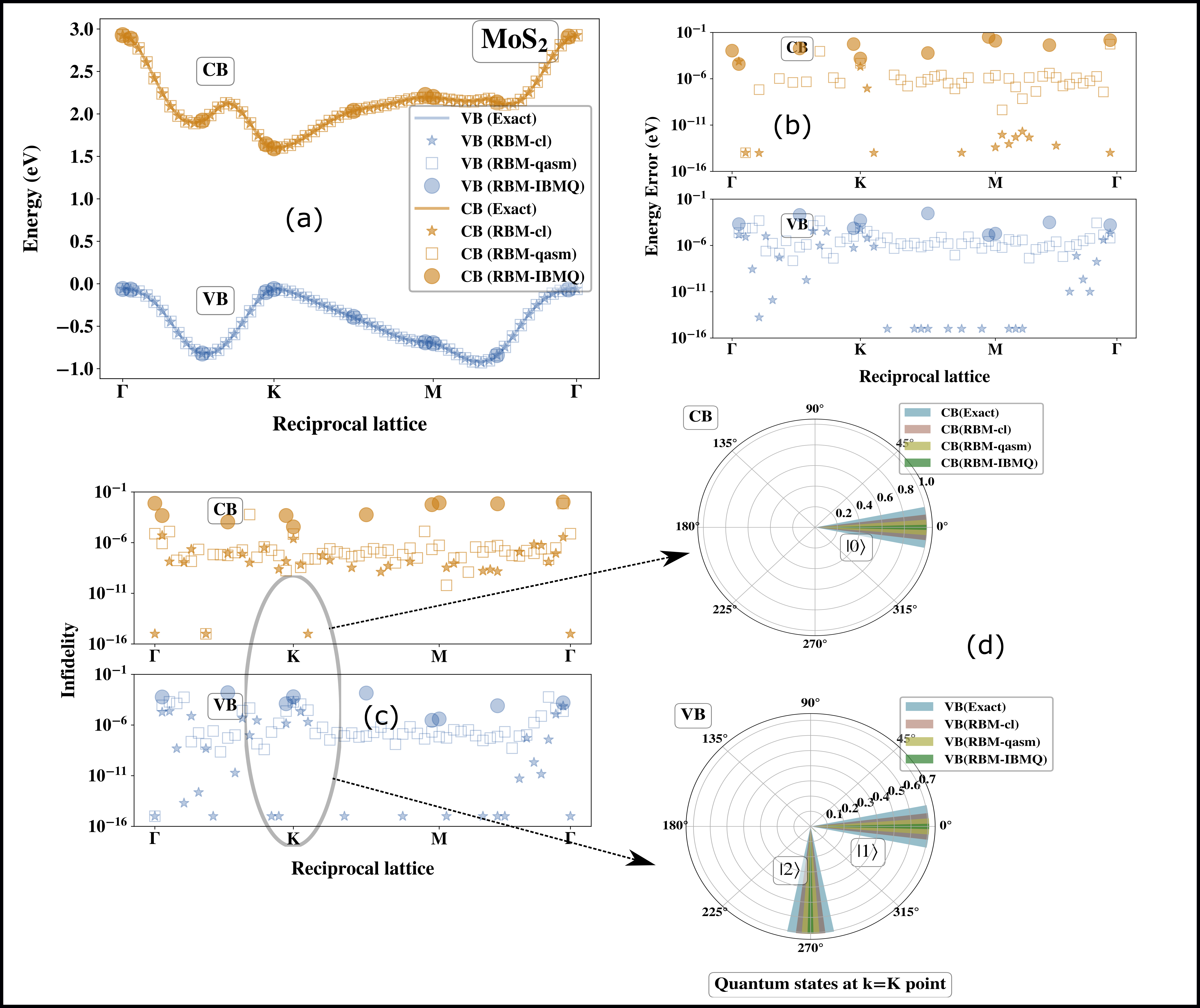}
    \caption{(a) Valence (VB) and conduction band (CB) of $\rm{MoS}_2$ calculated using all flavors of RBM and overlayed against exact diagonalization. The valence band is simulated using $\lambda=0$ in Eq. \ref{cost_fn} and the conduction band using ($O = |v_0\rangle \langle v_0|$, $\omega=0$, $\lambda=5$) in Eq. \ref{cost_fn} where $|v_0\rangle$ is the valence band state at each $k$-point. For IBMQ implementations we used `IBM-Sydney' and `IBM-Toronto'. All parameters are randomly initialized (see Fig. \ref{Fig:RBM_circ}) or warm-started with the initial guess of a converged nearby $k$-point.
    (b) The corresponding energy errors from (a) in eV. (c) The corresponding state infidelities (1-$Fid$) where $Fid=|\langle\Psi_{\rm{RBM}}|\Psi_{\rm{Exact}}\rangle|^2$ (d) The orbital decomposition of the states at $K$-point where $|0\rangle = d_{z^2}$, $|1\rangle = d_{xy}$, $|2\rangle = d_{x^2-y^2}$. The states from RBM calculations matches well with those from exact diagonalization in phase and amplitude. The width for each bar is set differently for visual clarity.}
\label{F:MoS2_exc_state}
\end{figure}

The results from the algorithm using the cost function in Eq. \ref{cost_fn} is displayed in Fig. \ref{F:MoS2_exc_state} for $\rm{MoS}_2$ and Fig. \ref{F:WS2_exc_state} for $\rm{WS}_2$. In Fig. \ref{F:MoS2_exc_state}(a) we have overlayed the energies obtained from our algorithm as a function of the wave-vector index sampled from the Brillouin zone following the usual $\Gamma-K-M-\Gamma$ path ({\color{black} see Fig. \ref{Fig:TMDC_pic}(c)}). The result for the valence band (VB) is denoted in blue and is obtained by setting $\lambda=0$ in Eq. \ref{cost_fn} which corresponds to the usual variational optimization to obtain the ground state at each k-point. The results for the conduction band (CB) are shown in orange in Fig. \ref{F:MoS2_exc_state}(a). They are thereafter computed as a separate set of calculations using $O = |v_0\rangle \langle v_0|$ and $\omega=0$ in the cost function in Eq. \ref{cost_fn} where the corresponding ground state in the valence band (VB) is denoted as $|v_0\rangle$. The penalty parameter is $\lambda=5$. The cost-function now samples a state orthogonal to ground state (null space of the projector $|v_0\rangle \langle v_0|$) for each of the k-points. The minimum energy criterion imposed by the first term in the cost function in Eq. \ref{cost_fn} guarantees obtaining the next higher excited state which happens to be the state space in the conduction band. 

We see for all flavors of our algorithm (RBM-cl, RBM-qasm, RBM-IBMQ) the simulated energy values for both the valence and the conduction band are in good agreement with the ones obtained from exact diagonalization. The corresponding errors in energy are displayed in Fig. \ref{F:MoS2_exc_state}(b) and are usually $\le$ $10^{-4}$ eV for RBM-cl and RBM-qasm which are noiseless pristine implementations but is around $10^{-2}-10^{-4}$ eV for the valence band (VB) and the conduction band for RBM-IBMQ indicating the worsening of performance due to faulty gate implementations in the Gibbs sampling circuit. Fig.
\ref{F:MoS2_exc_state}(c) plots the state infidelities i.e. 1-$Fid$ where $Fid=|\langle\Psi_{\rm{RBM}}|\Psi_{\rm{Exact}}\rangle|^2$.
We see that the infidelities are also quite small for each band with the performance worsened only in the IBMQ variant of the RBM implementation.

Like Fig.
\ref{F:MoS2_exc_state}(a), Fig.
\ref{F:WS2_exc_state}(a) displays the band structure of $\rm{WS}_2$ wherein the energies for both the valence and conduction band are overlayed against the energy values obtained from exact diagonalization.  All three flavors of RBM implementation yield reasonably accurate results as in the case for Fig.
\ref{F:MoS2_exc_state}(a). Fig.
\ref{F:WS2_exc_state}(b) and Fig.
\ref{F:WS2_exc_state}(c) display the energy error and the state infidelities of the state obtained from the RBM calculations against exact diagonalization. The error ranges in each case is similar to what has been discussed for $\rm{MoS}_2$.
\begin{figure}[H]
    \centering
    \includegraphics[width=1.0\textwidth]{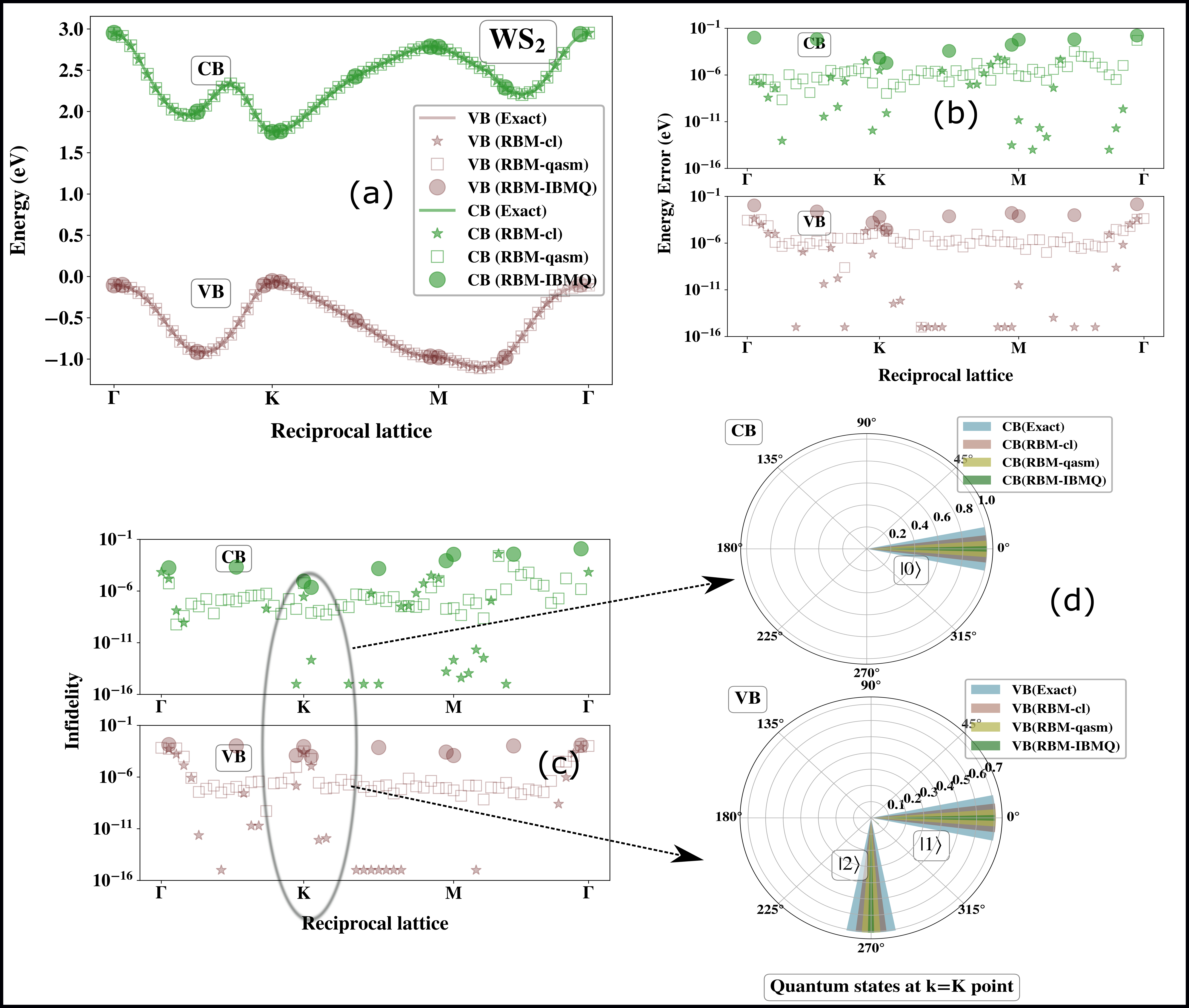}
    \caption{(a) Valence (VB) and conduction band (CB) of $\rm{WS}_2$ calculated using all flavors of RBM and overlayed against exact diagonalization. The valence band is simulated using $\lambda=0$ in Eq. \ref{cost_fn} and the conduction band using ($O = |v_0\rangle \langle v_0|$, $\omega=0$, $\lambda=5$) in Eq. \ref{cost_fn} where $|v_0\rangle$ is the valence band state at each $k$-point. For IBMQ implementations we used `IBM-Sydney' and `IBM-Toronto'. All parameters are randomly initialized (see Fig. \ref{Fig:RBM_circ}) or warm-started with the initial guess of a converged nearby $k$-point.
    (b) The corresponding energy errors from (a) in eV. (c) The corresponding state infidelities (1-$Fid$) where $Fid=|\langle\Psi_{\rm{RBM}}|\Psi_{\rm{Exact}}\rangle|^2$ (d) The orbital decomposition of the states at $K$-point where $|0\rangle = d_{z^2}$, $|1\rangle = d_{xy}$, $|2\rangle = d_{x^2-y^2}$. The states from RBM calculations matches well with those from exact diagonalization in phase and amplitude. The width for each bar is set differently for visual clarity. }
\label{F:WS2_exc_state}
\end{figure}

Fig. \ref{F:MoS2_exc_state}(d) and Fig. \ref{F:WS2_exc_state}(d) displays the orbital decomposition of the states in the conduction and valence band at the most important symmetry point i.e. the $K$- point. In our calculations qubit $|0\rangle = d_{z^2}$, $|1\rangle = d_{xy}$, $|2\rangle = d_{x^2-y^2}$ where $\{0,1,2\}$ are the integer equivalents of the two-qubit bit strings encoding the neurons of the visible node. We see from our calculations however that the exact state generated from the model lines up correctly against the RBM states in both amplitude and phase. While the state of the conduction band at $K$-point is exclusively populated by the $d_{z^2}$, that in the valence band is a superposition of $d_{xy}$ and $d_{x^2-y^2}$ with a phase shift of $3\pi/2$. This is consistent with the orbital decomposition given in Fig. 2 of Ref\cite{Liu2013} and is partly the reason given by the authors to use this three orbitals for generating the tight-binding Hamiltonian as the model yields correct state description near the band-gap. However as is clear from Fig. 2 of Ref\cite{Liu2013}, the orbital composition of the states at the $\Gamma$ and $M$-point has contribution from the p-orbitals of S and s-orbitals of both the metal and the S atoms. This makes the three-band model an approximation for the exact character of the states even though it can replicate the energy very well throughout the Brillouin zone. 

{\color{black} We further concentrate in this report on describing the low-energy physics near the $K$ or $K^\prime$ valley for which, as mentioned before, the state-description of the three-band model suffices. We construct the Hamiltonian \cite{PhysRevLett.108.196802, Kormanyos2013,sfluid_bose_gas, Ominato2020} near the $K$-valley in the basis of the states of the conduction band i.e. $|d_{z^2}\rangle$ (see Fig. \ref{F:MoS2_exc_state}(d) and Fig. \ref{F:WS2_exc_state}(d)) and that of the valence band i.e. $\frac{1}{\sqrt{2}}(|d_{x^2-y^2}\rangle +i|d_{xy}\rangle)$ (see Fig. \ref{F:MoS2_exc_state}(d) and Fig. \ref{F:WS2_exc_state}(d)). The states at the $K^\prime$ valley are related to those at the $K$ valley due to time-reversal symmetry\cite{VanDerDonck2017, sfluid_bose_gas} and hence is ignored from further discussion. The Hamiltonian is:
\begin{align}
    H &= (\frac{E_c}{2} + \frac{\lambda_c s}{2}) (\hat{I}-\hat{\sigma_z}) + (\frac{E_v}{2} + \frac{\lambda_v s}{2}) (\hat{I}+\hat{\sigma_z}) \nonumber \\
    &+ \gamma k_x \hat{\sigma_x} + \gamma k_y \hat{\sigma_y} \nonumber \\
    &+ \frac{\alpha}{2}(k_x^2 + k_y^2)(\hat{I} + \hat{\sigma_z}) +
    \frac{\beta}{2}(k_x^2 + k_y^2)(\hat{I} - \hat{\sigma_z}) \nonumber \\
    &+ \kappa_{TW} (k_x + ik_y)^2 (\hat{\sigma_x} + i\hat{\sigma_y}) + \kappa_{TW} (k_x - ik_y)^2 (\hat{\sigma_x} - i\hat{\sigma_y})
    \label{MD_Ham}
\end{align}
A effective description such as Eq. \ref{MD_Ham} is often referred in literature as the two-band $k\cdot p$ model constructed using Lowdin-Partitioning\cite{Liu2013, Kormanyos2013}.
The first two terms in Eq. \ref{MD_Ham} is the massive term required to create the band-gap($\Delta$) in the material at the $K$-point. These terms are absent in graphene. In most reports this term is written as $\frac{\Delta}{2}\sigma_z$ with a symmetrically located origin but we choose to use the $Ec$ and $Ev$ values obtained from our calculations in Fig. \ref{F:MoS2_exc_state} and Fig. \ref{F:WS2_exc_state}. The additional summands in each of the 1st two terms $(\lambda_v, \lambda_c)$ refer to band-splitting at thre always2 in numbere $K$-point due to spin-orbit coupling (SOC). In the three-band basis, SOC is entirely due to the $L_z$ operator (more on this in the next section) contribution of which in the chosen basis of can be effectively modeled as the first two terms \cite{PhysRevLett.108.196802, Ominato2020,Van_Donck_thesis}. Unlike the Bloch state in the conduction band, the valence band is exclusively dominated by metal orbitals $|d_{x^2-y^2}\rangle$ and $|d_{xy}\rangle$ with non-zero angular momentum leading to strong splitting \cite{Liu2013}. The spin-orbit splitting in the conduction band is weak\cite{PhysRevLett.108.196802, Liu2013,VanDerDonck2017} and below the resolvable limit of NISQ devices and hence has been ignored herein i.e. $\lambda_c =0$. The parameter $s \:\:\:\in \:\:\{1,-1\}$ is the spin index and labels the SOC split valence bands. The 3rd-6th term is the linear and quadratic extrapolation away from the $K$ point and yields a spherically isotropic band surface. The 7th-8th terms (parameterized by $\kappa_{TW}$) break the isotropy and lead to the well-known effect of trigonal warping (TW). The warped band surfaces in these materials are a consequence of the presence of a perpendicular $C_3$ axis due to the $D_{3h}$ symmetry of the associated real-space unit cells (see Fig. \ref{Fig:TMDC_pic}(b)). Further terms in \cite{Kormanyos2013} which removes anisotropy between valence and conduction band are ignored due to their small unresolvable contributions.

\begin{figure*}
    \centering
    \includegraphics[width=1.0\textwidth]{MoS2_SOC_full_panel.png}
    \caption{{\color{black}(a) The exact energy contours in valence band (VB) for s=1 within the three-band approximation for the Hamiltonian in Eq. \ref{MD_Ham} as a function of ($k_x$, $k_y$) near the $K$-point in $\rm{MoS}_2$ (b) Same as in a) but for s=-1 (c) Same as in a) for the conduction band (CB). The crosses in (a), (b) and (c) denotes the ($k_x$, $k_y$) pair wherein calculations for all three flavors of RBM have been executed. (d) Energy errors in eV from three flavors of RBM calculations for points denoted as cross in a) for the valence band (s=1) case computed using $\lambda=0$ in Eq. \ref{cost_fn} in $\rm{MoS}_2$. The x-axis is a flattened point index with $(k_x, k_y)$ pairs marked as crosses in (a) mapped to integers such that the origin is at the $K$-point. From the $K$-point, the flattened point index scale moves spirally outwards grouping all $(k_x, k_y)$ pairs satisfying $|k| = \sqrt{k_x^2 + k_y^2}$ as consecutive integers and then proceeding to the next $|k|$ (e) Same as in d) but with points denoted in b) as crosses for other valence band with s=-1 (f) Same as in d) but for points denoted in c) as crosses for the conduction band computed with $(\lambda=5, \omega=0, \hat{O}=|\nu_0\rangle \langle \nu_0|)$ in Eq. \ref{cost_fn}. (g) The amplitude for the occupancy of $d_{z^2}$ orbital on the metal for states computed at ($k_x$, $k_y$) pairs near the $K$-point from all three flavors of RBM as well as the exact states in valence band (s=1) for $\rm{MoS}_2$. The amplitude of states with the same $|k| = \sqrt{k_x^2 + k_y^2}$ appear bunched together as 'steps' due to flattened point-index scale used. Near the $K$-point the amplitude is the same for all such pairs within a given step due to isotropy of the energy surface. However away from the $K$-point deviations appear due to trigonal warping owing to the $D_{3h}$ symmetry of the unit cells in TMDCs. The states from all flavors of RBM can resolve the influence of warping accurately with the performance worsened for the noisy variant. (h) Same as in g) for valence band (s=-1) (i) Same as in g) for conduction band.}}
\label{Fig:MoS2_SOC_full}
\end{figure*}

Since the Hamiltonian in Eq. \ref{MD_Ham} is $2\times2$, we require a single visible neuron to encode the eigenstates, a single hidden neuron consistent with $\alpha=1$ and 1 additional ancillary qubit. The number of single-qubit $R_y$ gates is 2 and the number of $C-C-R_y$ gates is 1 and 6 $R_x$ gates. Calculations are performed using $\lambda=0$ in Eq. \ref{cost_fn} for the two SOC split valence bands with $s=\pm 1$ and $(\hat{O}=|\nu_0\rangle \langle \nu_0|$, $\lambda =5, \omega=0)$
for the conduction band. For NISQ devices we use `IBM-Sydney' and `IBM-Toronto' interchangeably as before. All calculations are performed for $(k_x, k_y)$ pairs centered at the $K$-point and with a cutoff $|k|$ of 0.1$K$ point to probe the low-energy regime. Since the $(k_x, k_y)$ pairs are near a symmetry point ($K$-point) warm starting was rarely observed to be required in RBM-cl and RBM-qasm but has been occasionally used in RBM-IBMQ for hastening convergence and reducing job queue. Each point on RBM-IBMQ are performed within a single run with Measurement Error Mitigation (MEM) as before for smooth self-convergence and consistency with other results.
Parameters for warping are obtained from \cite{Kormanyos2013}

In Fig. \ref{Fig:MoS2_SOC_full}(a), (b) and (c) we plot the exact 2D band surfaces obtained from Eq. \ref{MD_Ham} for the two SOC split valence bands (s=$\pm$ 1) and the conduction band. The crosses in each plot refer to the $(k_x, k_y)$ pairs wherein all flavors of RBM calculations have been performed. The results of such RBM calculations for each such pair are displayed as energy errors (eV) in \ref{Fig:MoS2_SOC_full} (d)-(f). The x-axis in each such plot is a flattened point index mapping $(k_x, k_y)$ pairs to integers by starting from pairs closest to the $K$-valley at the origin and proceeding spirally outwards. In other words, for a given $|k|$ the flattened point index groups all $(k_x, k_y)$ pairs satisfying $|k| = \sqrt{k_x^2 + k_y^2}$ as consecutive integers and then proceeds to the next $|k|$. We see that the energy error in each case is low for the RBM-cl and RBM-qasm variant ($\le$ $10^{-4}$ eV) for all three bands and $\le$ $10^{-2}$ eV for the IBMQ variant. Thus given the energy scale and extent of the splitting in the valence bands ($s=\pm 1$) in Fig. \ref{Fig:MoS2_SOC_full}(a)-(b) and the scale of the energy errors in Fig. \ref{Fig:MoS2_SOC_full}(d)-(e), it suffices to say that the performance of our algorithm is good enough to resolve band splitting due to features like spin-orbit coupling. 
To study the effect of warping parameters in Eq. \ref{MD_Ham} in the state, we plot in Fig. \ref{Fig:MoS2_SOC_full}(g)-(i) the amplitude of the corresponding states in the basis of $|d_{z^2}\rangle$ for the two SOC split valence bands ($s=\pm 1$) and the conduction band. The x-axis in each case is the flattened pair index as in Fig. \ref{Fig:MoS2_SOC_full}(d)-(f). 
At the $K$-point (origin), the conduction band is exclusively populated by $|d_{z^2}\rangle$ as discussed before but the reverse is true for the valence bands. In each of the plots Fig. \ref{Fig:MoS2_SOC_full} (g)-(i) all $(k_x, k_y)$ pairs which satisfy $|k| = \sqrt{k_x^2 + k_y^2}$ are bunched together as `steps' due to the flattened point index scale chosen. We see that near the $K$-point wherein the effect of warping is not prominent, all such points within a given `step' (same $|k|$) share the same amplitude. However away from the $K$-point deviation starts to become predominant. The amplitudes computed from the states of all three variants of RBM calculations line up well against the exact curve with the IBMQ variant showing some deviations albeit small considering the y-scale in these plots. Our algorithm thus can successfully resolve finer features like trigonal warping too in these Bloch states. A similar panel for $\rm{WS}_2$ is presented in Section 7 of Supplementary Information. Accurate computation of such Bloch states with these finer features preserved is necessary as momentum matrix elements between these states become important in simulating important properties of materials like optical conductivity\cite{Opt_cond,2nd_Harm_gen}, electrical and thermal conductivity\cite{Th_cond_ref} etc.}


\subsection{Filter for arbitrary states using symmetry operators}

In this section, we shall use the same set of TMDCs discussed above to explore how one can sieve arbitrary states based on symmetry constraints. 
To demonstrate the point we use orbital angular momentum symmetry. The $L_z$ operator in the three-band approximation commutes with the Hamiltonian \cite{Liu2013} in absence of spin-orbit coupling as has been considered in this work. The operators $L_x$, $L_y$ are essentially null matrices in the three-band basis of $\{d_{z^2}, d_{xy}, d_{x^2-y^2}\}$ as mentioned in \cite{Liu2013}. Hence $L^2$ enjoys exclusive contribution from $L_z$ and is a symmetry operator in the system. For computation, we use the Hamiltonian of the system at the $K$-point because the three-band approximation as discussed before is extremely accurate therein. 

The complete set of eigenvalues and eigenstates of $L_z$ and hence of $L^2$ operator is given in Section 8 of Supplementary Information. From the knowledge of the spectrum of $L^2$ operator we see that it has two distinct eigenvalues which are $\{0,4\}$ in atomic units. One of the eigenvectors of the doubly-degenerate eigenspace with eigenvalue 4 is the state in the valence band and the other is a higher energy excited state above the conduction band (not shown in Fig. \ref{F:MoS2_exc_state} or Fig. \ref{F:WS2_exc_state}). Both these states are exclusively made from the contribution of $\{d_{xy}, d_{x^2-y^2}\}$ as seen from the state decomposition in Section 8 of the Supplementary Information. The sector with eigenvalue 0 has single-fold degeneracy and is made from the excited state in the conduction band. As discussed before in  Fig. \ref{F:MoS2_exc_state}(d) and  Fig. \ref{F:WS2_exc_state}(d) (also in Section 8 of Supplementary Information) this state is exclusively made from the contribution of the $d_{z^2}$ which explains the absence of z-component angular momentum. We would thus expect that if we choose $\hat{O}=L^2$ and $\omega=\{0,4\}$ in Eq. \ref{cost_fn} for training the network, we should yield the excited state in the conduction band for $\omega=0$ and should yield the ground state in the valence band for $\omega=4$ as that is of lower energy (in compliance with the first term in Eq. \ref{cost_fn}) than the other degenerate eigenstate.

The qubit and gate resource requirements of this simulation are exactly the same as discussed in section \ref{exc_state} with 2 visible node neurons and 2 hidden node neurons for each of the two systems $\rm{MoS}_2$ and $\rm{WS}_2$. The Gibbs sampling circuit in Fig. \ref{Fig:RBM_circ} would {\color{black} need} a total of 8 qubits as before (2 for visible node + 2 for hidden node + 4 ancillary qubits). The gate requirements for the circuit to reproduce the amplitude are thus 4  single qubit Rotation gates $(R_y)$, 4 Controlled-Controlled Rotation gates ($C-C-R_y$) and also 24 Bit-flip (X) gates. We start the optimization with randomly initialized parameters.

In Fig. \ref{MoS2_L2} we display the results of our simulation. Like before, the results from all three flavors of RBM (marked as 2 = RBM-cl, 3 = RBM-qasm and 4 = RBM-IBMQ) are compared against the exact expected state (marked as 1 = Exact). In Fig. \ref{MoS2_L2}(a), the results of energy in eV from the three RBM simulations and the exact one are displayed for the eigenvalue sector $\omega=0$ a.u. This happens to be the conduction band (CB) energy in Fig. \ref{F:MoS2_exc_state}(a). We find an extremely good agreement for all flavors of RBM with the exact value. The corresponding energy error is displayed in Fig. \ref{MoS2_L2}(c) and is in the range of $10^{-5}-10^{-4}$ eV for RBM-cl and RBM-qasm but is within  $10^{-4}-10^{-3}$ eV for the RBM-IBMQ variant. Fig. \ref{MoS2_L2}(b) displays the constraint violation error i.e. how much the state encoded in the neural network after training has an $\langle L^2 \rangle$ equal to the target value of $\omega$ (in this case $\omega=0$ a.u.). We see that the violations are quite small for the noiseless implementations. Even for implementation on actual NISQ devices of IBM-Q, it is close to $10^{-3}$ a.u. 

\begin{figure}[H]
    \centering
    \includegraphics[width=1\textwidth]{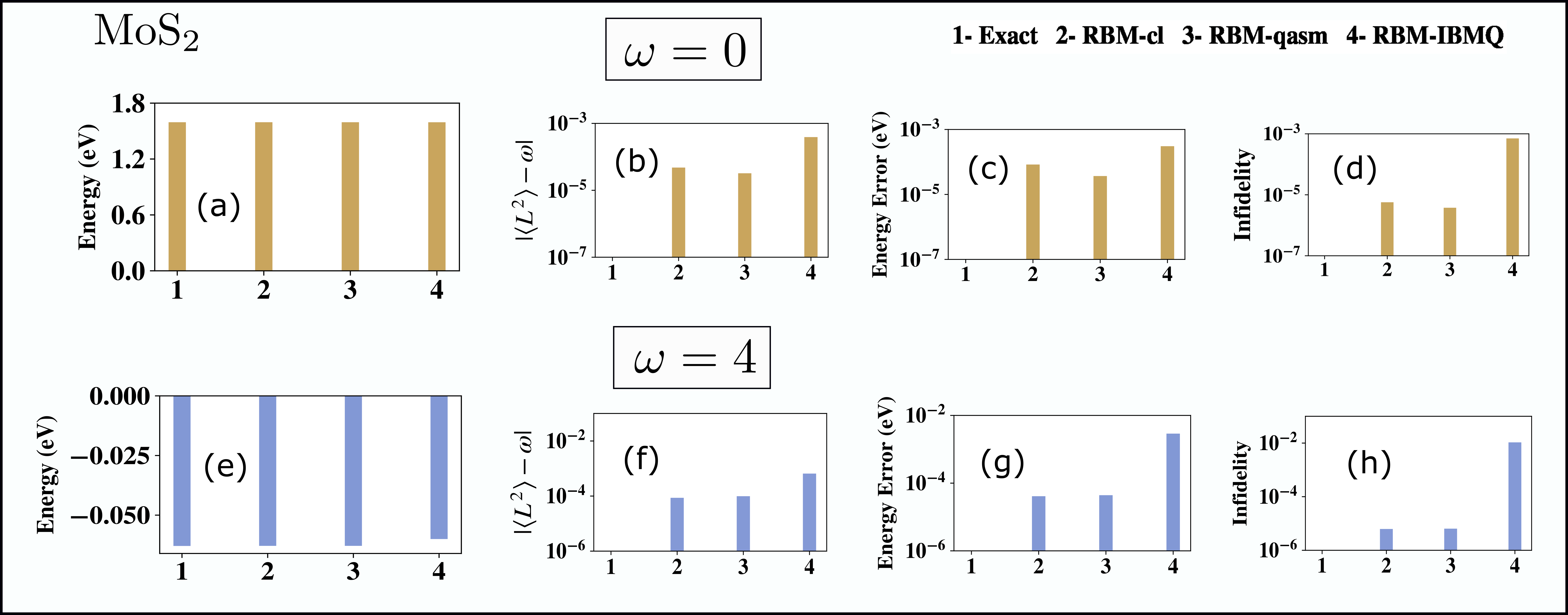}
    \caption{(a) The energy comparison between exact (1), RBM-cl (2), RBM-qasm (3), RBM-IBMQ (4) for computation with $\hat{O}=L^2$ and eigenvalue $\omega=0.0$ a.u. in Eq. \ref{cost_fn}. The exact energy is 1.5950 eV and is the conduction band energy at $K$-point in $\rm{MoS}_2$ shown in  Fig. \ref{F:MoS2_exc_state}. (b) The constraint violation error $|\langle L^2 \rangle - \omega|$ of the state obtained from different flavors of RBM and the desired value $\omega$. (c) The energy error in eV from (a) of the states obtained from RBM. (d) The state infidelities (1-$Fid$ where $Fid=|\langle\Psi_{\rm{RBM}}|\Psi_{\rm{Exact}}\rangle|^2$) obtained from RBM and the exact one (e-h) corresponds to an equivalent set of plots as in (a-d) just described but with the other eigenspace of $L^2$ with eigenvalue $\omega=4$ a.u. The exact energy here is the valence band energy at $K$-point for $\rm{MoS}_2$ shown in Fig. \ref{F:MoS2_exc_state} and is -0.0629 eV.} 
\label{MoS2_L2}
\end{figure}

\begin{figure}[H]
    \centering
    \includegraphics[width=1\textwidth]{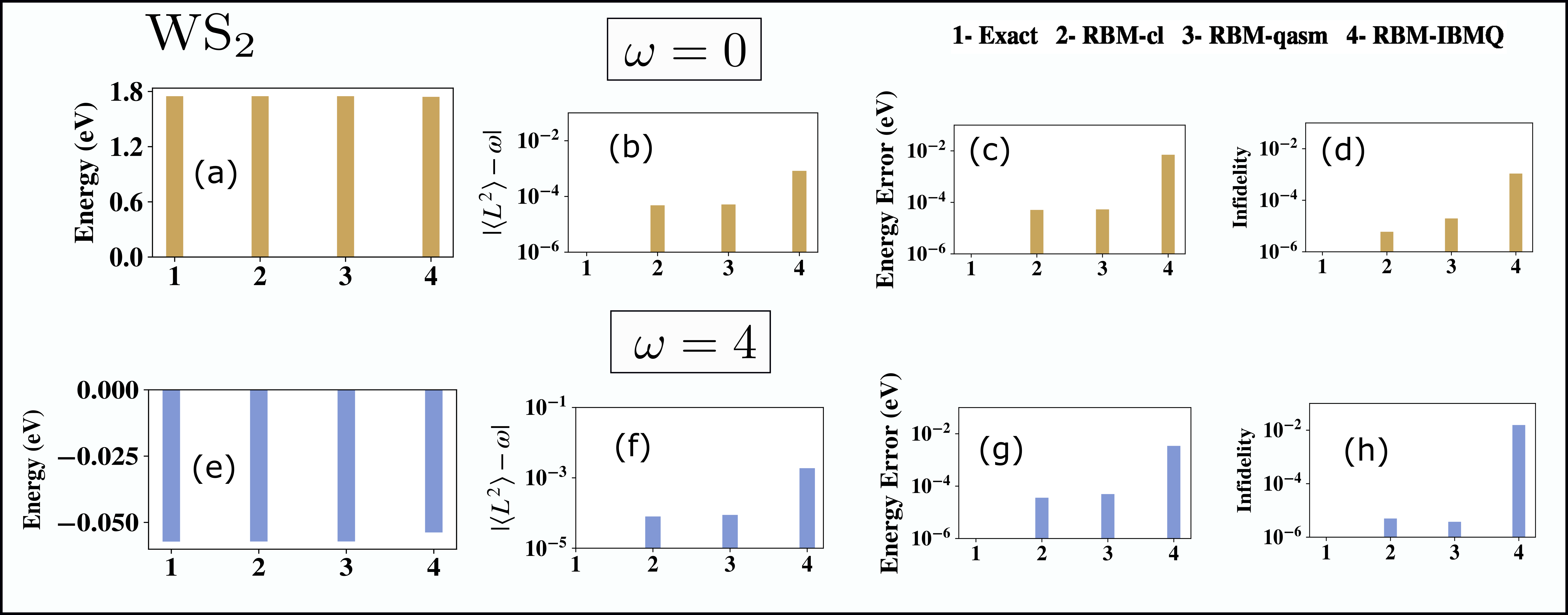}
    \caption{(a) The energy comparison between exact (1), RBM-cl (2), RBM-qasm (3), RBM-IBMQ (4) for computation with $\hat{O}=L^2$ and eigenvalue $\omega=0.0$ a.u. in Eq. \ref{cost_fn}. The exact energy is 1.749 eV and is the conduction band energy at $K$-point in $\rm{WS}_2$ shown in Fig. \ref{F:WS2_exc_state}. (b) The constraint violation error $|\langle L^2 \rangle - \omega|$ of the state obtained from different flavors of RBM and the desired value $\omega$. (c) The energy error in eV from (a) of the states obtained from RBM. (d) The state infidelities (1-$Fid$ where $Fid=|\langle\Psi_{\rm{RBM}}|\Psi_{\rm{Exact}}\rangle|^2$) obtained from RBM and the exact one (e-h) corresponds to an equivalent set of plots as in (a-d) just described but with the other eigenspace of $L^2$ with eigenvalue $\omega=4$ a.u.The exact energy here is the valence band energy at $K$-point for $\rm{WS}_2$ shown in Fig. \ref{F:WS2_exc_state} and is -0.0572 eV.}
\label{WS2_L2}
\end{figure}

Fig. \ref{MoS2_L2}(c) displays the energy error and Fig. \ref{MoS2_L2}(d) displays the state infidelity error (1-$Fid$ where $Fid=|\langle\Psi_{\rm{RBM}}|\Psi_{\rm{Exact}}\rangle|^2$). We see that for all flavors of RBM implementation the infidelities are quite small with the performance worsened for implementation on the actual IBM-Q device. Fig. \ref{MoS2_L2}(e-h) corresponds to similar plots as discussed above but this time in the other eigenvalue sector with $\omega=4$ a.u. We again see that the energy values (in eV) in Fig. \ref{MoS2_L2}(e) matches with the exact for all flavors of RBM-implementation. This state happens to be the ground state in the valence band (VB) shown in Fig. \ref{F:MoS2_exc_state}(a). The corresponding energy errors shown in Fig. \ref{MoS2_L2}(g) are like in the previous case ($\omega=0$) low for RBM-cl and RBM-qasm but in the range of $10^{-3}-10^{-2}$ eV for RBM-IBMQ. Similar analysis as in the case of $\omega=0$ a.u. can also be made for the constraint violation error in Fig. \ref{MoS2_L2}(f) and the state infidelity in Fig. \ref{MoS2_L2}(h). Both of these have low errors with the respective ranges as displayed.

Fig. \ref{WS2_L2} shows a similar plot for the other system studied
$\rm{WS}_2$. Just as before we display the results for $\omega=0$ a.u. in Fig. \ref{WS2_L2}(a-d) and for $\omega=4$ a.u. in Fig. \ref{WS2_L2}(e-h). Fig. \ref{WS2_L2}(a) shows the energy match between the RBM implementations and the exact value for $\omega=0$ a.u. and Fig. \ref{WS2_L2}(e) shows the same for $\omega=4$ a.u. The former is equal to the state in the conduction band at $K$-point (see Fig. \ref{F:WS2_exc_state}(a)) and the latter is the corresponding state in the valence band (see Fig. \ref{F:WS2_exc_state}(a)). We see good agreement for all RBM variants and the exact expected value. The corresponding energy errors are low (see Fig. \ref{WS2_L2}(c) and Fig. \ref{WS2_L2}(g)) with the range for IBMQ implementation being $10^{-3}-10^{-2}$ eV and even lesser for the pristine implementations. The respective constraint violation errors are displayed in Fig. \ref{WS2_L2}(b) and Fig. \ref{WS2_L2}(f) and are small too as seen from the scale. A similar statement can also be made for the state infidelity displayed in Fig. \ref{WS2_L2}(d) and Fig. \ref{WS2_L2}(h). We have seen that in both the systems $\rm{MoS}_2$ and $\rm{WS}_2$, the state infidelity and energy errors {\color{black}are} higher in the $\omega=4$ a.u. eigensector than in $\omega=0$ a.u. eigensector in the IBMQ implementation especially. In fact, the relative energy errors for the said sector are close to 5\% for RBM-IBMQ. However, the corresponding errors (both relative and absolute) {\color{black} are} low for the noiseless implementation (RBM-cl and RBM-qasm) indicating that the higher \% error is attributable to the imperfect implementation of gates in the Gibbs sampling circuit in an IBM-Q machine and hence can be mitigated with future quantum computing devices with better gate fidelities and error-correction schemes.

\section{Conclusion}
\label{Sec:conclusion}

In this study, we have demonstrated an algorithm which can filter arbitrary energy eigenstates in 2D materials like TMDCs using a quantum circuit with quadratic resources. We provided an original proof of feasibility for our cost function employed for the constrained optimization. We also proved a generic lower bound for the successful sampling of our quantum circuit from which previously known bounds can be extracted.
Our circuit trains a three-layered neural network that encodes the desired state using an RBM ansatz for the probability density. As an illustration, we were able to filter energy eigenstates in the conduction band of important TMDCs like $\rm{MoS}_2$ and $\rm{WS}_2$, and faithfully reproduce the band-gap. We were also able to filter arbitrary states based on a user-defined orbital angular momentum symmetry constraint. We trained the network on various flavors of computation using not only a classical computer, {\it qasm} backend quantum simulator in Qiskit but also a real IBMQ machine (IBM Sydney and IBM Toronto) with the objective to see the performance of the algorithm on actual NISQ devices. In all flavors of computation, our algorithm demonstrated very high accuracy when compared to the exact values obtained from direct diagonalization. 

Venturing beyond the ground state to obtain arbitrary states based on user-defined restrictions is the first of its kind in all flavors of QML. Furthermore, the systems of our choice happen to be TMDCs, an important class of 2D-periodic systems which have
never been studied using any quantum algorithm. Periodic systems in general have received scanty attention as far as quantum algorithms are concerned. Only two reports exist\cite{kanno2019manybody, Sureshbabu2021} both of which have simulated just the valence band in graphene and hexagonal Boron Nitride (h-BN). 

{\color{black} It must also be emphasized that a host of classical algorithms have been developed in traditional quantum chemistry that are extremely accurate and polynomially efficient. Over the past few decades, Density functional theory (DFT) has emerged into a leading candidate for accurate computation of wide-variety of electronic structure problems in molecules and materials\cite{doi:10.1063/1.4704546,doi:10.1146/annurev-matsci-070218-010143}. Variants of it are being developed for cases wherein multi-reference correlation would be important too\cite{MCSCF_DFT_2017}. Reduced density-matrix based methods are also polynomially scaling \cite{Mazziotti2011} and have shown excellent accuracy in strongly correlated systems\cite{Mazziotti2012, Montgomery2018}. Tensor-network based methods like Density-Matrix Renormalization Group (DMRG) \cite{PhysRevLett.69.2863, doi:10.1146/annurev-physchem-032210-103338,RevModPhys.77.259} have been developed which even though capable of exploiting rank-sparsity in strongly correlated one-dimensional systems yet loses the polynomial advantage in multi-dimensions. Like our algorithm which attempts to construct the many-body state, a plethora of similar wave-function based \textit{ab-initio} methods exist in traditional quantum chemistry too starting from the uncorrelated Hartree-Fock method to post-Hartree methods which can recover dynamic correlation like perturbative approaches (like MP2\cite{HEADGORDON1988503, MP2_CI_Rghavachari, doi:https://doi.org/10.1002/9781119019572.ch14}), Truncated Configuration-Interaction or CI (like CISD) \cite{MP2_CI_Rghavachari, doi:https://doi.org/10.1002/9781119019572.ch11}, Couple-Cluster (CC) methods \cite{RevModPhys.79.291} (like CCSD, CCSD(T), CCSDT or EOM-CCSD for excited states), recently developed SHCI methods\cite{doi:10.1063/1.5055390,doi:10.1063/1.4998614} to ones which are good for capturing static-correlation like Multi-Configurational Self-Consistent Field ( MCSCF\cite{MCSCF_rev}). A direct comparison of a quantum algorithm like ours with these classical algorithms can be attempted to be made in terms of accuracy and resource cost. In terms of resource requirements, the comparison is made  difficult by the fact that certain parameters like circuit-width, circuit-depth, etc which affect the performance of quantum algorithms like ours, have no classical analogues. If we consider an $N_{elec}$ electron system with $r=r_o + r_{uo}$ spin orbitals/fermionic modes such that $r_o=N_{elec}$ is the occupied orbital set in Hartree-Fock reference and $r_{uo}$ are virtual orbitals excluded from Hartree-Fock reference, then under the assumption that the orbital space rank loosely equates to qubits or number of visible neurons $n$ (see section \ref{The_Model} and \cite{JW_vs_BK}), we have shown in section \ref{res_req} that the spin-orbital cost of our algorithm would be $\approx O(r^2)=O(r_o^2 + r_{uo}^2 + 2r_{uo}r_{o})$. The numerical parameter count of our algorithm is also quadratic i.e. $O(\alpha n^2) \approx O(r^2)=O(r_o^2 + r_{uo}^2 + 2r_{uo}r_{o})$. This is unlike methods like CCSD (Coupled-Cluster Singles Doubles) which has a computational cost of $\approx O(r_o^2r_{uo}^4)$ (for CCSDT it is $\approx O(r_o^3r_{uo}^5)$ and for CCSD(T) it is $\approx O(r_o^3r_{uo}^4)$) \cite{RevModPhys.79.291, doi:https://doi.org/10.1002/9781119019572.ch13}. Also CCSD evaluates $\approx O(r_o^2r_{uo}^2)$ cluster amplitudes as parameters defining the excitations. For chemically important phenomenon like dissociation events which are no longer single-referenced are known to be difficult to treat with CCSD \cite{FAN20061}, even though pair cluster doubles can ameliorate the situation to some extent \cite{Scuseria_pCCD}. That being said, it must also be noted that traditional variant of CCSD unlike ours is non-variational. As far as accuracy is concerned, all results in this report are benchmarked against exponentially scaling exact diagonalization as that affords the exact value in a given basis. Not only the physics of Bloch states in the material TMDC but also a molecular example like $\rm{LiH}$ has been treated using our algorithm (see Section 10 in Supplementary Information). For both the ground and excited states of $\rm{LiH}$ we see good accuracy and improvement of error by enhancing the hidden node density which makes the ansatz more expressive. Studies on larger molecular systems for which the results of exact diagonalization may not be available may be undertaken in future. That will provide a platform for comparison in accuracy with a subset of the aforesaid classical algorithms. Desirable chemical features like size-consistency and size-extensivity may be probed too.}

{\color{black} One must also note that several quantum algorithms already exist which aim at obtaining ground and excited states of fermionic systems\cite{Cao2019}. Non-variational quantum algorithms like Quantum Phase estimation(QPE)\cite{Alan_2005, Lanyon2010,https://doi.org/10.1002/qua.25176} have exponential speed up\cite{Llyod_1999} yet require high circuit depth and long coherent operations which are beyond the limits of near-term hardwares \cite{preskill2018quantum}. Hybrid Variational Quantum Algorithms (VQA)  have also been developed which can ameliorate some of the above problems\cite{VQE_review}. The most notable one in the list is Unitary Coupled-Cluster Variational Quantum Eigensolver (UCC-VQE) \cite{Peruzzo2014}. In its most traditional variant, the unitary ansatz which UCC-VQE uses for state preparation consists of single and double-excitations\cite{theory_VQE,Romero2019, Sci_rep_uccsd}(hence often called the name Unitary Coupled Cluster Singles Doubles or UCCSD) from the reference state. However, the circuit depth in preparing such an ansatz is still large and the circuit is parameterized by many variables which necessitates a high-dimensional classical optimization routine \cite{VQE_lanl_rev} to update the parameters. To be concrete, for $r_o$ and $r_{uo}$ with the same meaning as described in previous paragraph, the UCCSD-VQE uses \cite{Cao2019} $O(r_o)$ qubits, $O(r_o^2r_{uo}^2)$ cluster amplitudes as parameters and $O(fr_o^4)$ gates where $f= O(r_o)\:\: \rm{or}\:\: O(log(r_o))$ depending on the qubit-mapping. Besides the UCC-VQE method can suffer from errors incurred due to operator ordering or Trotterization \cite{Trotter_UCC}. Also, the ansatz requires a high-degree of qubit connectivity for non-local operations which may not always be available in all hardware \cite{QAlgo_chem_rev}. Hardware-efficient ansatz \cite{Hardware_eff_Nat, PhysRevA.98.022322} has been developed to help solve the above issues which use an alternating framework of single-qubit gates and fixed entangling operations which can be chosen with the specific capabilities of the device in hand. However, unlike UCC-VQE, such an ansatz is not physically inspired and often suffers from trainability issues during parameter optimization \cite{VQE_review, BP_qnet}. Besides the number of parameters grow as a function of entangling blocks and can even surpass the size of the Hilbert space \cite{Cao2019, PhysRevA.98.022322}. A third variant that has low circuit dept
h and parameter cost is the ADAPT-VQE approach \cite{Grimsley2019}. Unlike in the previous two cases, this variant constructs the circuit from a pre-selected pool of operators and
changes the circuit architecture adaptively by adding operators from the pool which affects the energy gradient the most. The chosen pool decides the parameter count and gate-counts in the circuit. In this method, the number of measurement shots can be high for computing the gradients \cite{VQE_lanl_rev} and also it is generally not clear how to pre-select the operator pool and what guarantees that the pool is complete i.e. the ansatz it produces is expressive enough. Many different variants for each method have been constructed for which the reader is referred to many excellent reviews like \cite{VQE_review,QAlgo_chem_rev, Cao2019}. For excited states \cite{VQE_review}, deflation assisted VQE as described before \cite{Higgott2019} exist but for its implementation used the UCCSD ansatz which inherits some of the above problems of high parameter count and gates. A recent promising method known as Weighted Subspace- Search VQE uses an input array of several orthogonal states to construct a weighted Lagrangian as the cost function \cite{PhysRevResearch.1.033062}. In this case, the input states are mapped to the excited states of the system using a parameterized ansatz circuit. Depending on the nature of the ansatz circuit, the algorithm can have different gate-count or parameter count and hence it is hard to mention a general estimate.\\   
Our algorithm is also a hybrid variational algorithm like the ones in the aforesaid list but requiring quadratic resources always (see Section \ref{res_req}). However there are some key differences as well which need to be acknowledged. Unlike the above list of algorithms which prepare a unitary ansatz on a quantum computer to mimic the state, our algorithm proposes to construct a probability distribution that mimics the amplitude field of the target state on a quantum computer. As a result our algorithm is a distribution sampling protocol on a quantum computer using a non-unitary ansatz (RBM) which is manifested in the usage of ancilla and its subsequent measurement collapse. The measurement statistics of such a collapse are discussed in detail in Section 2 and Section 5 in Supplementary Information. Since the distribution encoding the amplitude field is based on RBM, unlike the ADAPT-VQE method, our protocol is largely \textit{problem agnostic}. This is due to the fact that RBM can act as a universal approximant to any probability density \cite{Roux_RBM} and hence can be used for a variety of problems provided it is made sufficiently expressive with an adequate hidden node density. Also unlike other algorithms wherein the nature of the excitations or operator pool used decides the cost-function gradient, in our case the distribution function being RBM always permits training the network with analytical gradients. Besides we have already demonstrated in section \ref{res_req} using Ref \cite{Long2010} that an analogous classical construction of RBM distribution has an exponential overhead whereas using a quantum algorithm like ours one can construct it using quadratic resources thereby illustrating the distinct quantum-classical advantage in our algorithm directly.}

Further extension of this algorithm can be made to compute operators using Hellmann-Feynmann method~\cite{Oh2009}, to characterize the influence of noise on the algorithm and to see it being extended to study  other interesting phenomena on 2D materials like Rashba splitting in polar TMDCs\cite{PhysRevB.95.165401} or even effect of strain\cite{Peng2020}. {\color{black} One must also note that in this work we construct the full $d=2^n$-dimensional eigenstate from the amplitude encoding using the RBM ansatz (Eq. \ref{rbm_dist}) and the phase encoding using the Eq. \ref{phase_encod}. This is because the primary quantum advantage of our algorithm lies in the fact we use quadratic resources to learn the full RBM distribution which classically would require exponential resources as necessitated in \cite{Long2010}. Besides access to the full state allows us to compute matrix elements of arbitrary operator between eigenstates important for spectral information i.e. learning in excitonic features\cite{Opt_cond} or thermal and electronic conductivity \cite{Th_cond_ref} which as said before are important future extensions of this work. Also once trained for a given system, the neural-network in our algorithm can be used to learn eigenstates of a closely related system accurately with faster convergence and lesser number of iterations indicating partial transferability of these models (see Section 9 of Supplementary Information for details). Benefits and scope of such `transferable training' for other chemically motivated systems will be investigated in future. It must be noted that the symmetry partitioning of the metal orbitals in TMDCs guaranteed in \cite{Liu2013} have reduced the effective size of the orbital space and qubit requirements in this study. However understanding spectral information in excitonic physics would require more involved models with a larger orbital space. A way forward may be focusing on low-energy excitons with a certain symmetry (like overall spin-angular momentum) characteristics only. For molecular systems such symmetry inspired cost-reductions are already beginning to be noticed \cite{Setia_Z2, PhysRevResearch.3.013039} as discussed earlier. However such an initiative for materials is largely an uncharted territory. Further reduction in qubit resource requirements of our algorithm may also help, even though the non-unitary nature of the ansatz as discussed before makes it harder. From the hardware point of view, robust large scale error mitigation strategies are beginning to be made available now \cite{Error_mit_google,GEM_qiskit} and devices with over 1000 qubits with low qubit decoherence errors and gate infidelities are also being promised in recent future \cite{IBM_roadmap}. Such resources would will certainly be beneficial to extensions of studies like these}. 

From the algorithmic point of view, besides being quadratic scaling in qubit and gate requirements and parameter count, our algorithm does not have any dependence on oracular objects like qRAM \cite{Ciliberto2017} which is responsible for creating a superposition of all possible basis states and is known to commonly sought in most quantum machine learning modules. As futuristic quantum devices are being developed with proper error mitigation schemes, we expect to have more such cross-pollination between machine learning algorithms and quantum computing with the promise to study electronic structure and dynamics in new complex materials. 


\section{Data and model availability}
The information about the input Hamiltonian and $L_z/L^2$ operator corresponding to $\rm{MoS}_2$ and $\rm{WS}_2$ can be found in the Supplementary Information in Section 3 and 8 respectively. Data will be made available upon reasonable request to the corresponding author. The codes associated with the classical simulation, simulation on the {\it qasm} backend, and the implementation on IBM's quantum computing devices is available with the corresponding author upon reasonable request. 

\section{Supporting Information}
Additional data can be found in Supplementary Information. Section 1 of Supplementary Information provides a proof of convergence of Theorem 2.1. in the manuscript. Section 2 shows the derivation of the generic lower bound for successful sampling. Section 3 describes the Hamiltonian matrix used for $\rm{MoS}_2$ and $\rm{WS}_2$. Section 4 describes the importance of Measurement Error Mitigation. Section 5 discusses the Measurement Statistics for the systems in the report. Section 6 shows how the results are affected by changing hidden node density. Section 7 shows the SOC splitting data for $\rm{WS}_2$. Section 8 discusses the $L^z/L^2$ operator corresponding to $\rm{MoS}_2$ and $\rm{WS}_2$. Section 9 discusses transferability of learning from trained network to other systems. Section 10 discusses the molecular example of $\rm{LiH}$.

\begin{acknowledgement}
We would like to thank Dr. Ruth Pachter, AFRL, for many useful discussions. AFRL support is acknowledged. We would also like to thank Sangchul Oh for his help during the preparation of the manuscript. We acknowledge funding by the U.S. Department of Energy (Office of Basic Energy Sciences) under Award No. DE-SC0019215 and the National Science Foundation under award number 1955907. This material is also based upon work supported by the U.S. Department of Energy, Office of Science, National Quantum Information Science Research Centers. We also acknowledge the use of IBM-Q and thank them for the support. The views expressed are those of the authors and do not reflect the official policy or position of IBM or the IBM Q team.

\end{acknowledgement}

\providecommand{\latin}[1]{#1}
\makeatletter
\providecommand{\doi}
  {\begingroup\let\do\@makeother\dospecials
  \catcode`\{=1 \catcode`\}=2 \doi@aux}
\providecommand{\doi@aux}[1]{\endgroup\texttt{#1}}
\makeatother
\providecommand*\mcitethebibliography{\thebibliography}
\csname @ifundefined\endcsname{endmcitethebibliography}
  {\let\endmcitethebibliography\endthebibliography}{}


\pagebreak
\begin{center}
\textbf{\Large Supplementary Information for Quantum Machine-Learning for Eigenstate Filtration in Two-Dimensional Materials}
\end{center}
 
\setcounter{equation}{0}
\setcounter{figure}{0}
\setcounter{table}{0}
\setcounter{page}{1}
\setcounter{section}{0}
\makeatletter
\renewcommand{\theequation}{S\arabic{equation}}
\renewcommand{\thefigure}{S\arabic{figure}}
\renewcommand{\bibnumfmt}[1]{[S#1]}
\renewcommand{\citenumfont}[1]{S#1} 
  
\section{Proof of Theorem 2.1}

The proof of feasibility of general penalty functions is known in optimization theory. Since both our objective function and penalty term are quadratic forms herein we construct an original, formal and a simple proof for Theorem 2.1.

Let us recollect the cost function $\rm {F(\lambda, \hat{H}, \hat{O}, |\psi\rangle)}$ defined in text:
\begin{align}
    \rm {F} (\lambda, \hat{H}, \hat{O}, |\psi\rangle) = \langle \psi| \hat {H} |\psi\rangle + \lambda \langle \psi|(\hat{O}-\omega)^2|\psi\rangle \label{cost_fn_supp}
\end{align}
\newpage
\subsection{Definitions} \label{defn}
We re-iterate the following definitions as considered in the main text
\begin{enumerate}
    \item $\hat{H} \in \mathbb{C}^{d\times d}$ and $\hat{H}=\hat{H}^\dagger$. This is the Hamiltonian operator in the problem and we denote the spectrum of $\hat{H}$ by $\Vec{\sigma}=[\sigma_{0}, \sigma_{1},.....\sigma_{n}]^T$
    where $\sigma_{0} \le \sigma_{1} \le .... \sigma_{n}$ . We shall assume that the entries of $\hat{H}$ in the chosen basis is finite.
    
    \item $\hat{O} \in \mathbb{C}^{d\times d}$ and $\hat{O}=\hat{O}^\dagger$ is an user-defined operator for the problem. $\omega$ is an eigenvalue of operator $\hat{O}$. We denote the spectrum of operator $(\hat{O}-\omega)^2$ as $\Vec{\eta} = [\eta_{0}, \eta_{1},\eta_{2}....\eta_{n}]^T$ where $\eta_0 \le \eta_1 \le \eta_2....\le \eta_n$. Further $\eta_i \ge 0 \:\:\forall\:\: i$ as $(\hat{O}-\omega)^2 \succeq 0$ (positive-semidefinite by construction)
    
    \item $Null(\hat{A})$ = $\{|x\rangle |\:\: \hat{A}|x\rangle =0, \:\:\forall \:\:|x\rangle \in \mathbb{C}^d\}$ where $\hat{A}$ is any arbitrary operator $\in \mathbb{C}^{d\times d}$.
    
    \item $\lambda \in \mathbb{R}_{++}$ is a penalty parameter
    
    \item $|\psi\rangle \in \mathbb{C}^d$
    is the state-vector of the system sought from the minimization scheme by training the neural network.
    
    \item $\{\lambda_i\}_{i=1}^{\infty}$ is a sequence in the penalty parameter such that $\lambda_1 \le \lambda_2 \le \lambda_3..... \lambda_{\infty} \rightarrow \infty$
    
    \item $P=\{|\psi_i\rangle\}_{i=1}^{\infty}$ such that $\forall$ $|\psi_i\rangle$ $\in$ $P$ the following is true.
\begin{align}
    |\psi_i\rangle &= \argmin_{\psi} {\rm{F}(\lambda_i, \hat{H}, \hat{O}, |\psi\rangle)}
\end{align}
In other words $P$ is the set of minimizers for Eq. \ref{cost_fn_supp} for each penalty parameter $\lambda \in \{\lambda_i\}_{i=1}^{\infty}$.

\item $|\psi^*\rangle = \lim_{i \to \infty} |\psi_i\rangle$ be the limit point of a convergent sequence in $P$

\item All vectors $\in \mathbb{C}^d$ discussed below will be considered normalized unless otherwise stated
\end{enumerate}

Using the definitions above we construct the following lemmas.
\newpage

\begin{lemma}\label{lem1}
For any $|\psi\rangle \in \mathbb{C}^d$, $\langle \psi|\hat{A}|\psi\rangle \le \sqrt{Tr(\hat{A}^\dagger \hat{A})}$ where $\hat{A}$ is any arbitrary hermitian operator $\in \mathbb{C}^{d\times d}$.
\begin{proof}
Let us denote the variance of operator $\hat{A}$ as $Var(\hat{A})$ evaluated in an arbitrary state $|\psi\rangle \in \mathbb{C}^d$. $Var(\hat{A})$ by definition is always non-negative. From this we can claim
\begin{align}
    Var(\hat{A}) = \langle \psi|\hat{A}^2|\psi\rangle - (\langle \psi|\hat{A}|\psi\rangle)^2
    &\ge 0  \:\:\:\:\because (\rm{by\:\: definition}) \nonumber \\
    \langle \psi|\hat{A}^2|\psi\rangle &\ge (\langle \psi|\hat{A}|\psi\rangle)^2 \nonumber \\
    \langle \psi|\hat{A}^\dagger\hat{A}|\psi\rangle &\ge (\langle \psi|\hat{A}|\psi\rangle)^2 \:\:\:\:\because (\hat{A}=\hat{A}^\dagger) \label{eqn_exp_val} 
\end{align}
Now let us consider a complete set of eigenvectors of $\hat{A}^\dagger \hat{A}$ denoted as $\{|s_i\rangle\}_{i=1}^d$ with corresponding eigenvalues
$\{s_i\}_{i=1}^d$ which are non-negative as $\hat{A}^\dagger \hat{A} \succeq 0$. One can resolve the state $|\psi\rangle$ in the basis $\{|s_i\rangle\}_{i=1}^d$ as follows
\begin{align}
    |\psi\rangle = \sum_{i=1}^d \langle s_i|\psi\rangle |s_i\rangle\ \label{psi_res}
\end{align}
Using Eq. \ref{psi_res} to express $\langle \psi|\hat{A}^\dagger\hat{A}|\psi\rangle$ we have
\begin{align}
    &\langle \psi|\hat{A}^\dagger\hat{A}|\psi\rangle \nonumber \\
    &=\sum_{i=1}^d s_i |\langle s_i|\psi\rangle|^2 \nonumber \\
    & \le \sum_{i=1}^d s_i \:\:\:\:\:\:\because s_i \ge 0 \:\:\rm{and}\:\: 0 \le |\langle s_i|\psi\rangle|^2 \le 1 \nonumber \\
    & = Tr(\hat{A}^\dagger A) \label{tr_defn}
\end{align}
We can thus make the following claim
\begin{align}
    &Tr(\hat{A}^\dagger A) \nonumber \\
    &\ge \langle \psi|\hat{A}^\dagger\hat{A}|\psi\rangle \:\:\:\: \because Eq. \ref{tr_defn}\nonumber \\
    &\ge (\langle \psi|\hat{A}|\psi\rangle)^2 \:\:\:\: \because Eq. \ref{eqn_exp_val}\nonumber \\
    \implies \langle \psi|\hat{A}|\psi\rangle \le \sqrt{Tr(\hat{A}^\dagger A)}
\end{align}
\end{proof}
\end{lemma}

\begin{lemma} \label{lem2}
$|\psi^*\rangle \in P$ and is a limit-point of the convergent sequence in $P$ if  $\langle \psi^*|(\hat{O}-\omega)^2|\psi^*\rangle =0  $
\end{lemma}
\begin{proof}
Let us consider a state $|\psi^\prime\rangle$ $\in$  $S$, where the set $S$ is defined as
\begin{align}
    S &= \{|x\rangle |\: \hat{O}|x\rangle = \omega|x\rangle\:\: \forall\: |x\rangle \in \mathbb{C}^d\}
\end{align}

The following is  then true
\begin{align}
    \langle \psi^\prime|\hat{H}|\psi^\prime\rangle &= \langle \psi^\prime| \hat{H}  |\psi^\prime\rangle + \lambda_k\langle \psi^\prime| (\hat{O}-\omega)^2|\psi^\prime \rangle \:\:\:\:\:\:\:\: \because |\psi^\prime\rangle \in  S \label{quad_form}\\
    &= \rm{F}(\lambda_k, \hat{H}, \hat{O}, |\psi^\prime\rangle) \label{F_psi_p} \\
    & \le  \sqrt{Tr(\hat{H}^\dagger H)} \:\:\:\:\: \because (\rm{Lemma \:\ref{lem1} \:\: using \:\hat{A}= \hat{H}})  \label{HH_tr}
\end{align}
Now since the $\hat{H}$ operator is assumed to have elements which are all finite (see definitions in \ref{defn}) (1), and since $\sqrt{Tr(\hat{H}^\dagger H)}$ is a polynomial on the matrix elements of $\hat{H}$, one can say $\rm{F}(\lambda_k, \hat{H}, \hat{O}, |\psi^\prime\rangle)$ in Eq. \ref{F_psi_p} is always upper-bounded by a finite number (see Eq. \ref{F_psi_p} and Eq. \ref{HH_tr}). One should note that $\sqrt{Tr(\hat{H}^\dagger H)}$ being a trace property is invariant to the choice of basis for 
expressing the matrix elements of $\hat{H}$ and also independent of any state $|\psi^\prime\rangle$ used for computing $\langle \psi^\prime|\hat{H}|\psi^\prime\rangle$. Also $ \rm{F}(\lambda_k, \hat{H}, \hat{O}, |\psi^\prime\rangle) \ge  \rm{F}(\lambda_k,\hat{H}, \hat{O}, |\psi_k\rangle)$ as $|\psi_k\rangle = \argmin_{\psi} {\rm{F}(\lambda_k, \hat{H}, \hat{O}, |\psi\rangle)}$. This is true for any $\lambda_k$ as $|\psi_k\rangle$ is the specific minimizer of the cost function in Eq. \ref{cost_fn_supp} for that $\lambda_k$ and hence will produce a cost-function in Eq. \ref{cost_fn_supp} lower in value
than with any state $|\psi^\prime\rangle$. With this information we see
\begin{align}
 \lim_{k \to \infty}\rm{F}(\lambda_k, \hat{H}, \hat{O},  |\psi_k\rangle) \nonumber
 &\le \sqrt{Tr(\hat{H}^\dagger H)} \nonumber \\
 \lim_{k \to \infty} \langle \psi_k| \hat {H} |\psi_k \rangle + \lim_{k \to \infty} \lambda_k \langle \psi_k|(\hat{O}-\omega)^2\rangle|\psi_k\rangle &\le \sqrt{Tr(\hat{H}^\dagger H)}  \\
  \lim_{k \to \infty} \lambda_k \langle\psi_k|(\hat{O}-\omega)^2|\psi_k\rangle &\le \sqrt{Tr(\hat{H}^\dagger H)} - \lim_{k \to \infty} \langle \psi_k| \hat {H} |\psi_k\rangle  \label{eq_8}\\
  \lim_{k \to \infty} \langle\psi_k|(\hat{O}-\omega)^2|\psi_k\rangle = \langle \psi^*|(\hat{O}-\omega)^2|\psi^*\rangle &=0 \:\:\:\:\:\:\:\:\:  (\lambda_k \rightarrow \infty)   \label{eq_9}
\end{align}

where in arriving at Eq. \ref{eq_9} from Eq. \ref{eq_8} we have used the fact that $\sqrt{Tr(\hat{H}^\dagger H)}$ is finite (as per definitions \ref{defn} (1) and the fact that $\sqrt{Tr(\hat{H}^\dagger H)}$ is a polynomial on the matrix elements of $\hat{H}$) and $\langle \psi_k|\hat{H}\psi_k\rangle$ being a quadratic form is also upper-bounded using same Lemma \ref{lem1} and hence is finite. Thus the RHS of Eq. \ref{eq_8} is a finite-upper bound on the LHS. Only way then the LHS of Eq. \ref{eq_8} can thus stay finite in the limit $\lambda \rightarrow \infty$ is when $\langle \psi_k|(\hat{O}-\omega)^2|\psi_k \rangle$ is pinned to zero. Since $|\psi^*\rangle = \lim_{k \to \infty} |\psi_k\rangle$ is the convergent limit point, the result immediately follows.
\end{proof}

\begin{lemma}\label{lem3}
$\langle \psi^*|(\hat{O}-\omega)^2|\psi^*\rangle =0  $ if and only if $|\psi^*\rangle \in Null(\hat{O}-\omega)^2)$
\end{lemma}
\begin{proof}
\underline{If-part}\\
If $|\psi^*\rangle \in Null((\hat{O}-\omega)^2)$
\begin{align}
    (\hat{O}-\omega)^2|\psi^*\rangle &=0 \:\:\:\:\:\:\:\:\:\:\: (|\psi^*\rangle \in Null(\hat{O}-\omega)^2)) \nonumber \\
    \langle \psi^*|(\hat{O}-\omega)^2|\psi^*\rangle &=0 \nonumber
\end{align}
\underline{Only If-part}\\
Let us define a set of eigenvectors of $(\hat{O}-\omega)^2$ as $\{|\eta_i\rangle\}_i$. Since the set is complete one can expand 
\begin{align}
    |\psi^*\rangle &= \sum_i \langle \eta_i|\psi^*\rangle |\eta_i\rangle  \nonumber \\
    &=\sum_{|\eta_i\rangle \in Null((\hat{O}-\omega)^2)} \langle \eta_i|\psi^*\rangle |\eta_i\rangle + \sum_{|\eta_i\rangle \not\in Null((\hat{O}-\omega)^2)} \langle \eta_i|\psi^*\rangle |\eta_i\rangle \label{psi_exp}
\end{align}
One can use Eq. \ref{psi_exp} in 
$\langle \psi^*|(\hat{O}-\omega)^2|\psi^*\rangle$ to arrive at
\begin{align}
  \langle \psi^*|(\hat{O}-\omega)^2|\psi^*\rangle &= \sum_{|\eta_i\rangle \in Null((\hat{O}-\omega)^2)} 0 |\langle \eta_i|\psi^*\rangle|^2 + \sum_{|\eta_i\rangle \not\in Null((\hat{O}-\omega)^2)} \eta_i |\langle \eta_i|\psi^*\rangle|^2 \nonumber \\
  &= 0 \:\:\:\:\: \rm{(by\:\:\: condition)}\nonumber \\
  &\implies \langle \eta_i|\psi^*\rangle =0 \:\: \forall \:\:|\eta_i\rangle \:\:\not\in Null((\hat{O}-\omega)^2)   \:\:\:\:\: (\eta_i \ge 0, \:\:see\:\: \ref{defn}(2)) \nonumber \\
  &\implies |\psi^*\rangle \in Null((\hat{O}-\omega)^2) \:\:\:\: \because\:\:\rm{Using}\:\:Eq. \ref{psi_exp}
\end{align}
\end{proof}

\begin{lemma}\label{lem4}
$Null(\hat{O}-\omega)) = Null (\hat{O}-\omega)^2)$
\end{lemma}
\begin{proof}
This is actually trivial to show. For proof one can see \cite{Axler_s}
\end{proof}

\subsection{Theorem 2.1 in main text}
Then using the Lemmas above the following is true.

\begin{theorem}
Let $\{\lambda_i\}_{i=1}^{\infty}$ be a sequence in the penalty parameter such that $\lambda_1 \le \lambda_2 \le \lambda_3..... \lambda_{\infty} \rightarrow \infty$
Also let $P=\{|\psi_i\rangle\}_{i=1}^{\infty}$ such that $\forall$ $|\psi_i\rangle$ $\in$ $P$ the following is true.
\begin{align}
    |\psi_i\rangle &= \argmin_{\psi} {\rm{F}(\lambda_i, \hat{H}, \hat{O}, |\psi\rangle)}
\end{align}
In other words $P$ is the set of minimizers for Eq. \ref{cost_fn_supp} for each penalty parameter $\lambda \in \{\lambda_i\}_{i=1}^{\infty}$
If $|\psi^*\rangle \in P$  is a limit-point of the convergent sequence $\{\psi_i\}_{i=1}^\infty$ in $P$ i.e $|\psi^*\rangle = \lim_{i \to \infty} |\psi_i\rangle$ then $|\psi^*\rangle$ $\in$ $S$ (defined in lemma \ref{lem2})
\end{theorem}

\begin{proof}
If $|\psi^*\rangle$ is a limit-point of the convergent sequence $\{\psi_i\}^\infty$ in $P$ then
\begin{align}
  \langle \psi^*|(\hat{O}-\omega)^2  |\psi^*\rangle &= 0 \:\:\:\:\:\:\because \rm{see\:\:Lemma\:\:} \ref{lem2}\nonumber \\
  &\implies |\psi^*\rangle \in Null((\hat{O}-\omega)^2) \:\:\:\:\:\:\because \rm{see\:\:Lemma\:\:} \ref{lem3}\nonumber \\
  &\implies |\psi^*\rangle \in Null((\hat{O}-\omega)) \:\:\:\:\:\:\because \rm{see\:\:Lemma\:\:} \ref{lem4}\nonumber \\
  &\implies |\psi^*\rangle \in S
\end{align}
\end{proof}

\section{Deduction of a generic lower bound for successful sampling and characterization of k-parameter}

After all the single qubit $R_y$ rotations (parameterized by the bias vectors of the network $(\vec{a}, \vec{b})$) and Controlled-Controlled Rotations ($C-C-R_y$) targeting the ancillas (parameterized by the interconnecting weights $\vec{W}$ between visible and hidden neurons), the state-vector $|\psi_{v,h,a}\rangle$ of the full set of $(m+n+m\times n)$ qubits is
\begin{align}
   |\psi_{v,h,a}\rangle &= \sum_{(\vec{\sigma}, \vec{h})} \sqrt{O(\vec{\sigma}, \vec{h}, \vec{a}, \vec{b})} |\vec{\sigma} \vec{h}\rangle_{vh} \otimes (\sqrt{(1-\eta(\vec{W},\vec{\sigma}, \vec{h})}|\vec{0}\rangle_a + \sqrt{\eta(\vec{W},\vec{\sigma}, \vec{h})}|\vec{1}\rangle_a ) \label{state_all}
\end{align}
where the following definitions is used.
\begin{enumerate}
    \item $|\psi_{v,h,a}\rangle$ is the combined state of the visible node qubits (abbreviated by subscript v), hidden node qubits (abbreviated by subscript h) and ancilla register (abbreviated by subscript a)
    
    \item $(\vec{\sigma}, \vec{h})$ denotes a sum over the $2^{m+n}$ bit strings where each $\{\sigma_i\}$ or  $\{h_j\} \in \{1,-1\}$
    
    \item $|\vec{\sigma}\vec{h}\rangle_{vh}$ is the $2^{m+n}$ dimensional state space of $n$ visible node qubits and $m$ hidden node qubits. Note the state $|0\rangle_v$ corresponds to $\sigma_i = -1$ as mentioned in the main text. Similar statement holds for $h_j=-1$ and $|0\rangle_h$ too.
    
    \item The distribution $O(\vec{\sigma}, \vec{h}, \vec{a}, \vec{b})$ is
    \begin{align}
        O(\vec{\sigma}, \vec{h}, \vec{a}, \vec{b}) &= \frac{e^{\frac{1}{k}(\sum_{i}a_i\sigma_i + \sum_{j}b_j h_j)}}{\sum_{\vec{\sigma}, \vec{h}}e^{\frac{1}{k}(\sum_{i}a_i\sigma_i + \sum_{j}b_j h_j)}} \label{O_dist}
    \end{align}
    \item The distribution $\eta(\vec{W},\vec{\sigma}, \vec{h})$ is
       \begin{align}
        \eta(\vec{W},\vec{\sigma}, \vec{h}) = \frac{e^{\frac{1}{k}(\sum_{i,j}w_{ij}\sigma_ih_j)}}{e^{\frac{1}{k}\sum_{i,j}|w_{ij}|}}
    \end{align}
    
    \item $|\vec{0}\rangle_a$ and $|\vec{1}\rangle_a$ are the states of the $m\times n$ ancilla qubits (abbreviated as superscript a)
\end{enumerate}

Thus from Eq. \ref{state_all} we see that when all the qubits are measured, the probability of selecting a bit string  $(\vec{\sigma}, \vec{h})$ \textbf{and} collapsing the ancilla qubits in state $|\vec{1}\rangle_a$ (only such states are important to this work as they are post-selected after measurement).

\begin{align}
    H((\vec{\sigma}, \vec{h}) \cap \vec{1}) &= O(\vec{\sigma}, \vec{h}, \vec{a}, \vec{b})\eta(\vec{W},\vec{\sigma}, \vec{h}) \nonumber \\
    &=\frac{e^{\frac{1}{k}(\sum_{i}a_i\sigma_i + \sum_{j}b_j h_j)}}{\sum_{{\vec{\sigma}, \vec{h}}}e^{\frac{1}{k}(\sum_{i}a_i\sigma_i + \sum_{j}b_j h_j)}} \times  \frac{e^{\frac{1}{k}(\sum_{i,j}w_{ij}\sigma_ih_j)}}{e^{\frac{1}{k}\sum_{i,j}|w_{ij}|}} \label{dist_H}
\end{align}
Now successful sampling would be an event wherein all ancilla would collapse to $|\vec{1}\rangle_a$. The probability of such events (denoted as $P_{success} = P(\vec{1}_a)$) irrespective of $(\vec{\sigma}, \vec{h})$ string selected can be obtained by marginalizing $H((\vec{\sigma}, \vec{h}) \cap \vec{1})$ \textbf{over all} bit-strings as follows: \newpage

\begin{align}
    P_{success} = P(\vec{1}_a) &= \sum_{\vec{\sigma},\vec{h}} H((\vec{\sigma}, \vec{h}) \cap \vec{1}) \nonumber \\
    & = \sum_{{\vec{\sigma}, \vec{h}}}\frac{e^{\frac{1}{k}(\sum_{i}a_i\sigma_i + \sum_{j}b_j h_j)}}{\sum_{{\sigma, h}}e^{\frac{1}{k}(\sum_{i}a_i\sigma_i + \sum_{j}b_j h_j)}} \times 
    \frac{e^{\frac{1}{k}(\sum_{i,j}w_{ij}\sigma_ih_j)}}{e^{\frac{1}{k}\sum_{i,j}|w_{ij}|}}  \label{P_act}\\
    & = \frac{\langle e^{\frac{1}{k}(\sum_{i,j}w_{ij}\sigma_ih_j)} \rangle _{O(\vec{\sigma}, \vec{h}, \vec{a}, \vec{b})}}{e^{\frac{1}{k}\sum_{i,j}|w_{ij}|}} \nonumber \\
    & \ge \frac{ e^{\frac{1}{k}(\sum_{i,j}w_{ij} \langle \sigma_ih_j\rangle _{O(\vec{\sigma}, \vec{h}, \vec{a}, \vec{b})} )} }{e^{\frac{1}{k}\sum_{i,j}|w_{ij}|}} \:\:\:\:\:\: \because (\rm{Jensen's \:\: inequality}) \label{P_succ} \\\nonumber\\
    \langle \sigma_ih_j\rangle _{O(\vec{\sigma}, \vec{h}, \vec{a}, \vec{b})} &= \sum_{\vec{\sigma}, \vec{h}} O(\vec{\sigma}, \vec{h}, \vec{a}, \vec{b}) \sigma_ih_j\nonumber \\
    &= \sum_{\vec{\sigma}} O_1(\vec{\sigma}, \vec{a}) \sigma_i\sum_{\vec{h}} h_j O_2(\vec{h}, \vec{b}) \:\:\:\:\:\:\:\: \because(Eq. \ref{O_dist}\:\: O(\vec{\sigma}, \vec{h}, \vec{a}, \vec{b}) = O_1(\vec{\sigma}, \vec{a})O_2(\vec{h}, \vec{b}) \nonumber \\
    &= \sum_{\vec{\sigma}} \prod_m O_1(\sigma_m, a_m) \sigma_i\sum_{\vec{h}} h_j \prod_p O_2(h_p, b_p) \:\:\:\:\:\: (\because \rm{no\:\: intralayer\:\: connections\:\: as\:\: RBM} ) \nonumber \\
    &= \sum_{\sigma_i \in \{1,-1\}} O_1(\sigma_i, a_i) \sigma_i\sum_{h_j \in \{1,-1\}} h_j O_2(h_j, b_j) \nonumber \\
    &= (\frac{e^{a_i/k}-e^{-a_i/k}}{e^{a_i/k}+e^{-a_i/k}})(\frac{e^{b_j/k}-e^{-b_j/k}}{e^{b_j/k}+e^{-b_j/k}}) \nonumber \\
    &= \tanh{(a_i/k)}\tanh{(b_j/k)} \label{tanh_ab}
\end{align}
where we have used $\langle ...\rangle_{O(\vec{\sigma}, \vec{h}, \vec{a}, \vec{b})}$ to denote an average over the distribution $O(\vec{\sigma}, \vec{h}, \vec{a}, \vec{b})$ defined in Eq. \ref{O_dist}. Using Eq. \ref{tanh_ab} in Eq. \ref{P_succ} we thus get
\begin{align}
    P_{success} &= P(\vec{1}_a) \nonumber \\
    &\ge \frac{ e^{\frac{1}{k}(\sum_{i,j}w_{ij} \langle \sigma_ih_j\rangle _{O(\vec{\sigma}, \vec{h}, \vec{a}, \vec{b})} )} }{e^{\frac{1}{k}\sum_{i,j}|w_{ij}|}} \nonumber \\
   \Aboxed{P_{success} &\ge \frac{ e^{\frac{1}{k}(\sum_{i,j}w_{ij} \tanh{(a_i/k)}\tanh{(b_j/k)}}}{e^{\frac{1}{k}\sum_{i,j}|w_{ij}|}}} \:\:\:\:\:\:\: \because Eq. \ref{tanh_ab}\label{low_bnd}
\end{align}

The above lower bound in Eq. \ref{low_bnd} is a generic lower bound independent of any occurrence of random variables $(\vec{\sigma}, \vec{h})$ and only dependant on the parameters of the network $(\vec{a}, \vec{b}, \vec{W}, k)$. {\color{black}In the plot below, we simulate the performance of the lower bound deduced in Eq. \ref{low_bnd} against the actual probability using \textbf{`RBM-qasm'} . In the plot, we designate the R.H.S of Eq. \ref{low_bnd} as $\rm{P}_{lb}$ signifying lower bound and the actual probability of the event (L.H.S of Eq. \ref{low_bnd}) as $P_{success}$ as used before}.

\begin{figure}[!htb]
    \centering
    \includegraphics[width=0.7\textwidth]{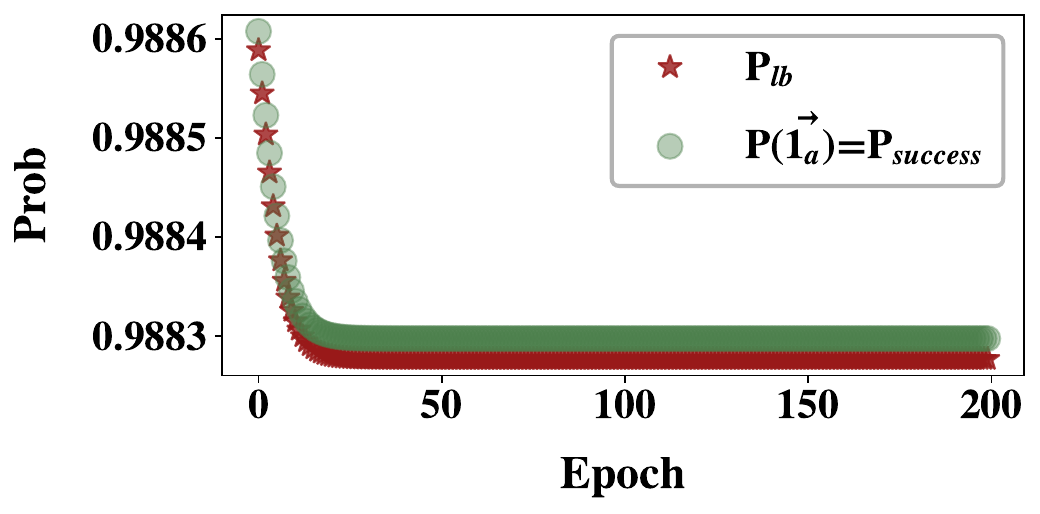}
    \caption{{\color{black} The lower bound for probability of successfully collapsing the ancilla register in state $|\vec{1_a}\rangle$ as deduced in Eq. \ref{low_bnd} (R.H.S of Eq. \ref{low_bnd}) is plotted in red and the actual probability of such an event (Eq. \ref{P_act} or L.H.S of Eq. \ref{low_bnd}) is plotted in green as a function of training iteration. We see that the green curve is always slightly above the red one in accordance with Eq. \ref{low_bnd} deduced above. The simulation is from the \textbf{`RBM-qasm'} variant for conduction band (CB) of $\rm{MoS}_2$ at $(k_x, k_y)=K$ symmetry point which is where the direct band-gap is lowest. The simulation is performed by warm starting with  initial parameter set from a converged run at a nearby k-point to ensure faster and accurate convergence. With the said warm start the desired accuracy of $\le 10^{-4}$ eV in energy error was reached within 200 iterations. The high value of the probability of the successful event (and the associated lower bound) as seen from the y-scale is problem-specific as it entirely depends on $(\vec{a}, \vec{b}, \vec{W}, k)$ (see Eq. \ref{P_act} and Eq. \ref{low_bnd}) . For these systems, even with a  moderate $k$ parameter i.e. $k \le 1.5$ for all iterations, the remaining set $(\vec{a}, \vec{b}, \vec{W})$ are such that a high value of $P_{success}$ as seen on the y-scale is attained. The parameter set $(\vec{a}, \vec{b}, \vec{W})$ depends on the updates from the cost-function and hence on the Hamiltonian of the system being treated. For this choice of $k$ parameter, the specific values of the y-scale in the plot is thus characteristic of the systems   being studied in this report and may be different for other systems. However the lower bound deduced in Eq. \ref{low_bnd} is mathematically generic and should be valid for any arbitrary system and a given $k$ parameter, even though the specific value it acquires during training may vary}}
    \label{Fig:Plb_vs_Psucc}
\end{figure}

From this lower bound, all pre-existing known bounds can be recovered as we shall see next.

\textbf{Limiting cases}
\begin{itemize}
\item \underline {$\tanh{(a_i/k)} \to \pm 1 , \tanh{(b_j/k)} \to \pm 1$}

In this case we get from Eq. \ref{low_bnd}
\begin{align}
    P_{success} &= P(\vec{1}_a) \nonumber \\
    &\ge \frac{ e^{\frac{1}{k}(\sum_{i,j}w_{ij} \tanh{(a_i/k)}\tanh{(b_j/k)}}}{e^{\frac{1}{k}\sum_{i,j}|w_{ij}|}} \nonumber \\
    &\ge \frac{ e^{-\frac{1}{k}|(\sum_{i,j}w_{ij} \tanh{(a_i/k)}\tanh{(b_j/k)}|}}{e^{\frac{1}{k}\sum_{i,j}|w_{ij}|}} \nonumber \:\:\:\:\: \because e^{\alpha x} \ge e^{-\alpha|x|} \:\:\forall x \in \mathbb{R}\:\:, \alpha \ge 0) \\
    &\ge \frac{ e^{-\frac{1}{k}(\sum_{i,j}|w_{ij} \tanh{(a_i/k)}\tanh{(b_j/k)}|}}{e^{\frac{1}{k}\sum_{i,j}|w_{ij}|}} \nonumber \:\:\:\:\: \because \rm{Triangle \:\: Inequality} \\
    &\ge \frac{ e^{-\frac{1}{k}(\sum_{i,j}|w_{ij}| |\tanh{(a_i/k)}||\tanh{(b_j/k)}|}}{e^{\frac{1}{k}\sum_{i,j}|w_{ij}|}} \nonumber  \\
    &\ge \frac{ e^{-\frac{1}{k}(\sum_{i,j}|w_{ij}| )}}{e^{\frac{1}{k}\sum_{i,j}|w_{ij}|}} \:\:\:\:\:\:\: \because |\tanh{(a_i/k)}| = |\tanh{(b_j/k)}|=1 \nonumber \\
    &\ge \frac{1}{e^{\frac{2}{k}(\sum_{ij}|w_{ij}|)}} \label{rongxin_bnd} 
\end{align}
Thus choosing $k = \max(\sum_{ij}\frac{|w_{ij}|}{2}, 1)$, the lower bound attained in Eq. \ref{rongxin_bnd} for the probability for successful sampling becomes a constant value of $e^{-4}$. This bound was deduced in \cite{Xia2018_s} in a completely different manner. Here we derived a master bound from which this is recovered. \newpage

\item \underline {$\tanh{(a_i/k)} \to (a_i/k) , \tanh{(b_j/k)} \to (b_j/k)$}

Using similar kind of reasoning as in the previous case one can show
\begin{align}
P_{success} &= P(\vec{1}_a) \nonumber \\
&\ge \frac{ e^{-\frac{1}{k}(\sum_{i,j}|w_{ij}| |\tanh{(a_i/k)}||\tanh{(b_j/k)}|}}{e^{\frac{1}{k}\sum_{i,j}|w_{ij}|}} \nonumber  \\
&\ge \frac{ e^{-\frac{1}{k}(\sum_{i,j}|w_{ij}| |(a_i/k)||(b_j/k)|}}{e^{\frac{1}{k}\sum_{i,j}|w_{ij}|}} \nonumber
\end{align}
Here one cannot do much to proceed unless some assumptions are made 
.We assume that $|a_i| \le a_0 \:\:\forall\:\: i$ and $|b_j| \le b_0 \:\:\forall\:\: j$. Then we get
\begin{align}
P_{success} &= P(\vec{1}_a) \nonumber \\
&\ge \frac{ e^{-\frac{1}{k}(\sum_{i,j}|w_{ij}| |(a_i/k)||(b_j/k)|}}{e^{\frac{1}{k}\sum_{i,j}|w_{ij}|}} \nonumber  \\
& \ge \frac{ e^{-\frac{a_0b_0}{k^3}(\sum_{i,j}|w_{ij}| }}{e^{\frac{1}{k}\sum_{i,j}|w_{ij}|}} \nonumber  \\
& \ge \frac{1}{e^{(\frac{1}{k}+\frac{a_0b_0}{k^3})(\sum_{i,j}|w_{ij}|}} \label{new_bnd}
\end{align}
One can then numerically choose $k$ to make the bound in this limit in Eq. \ref{new_bnd} a constant value greater than a user-defined real number.\newpage
\end{itemize}

\section{Hamiltonian for $\rm{MoS}_2$ and $\rm{WS}_2$}

For the Hamiltonian matrix used in this report for $\rm{MoS}_2$ and $\rm{WS}_2$, we use the 3-band third-nearest neighbor tight-binding description as adopted in Ref.\cite{Liu2013_s}. The Hamiltonian matrix elements are provided below for brevity and completeness.

The Hamiltonian matrix ($\hat{H}$) being a $3\times 3$ description is written as 
\begin{align}
\hat{H} &=
\begin{bmatrix}
H_{11} & H_{12} & H_{13} \\
H_{21} & H_{22} & H_{23} \\
H_{31} & H_{32} & H_{33} \\
\end{bmatrix}
\end{align}

Since the Hamiltonian is hermitian ($\hat{H} = \hat{H}^\dagger$), the only unique elements are the upper triangular block. Each such element is described below. For each of the elements we use the symbol $1.0i = \sqrt{-1} $ to denote the imaginary components. Also $a_0$ is the lattice constant which is 3.190 \AA\: for $\rm{MoS}_2$ and 3.191 \AA\: for $\rm{WS}_2$ \cite{Liu2013_s}

\begin{align}
    H_{11} &= \epsilon_1 + 2t_0(2\cos({\frac{k_x a_0}{2}})\cos({\frac{\sqrt{3}k_y a_0}{2}}) + \cos({k_x a_0}))+ 2r_0(2\cos({\frac{3k_x a_0}{2}})\cos({\frac{\sqrt{3}k_y a_0}{2}}) + \cos({\sqrt{3}k_y a_0}))  \nonumber \\ &+ 2u_0(2\cos({k_x a_0})\cos({\sqrt{3}k_y a_0}) + \cos({2k_x a_0})) \nonumber \\ \nonumber \\
 H_{12} &= -2\sqrt{3}t_2\sin({\frac{k_x a_0}{2}})\sin({\frac{\sqrt{3}k_y a_0}{2}})+ 2(r_1 + r_2)\sin({\frac{3k_x a_0}{2}})\sin({\frac{\sqrt{3}k_y a_0}{2}}) \nonumber \\ &-2\sqrt{3}u_2\sin({k_x a_0})\sin({\sqrt{3}k_y a_0}) +  1.0i*2t_1\sin({\frac{k_x a_0}{2}})(2\cos({\frac{k_x a_0}{2}})+\cos({\frac{\sqrt{3}k_y a_0}{2}})) \nonumber\\ &+ 2(r_1-r_2)(\sin({\frac{3k_x a_0}{2}})\cos({\frac{\sqrt{3}k_y a_0}{2}})) + 2u_1 \sin({k_x a_0})(2\cos({k_x a_0})+\cos({\sqrt{3}k_y a_0})) \nonumber \\ \nonumber \\
 H_{13} &= 2t_2(\cos({k_x a_0})-\cos({\frac{\sqrt{3}k_y a_0}{2}})\cos({\frac{k_x a_0}{2}})) - \frac{2}{\sqrt{3}}(r_1 + r_2)(\cos({\frac{3k_x a_0}{2}})\cos({\frac{\sqrt{3}k_y a_0}{2}}) -\cos({\sqrt{3}k_y a_0}))\nonumber \\
 &+ 2u_2(\cos({2k_x a_0}) - \cos({k_x a_0})\cos({\sqrt{3}k_y a_0})) + 1.0i*2\sqrt{3}t_1\cos({\frac{k_x a_0}{2}})\sin({\frac{\sqrt{3}k_y a_0}{2}})\nonumber \\ &+\frac{2}{\sqrt{3}}(r_1 - r_2)\sin({\frac{\sqrt{3}k_y a_0}{2}})(\cos({\frac{3k_x a_0}{2}}) + 2\cos({\frac{\sqrt{3}k_y a_0}{2}})) + 2\sqrt{3}u_1\cos({k_x a_0})\sin({\sqrt{3}k_y a_0}) \nonumber \\ \nonumber \\
 H_{22} &= \epsilon_2 + (t_{11} +3t_{22})\cos({\frac{k_x a_0}{2}})\cos({\frac{\sqrt{3}k_y a_0}{2}})) + 2t_{11}\cos({k_x a_0}) +
 4r_{11}\cos({\frac{3k_x a_0}{2}})\cos({\frac{\sqrt{3}k_y a_0}{2}}) \nonumber \\ &+ 2(r_{11} + \sqrt{3}r_{12})\cos({\sqrt{3}k_y a_0}) + (u_{11} + 3u_{22})\cos({k_x a_0})\cos({\sqrt{3}k_y a_0}) + 2u_{11}\cos({2k_x a_0}) \nonumber \\ \nonumber \\
 H_{23} &= \sqrt{3}(t_{22} - t_{11})\sin({k_x a_0})\sin({\sqrt{3}k_y a_0}) + 4r_{12}\sin({\frac{3k_x a_0}{2}})\sin({\frac{\sqrt{3}k_y a_0}{2}}) \nonumber \\ &+ \sqrt{3}(u_{22} - u_{11})\sin({k_x a_0})\sin({\sqrt{3}k_y a_0}) + 1.0i*4t_{12}\sin({\frac{k_x a_0}{2}})(\cos({\frac{k_x a_0}{2}}) - \cos({\frac{\sqrt{3}k_y a_0}{2}})) \nonumber \\
 &+ 4u_{12}\sin({k_x a_0})(\cos({k_x a_0}) - \cos({\sqrt{3}k_y a_0}))\nonumber \\ \nonumber \\
 H_{33} &= \epsilon_2 + (3t_{11} +t_{22})\cos({\frac{k_x a_0}{2}})\cos({\frac{\sqrt{3}k_y a_0}{2}}) + 2t_{22}\cos({k_x a_0}) +
 2r_{11}(2\cos({\frac{3k_x a_0}{2}})\cos({\frac{\sqrt{3}k_y a_0}{2}}) \nonumber \\ &+
 \cos({\sqrt{3}k_y a_0})) + 
 \frac{2}{\sqrt{3}}r_{12} (4\cos({\frac{3k_x a_0}{2}})\cos({\frac{\sqrt{3}k_y a_0}{2}}) - \cos({\sqrt{3}k_y a_0})) \nonumber \\ &+ (3u_{11} + u_{22})\cos({k_x a_0})\cos({\sqrt{3}k_y a_0}) + 2u_{22}\cos({2k_x a_0}) 
 \end{align}

The energy parameter set\cite{Liu2013_s} for both the systems $\rm{MoS}_2$ and $\rm{WS}_2$ is tabulated below\\

\begin{tabular}{ |p{3cm}||p{3cm}|p{3cm}|}
 \hline
 \multicolumn{3}{|c|}{Parameter List for three-band model from GGA calculations} \\
 \hline
 Parameter(eV) & $\rm{MoS}_2$ & $\rm{WS}_2$\\
 \hline
 $\epsilon_1$   & 0.683    & 0.717 \\
 $\epsilon_2$ & 1.707  & 1.916\\
 $t_0$ & -0.146 & -0.152\\
 $t_1$ & -0.114 & -0.097\\
 $t_2$ &   0.506  & 0.590\\
 $t_{11}$& 0.085  & 0.047\\
 $t_{12}$& 0.162 & 0.178\\
 $t_{22}$& 0.073 & 0.016\\
 $r_{0}$& 0.060 & 0.069\\
 $r_{1}$& -0.236 & -0.261\\
 $r_{2}$& 0.067 & 0.107\\
 $r_{11}$& 0.016 & -0.003\\
 $r_{12}$& 0.087 & 0.109\\
 $u_{0}$& -0.038 & -0.054\\
 $u_{1}$& 0.046 & 0.045\\
 $u_{2}$& 0.001 & 0.002\\
 $u_{11}$& 0.266 & 0.325\\
 $u_{12}$& -0.176 & -0.206\\
 $u_{22}$& -0.150 & -0.163\\
 \hline
\end{tabular}
\newpage

\section{Importance of Measurement Error Mitigation}\label{Mem_imp}

{\color{black}In this section we simulate the performance of the algorithm using the \textbf{`RBM-IBMQ'} variant
with and without the use of Measurement Error Mitigation(MEM) \cite{Barron2020_s}. In 
Fig. \ref{Fig:MEM_vs_NOMEM}(a) we plot the final energy error after the training process for four 4 arbitrarily chosen $(k_x, k_y)$ points on the conduction band (CB) of $\rm{MoS}_2$. These points are not coincident with the symmetry points of the system as symmetry points converges better regardless. We also plot in Fig. \ref{Fig:MEM_vs_NOMEM}(b)-(e) the change in the energy error during the training iterations/epoch for each of the four points marked as (1),(2),(3),(4). These points are marked on the x-axis in Fig. \ref{Fig:MEM_vs_NOMEM}(a). From each of the plots in Fig. \ref{Fig:MEM_vs_NOMEM}(b)-(e) we see that the the red curve (with MEM) converges smoothly whereas the green dots (without MEM) displays noisy oscillations leading to poor self-convergence. To account for this statistical uncertainty, the last 30 points from each of the curves Fig. \ref{Fig:MEM_vs_NOMEM}(b)-(e) is time averaged and the results constitutes the points in Fig. \ref{Fig:MEM_vs_NOMEM}(a). The error bars on each points for the results without MEM are the sample standard deviation of these 30 points and the points themselves are sample mean. For the results with MEM since the convergence is smooth the corresponding error bars over the last 30 iterations are an order of magnitude lesser than the sample mean and hence not displayed. All of the simulations are performed in IBM-Sydney in a single run and followed till 150 iterations. Each of the 4 points are warm-started with parameter set from nearby $k$-points in the \textbf{`RBM-qasm'} variant to hasten convergence . Also in each case, the operator $\hat{O}=|v_0\rangle \langle v_0| $ is the corresponding ground state from \textbf{`RBM-qasm'} variant. This together with the fact that the same initial parameter set is used for both the runs with and without MEM eliminates errors due to faulty initialization and erroneous construction of the operator $\hat{O}$ (due to ground state infidelity) and focuses only on the errors introduced in the algorithm in the presence and absence of MEM . We see that for some of the points in the main manuscript the energy errors in the CB for $\rm{MoS}_2$ in the \textbf{`RBM-IBMQ'} variant were higher due to infidelity in the corresponding ground state which adversely affects the operator $\hat{O}$. Thus the usage of MEM, appropriate warm starting and the fact that the probability of successfully sampling of the quantum circuit for these systems is naturally very favorable (see section \ref{Mes_stats}) explains the superior quality of the results on the actual IBM hardware for all systems studied in this report. Among other factors, quality of results on the hardware can also be affected by sparsity of the $H$ matrix for a given architecture of the network with $(n,m)$ even though that point has not been investigated much in this report. Also it is generally assumed that hybrid variational algorithms like ours are resilient to certain noises\cite{Alan_noise_s} which may be playing a role as well.}


\begin{figure}[!htb]
    \centering
    \includegraphics[width=1.0\textwidth]{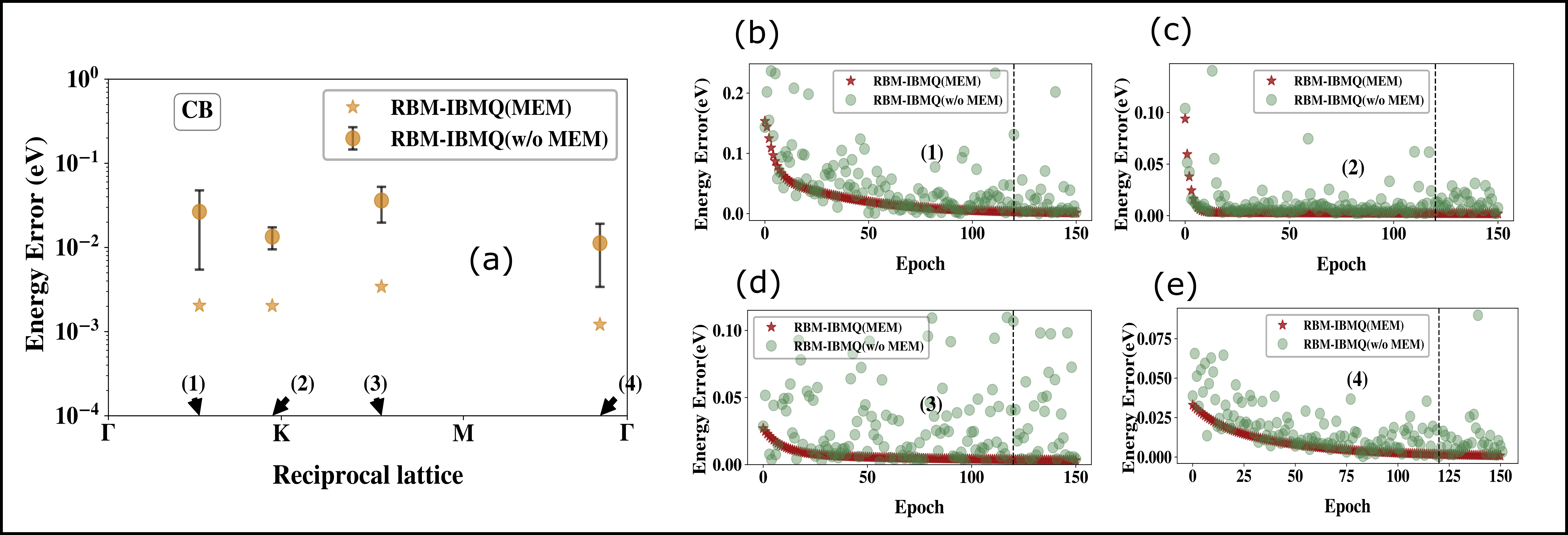}
    \caption{{\color{black} The energy error (eV) for 4 arbitrarily chosen $(k_x, k_y)$ points within the Brillouin zone in the $\Gamma-K-M-\Gamma$ path of the conduction band (CB) of $\rm{MoS}_2$ after training the network using the \textbf{`RBM-IBMQ'} variant with and without Measurement Error Mitigation (MEM). Each of the 4 points is marked on the x-axis as (1),(2),(3),(4). We also plot the energy error as a function of the training epoch/iteration with and without MEM in (b) for point index (1) in (c) for point index (2) (d) for point index (3)
    (e) for point index (4).
    We see from (a) that the results with MEM are of higher accuracy for all 4 points than without MEM. However by far the greatest impact which MEM has on the results is on improving self-convergence. This is best seen from (b)-(d). The error bars in (a) on the points without MEM are to highlight the statistical uncertainty due to time averaging from this poor self-convergence. Each such bar designates the sample standard deviation of the last 30 points (marked in (b)-(d) with a vertical dashed line) in the training process whereas the orange circles in (a) are the corresponding sample mean. Each calculation with and without MEM is done using a single run on IBM-Sydney and followed till 150 iterations. 
    All simulations are performed by warm starting with initial parameter set from a converged run at a nearby k-point in \textbf{`RBM-qasm'} variant. This is done so that same initial parameter set is used for simulations with and without MEM which eliminates biases due to random initial parameterization and affords a strictly fair comparison. The 4 points chosen are not the symmetry points as we have seen that symmetry points usually converges better regardless}}
    \label{Fig:MEM_vs_NOMEM}
\end{figure}
\newpage

\section{Measurement statistics} \label{Mes_stats}

{\color{black} For all the systems studied in the manuscript, while training the network using the quantum circuit we use $10^6$ measurement shots for the \textbf{`RBM-qasm'} variant. For the \textbf{`RBM-IBMQ'} variant (IBM-Sydney and IBM-Toronto) we use 8192 measurement shots which happens to be the maximum allowed value. We show in this section why the said number of measurement counts are adequate for all the systems we study in this report.

We simulate the statistics of the measurement using \textbf{`RBM-qasm'} in Fig. \ref{Fig:Pmeas_stat} wherein the probability of successfully collapsing the ancilla register in $|\vec{1_a}\rangle$ is plotted both in the ideal case from Eq. \ref{P_act} (see Fig. \ref{Fig:Pmeas_stat}(a)) and from the frequency distribution of the actual quantum measurement  (see Fig. \ref{Fig:Pmeas_stat}(b)). The total number of measurement count used is $10^6$. We see that for these systems, the favorable event is nearly exclusive indicating that in only less than 2\% of the measurements the ancilla register collapsed into the unwanted state $|\vec{0_a}\rangle$. The latter events are discarded when constructing the RBM distribution from the measurement statistics (see Eq. \ref{state_all}). As a result nearly $9.8\times10^5$ samples are available to faithfully construct the target distribution. Each such sample constitutes a 
$(\vec{\sigma}, \vec{h})$ pair. The frequency distribution of such samples is the simulated RBM distribution ($P_{meas}$) as obtained from the quantum circuit. The fact that this distribution agrees with the actual one ($P_{RBM}$) is quantified using the KL divergence between the two which is defined as 
\begin{eqnarray}
KL\:\:div = \langle -\rm{log} (\frac{P_{meas}}{P_{RBM}})\rangle_{P_{RBM}}
\end{eqnarray}
KL divergence can only be zero if the two distributions  $P_{meas}$ and $P_{RBM}$ exactly agree. This is seen to be the case in Fig. \ref{Fig:Pmeas_stat}(c) for all iterations during the training process. This indicates that with the $10^6$ total measurement shots , the quantum circuit faithfully produces the RBM distribution for the systems studied in the report.} 

\begin{figure}[!htb]
    \centering
    \includegraphics[width=1.0\textwidth]{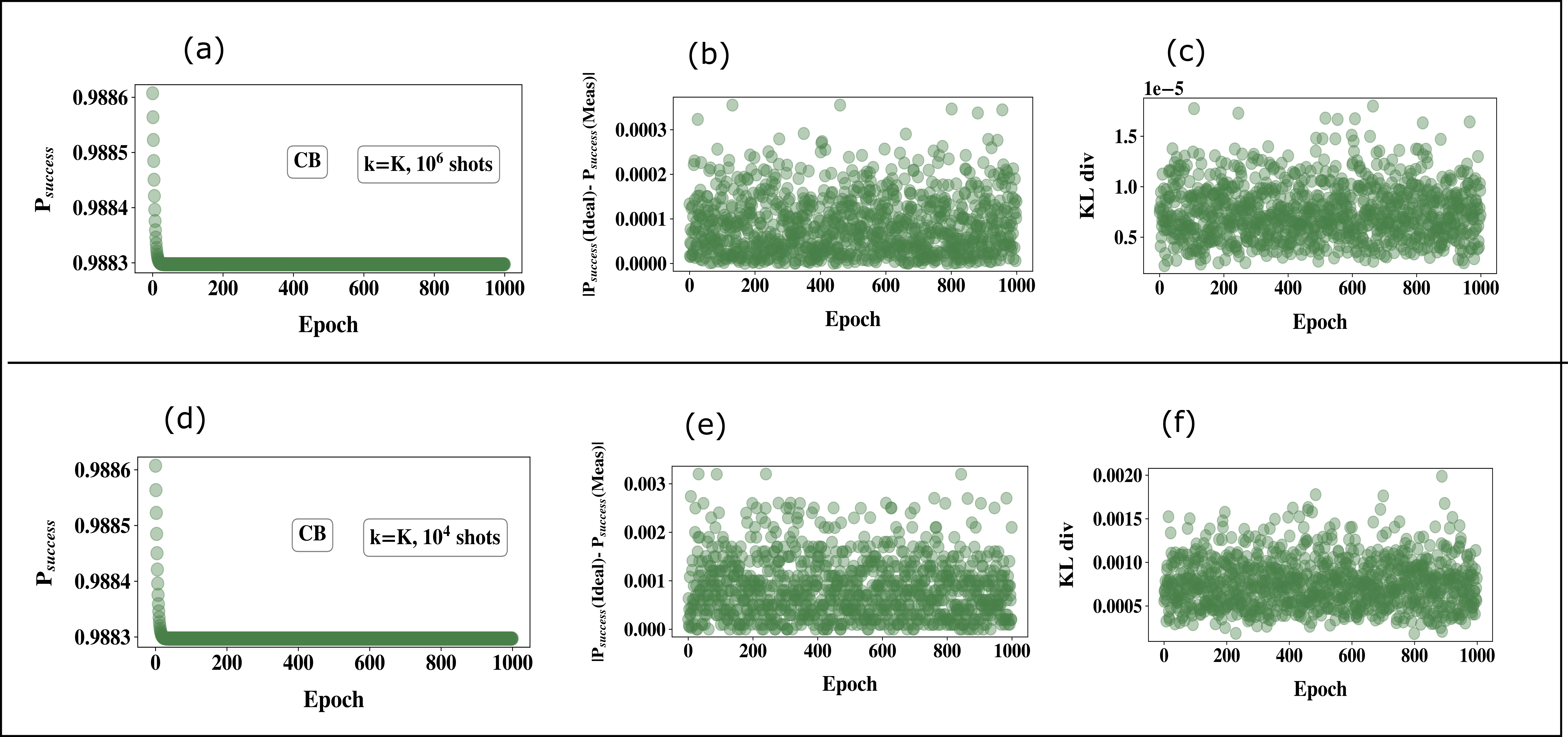}
    \caption{{\color{black}The ideal probability of successfully collapsing the ancilla register in state $|\vec{1_a}\rangle$ as computed from Eq. \ref{P_act} 
    (b) The difference between a) and the fraction of the total number of times the ancilla register collapsed in state $|\vec{1_a}\rangle$ as obtained from the direct measurement statistics in the quantum circuit. This quantity is procured by counting the number of times such an event happened while measuring all the qubits and dividing the count with the total number of measurement shots used ($10^6$ in this case). We see this value deviates only marginally from the ideal value in a) indicating that the desired event is extremely favorable. (c) The KL divergence of the distribution constructed from post-selecting all bit-strings $(\vec{\sigma}, \vec{h})$ for the visible and hidden-node qubits after the ancilla register collapsed in state $|\vec{1_a}\rangle$ and the exact RBM distribution. We see that the KL divergence is extremely close to zero indicating that the circuit can correctly learn the RBM distribution with the designated number of shots. This is because in most of them the favorable outcome of ancilla register collapsing to 
    $|\vec{1_a}\rangle$ happens naturally (see (a) and (b)) for the systems being studied in this report even with moderate $k$-parameter ($k \le 1.5$) (see text for discussion). (d)-(e)-(f) are similar plots as (a)-(b)-(c) but with $10^4$ total measurement shots. All the results are simulated in the \textbf{`RBM-qasm'} variant for conduction band (CB) of $\rm{MoS}_2$ at $(k_x, k_y)=K$ symmetry point which is where the direct band-gap is lowest. The simulation is performed by warm starting with initial parameter set from a converged run at a nearby k-point.}}
    \label{Fig:Pmeas_stat}
\end{figure} 

{\color{black} We see from Fig. \ref{Fig:Pmeas_stat} (d)-(f) that the distribution can be constructed with appreciable accuracy with $10^4$ measurement shots too. This situation simulates the measurements on \textbf{`RBM-IBMQ'} (without the noise) wherein only 8192 shots are maximally allowed. This is a consequence of the fact that the probability of successful sampling (see Fig.
\ref{Fig:Pmeas_stat}(a) and (d) is high and hence practically all measurements yields the favorable outcome. Even though from Fig. \ref{Fig:Pmeas_stat}(e) and Fig. \ref{Fig:Pmeas_stat}(f) one can see that the deviations from the expected value has increased than in \ref{Fig:Pmeas_stat}(b) or \ref{Fig:Pmeas_stat}(c) due to lesser number of shots yet such deviations are still too small to be of practicable consequences. For example, the final error in energy at the end of the training in both $10^6$ and $10^4$ shots is $\le 10^{-4}$ eV.
Thus we conclude for reasonably large number of measurements ($10^6 \:\:\rm{or}\:\: 10^4$) one can construct the distribution faithfully for the systems studied in this report with moderate $k$ parameter. Fig \ref{Fig:Pmeas_stat}(d)-(f) simulates only the ideal case for the hardware data. Of course for the \textbf{`RBM-IBMQ'} variant, the distribution during the training is further corrupted due to gate infidelities and qubit-decoherence error which is manifested in the higher errors in the final result for \textbf{`RBM-IBMQ'} variant as compared to the \textbf{`RBM-qasm'} variant in all the systems studied in the main manuscript.

It must be emphasized that 
the probability of sampling and collapsing the ancilla register in the desired state $|\vec{1_a}\rangle$ is entirely a function of the parameter set $(\vec{a}, \vec{b}, \vec{W}, k)$ encountered during the training process (see Eq. \ref{P_act}, Eq. \ref{low_bnd}). For these systems, the $k$-parameter during the course of the training is always found to be moderate $k \le 1.5$. Even with such a low $k$, the remaining set $(\vec{a}, \vec{b}, \vec{W})$ is such that a high value of $P_{success}$ is naturally obtained.
The parameter set $(\vec{a}, \vec{b}, \vec{W})$ is updated using the cost function which makes the values it acquires dependant on the Hamiltonian $\hat{H}$ and symmetry operator $\hat{O}$ of the specific system being treated. Thus the specific values of the successful sampling probability $P_{success}$ in Fig. \ref{Fig:Pmeas_stat}(a) and Fig. \ref{Fig:Pmeas_stat}(d) is somewhat characteristic of the system. This makes formulating a general expression for the total number of measurements required in our algorithm to successfully construct the target distribution (with a chosen error) difficult as even for a given architecture of the network (visible layer $n$ and hidden layer $m$) and a certain $k$ parameter, the parameter set $(\vec{a}, \vec{b}, \vec{W})$ may change for different systems which may in turn alter the probability of successful sampling thereby necessitating different number of total required measurements. 
But in systems wherein the the parameter set $(\vec{a}, \vec{b}, \vec{W})$ do not lead to favorable value of $P_{success}$ with low $k$ , $k$-parameter in our model will have to be tuned by the user adaptively to a higher value to make the lower bound for the probability of successful sampling (see Eq. \ref{low_bnd}) greater than a chosen preset thereby guaranteeing that at each iteration a good subset of these measurements are always fruitful with which the distribution can be constructed.}

\newpage

\section{Variation of the results with changing hidden node density}

{\color{black} In this section we simulate the effect of changing the hidden node density $\alpha = \frac{m}{n}$ on the results for the systems treated in this report. We take $\rm{MoS}_2$ as a prototypical example and train the network using the \textbf{`RBM-qasm'} and \textbf{`RBM-cl'} variant for the $(k_x, k_y)=K$ point for both the valence band (VB) and conduction band (CB). The number of visible node neurons $n$ as discussed for this system is 2. We vary the number of hidden neurons $m = [2,3,4]^T$ which corresponds to a changing density $\alpha = [1, 1.5, 2]^T$ respectively. Varying the hidden node density changes the circuit depth/gate-requirements and the number of parameters used for training the network. We see for these systems, the results do not change much as all energy errors are below $10^{-3}$ eV and hence way below the threshold of chemical accuracy. However for correlated systems one may need to make the RBM ansatz more expressive by enhancing the hidden node density\cite{Melko2019a_s}. We shall explore this point again for a molecular example in section \ref{LiH}}

\newpage

\begin{figure}[!htb]
    \centering
    \includegraphics[width=1.0\textwidth]{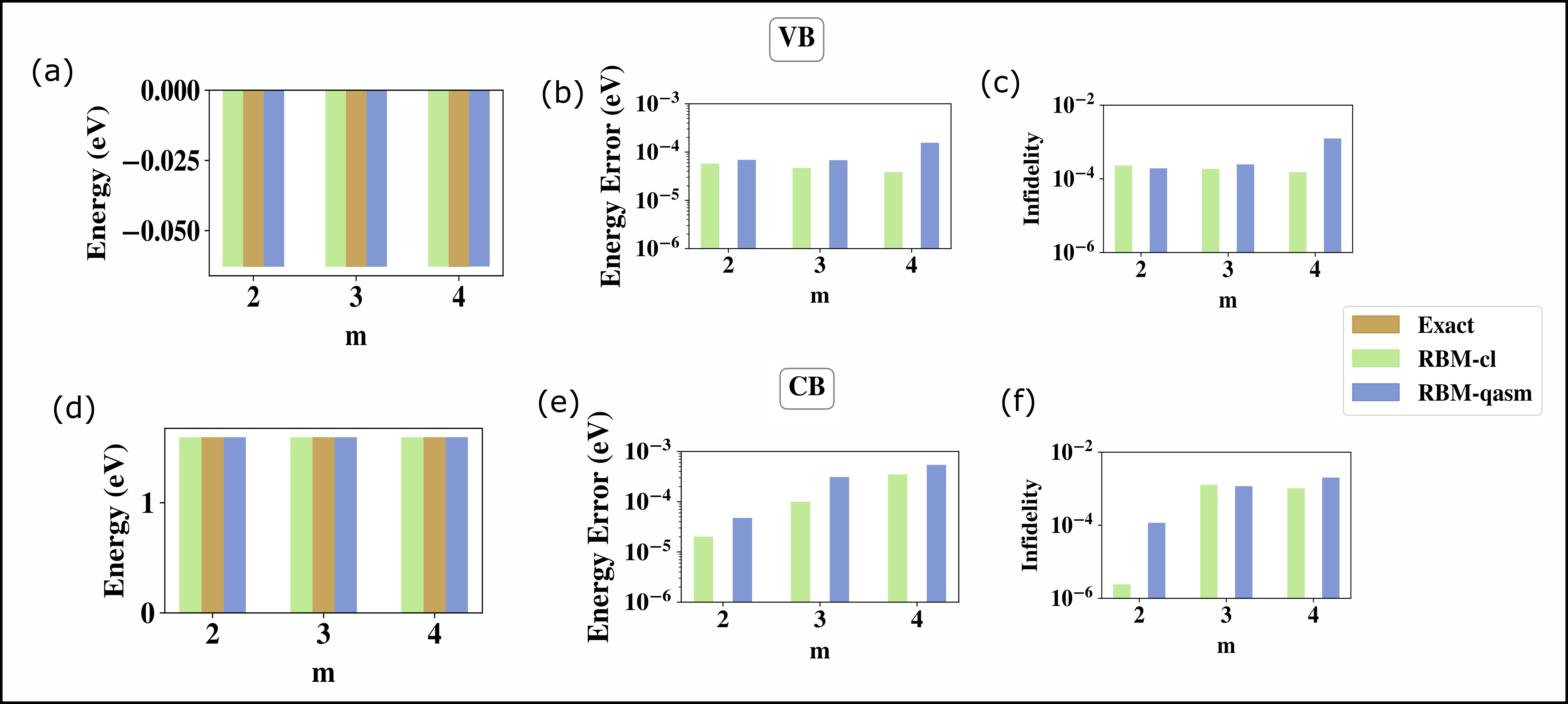}
    \caption{{\color{black}(a) The energy of the valence band (VB) at $(k_x, k_y)=K$ point for $\rm{MoS}_2$ for the exact case and the two flavors of RBM namely \textbf{`RBM-qasm'} and \textbf{`RBM-cl'} are plotted by changing the number of neurons $m$ in the hidden node. (b) The corresponding energy errors from the calculations in a) (c) The corresponding state infidelities from the calculations in a) (d) Similar result as in a) but for the conduction band (CB) at $(k_x, k_y)=K$ point for $\rm{MoS}_2$. (e) The corresponding energy errors from the calculations in d) (f) The corresponding state infidelities from the calculations in d)}}
    \label{Fig:MoS2_m_n}
\end{figure}

\newpage

\section{Spin-Orbit Coupling (SOC) data for $\rm{WS}_2$}

\begin{figure*}[!htb]
    \centering
    \includegraphics[width=0.6\textwidth]{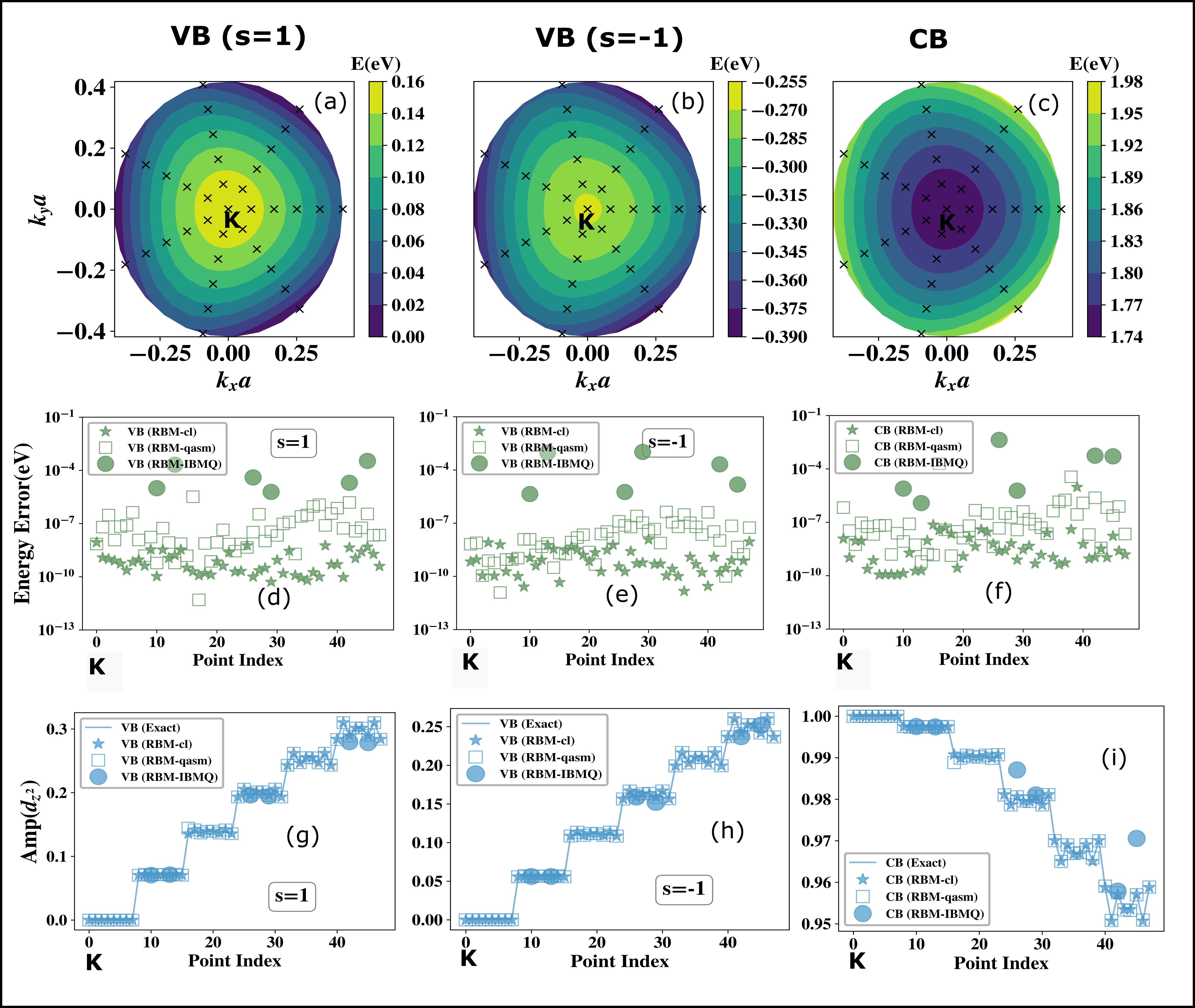}
    \caption{{\color{black}(a) The exact energy contours in valence band (s=1) within the three-band approximation for the Hamiltonian (see Eq. 8 in main manuscript) as a function of ($k_x$, $k_y$) near the $K$-point in $\rm{WS}_2$ (b) Same as in a) but for s=-1 (c) Same as in a) for the conduction band. The crosses in (a), (b) and (c) denotes the ($k_x$, $k_y$) pair wherein calculations for all three flavors of RBM have been executed. (d) Energy errors in eV from three flavors of RBM calculations for points denoted as cross in a) for the valence band (s=1) case computed using $\lambda=0$ in Eq.  \ref{cost_fn_supp} in $\rm{WS}_2$. The x-axis is a flattened point index with $(k_x, k_y)$ pairs marked as crosses in (a) mapped to integers such that the origin is at the $K$-point. From the $K$-point, the flattened point index scale moves spirally outwards grouping all $(k_x, k_y)$ pairs satisfying $|k| = \sqrt{k_x^2 + k_y^2}$ as consecutive integers and then proceeding to the next $|k|$ (e) Same as in d) but with points denoted in b) as crosses for other valence band with s=-1.(f) Same as in d) but for points denoted in c) as crosses for the conduction band computed with $(\lambda=5, \omega=0, \hat{O}=|\nu_0\rangle \langle \nu_0|)$ in Eq. \ref{cost_fn_supp} 
    (g) The amplitude for the occupancy of $d_{z^2}$ orbital on the metal for states computed at ($k_x$, $k_y$) pairs near the $K$-point from all three flavors of RBM as well as the exact states in valence band (s=1) for $\rm{WS}_2$. The amplitude of states with the same $|k| = \sqrt{k_x^2 + k_y^2}$ appear bunched together as 'steps' due to flattened point-index scale used. Near the $K$-point the amplitude is the same for all such pairs within a given step due to isotropy of the energy surface. However away from the $K$-point deviations appear due to trigonal warping owing to the $D_{3h}$ symmetry of the unit cells in TMDCs. The states from all flavors of RBM can resolve the influence of warping accurately with the performance worsened for the noisy variant.(h) Same as in g) for valence band (s=-1) (i) Same as in g) for conduction band. For all these calculations the warping parameters are kept the same as that for $\rm{MoS}_2$ even through the band energies are obtained within the three-band approximation calculated using RBM for $\rm{WS}_2$ in the main manuscript}}
    \label{Fig:WS2_SOC}
\end{figure*}
\newpage

\section{Eigenvectors of $L_z$ and $L^2$ operator for $\rm{MoS}_2$ and $\rm{WS}_2$}
The $L_z$ operator at $K-$ point in the three-band basis of $(d_{z^2}, d_{xy}, d_{x^2 - y^2})$ orbitals of the metal centre is given as \cite{Liu2013_s}
\begin{align}
\hat{L_z} &=
\begin{bmatrix}
0 & 0 & 0 \\
0 & 0 & 2i \\
0 & -2i & 0 \\
\end{bmatrix}
\end{align}
where $i=\sqrt{-1}$. The eigenvectors and eigenvalues of the $L_z$ and $L^2$ matrix is given as \\

\begin{tabular}{ |p{3cm}|p{3cm}|p{3cm}|p{3cm}|}
 \hline
 \multicolumn{4}{|c|}{Eigenvectors and Eigenvalues of $L_z/L^2$ operator} \\
 \hline
 Eigenvalue ($L_z$) & Eigenvalue ($L^2$) & Eigenvector & Band Index\\
 \hline
 -2   & 4 & $\frac{1}{\sqrt{2}}[0, -1, -i]^T$   & CB+1 \\
 0  & 0 & $[1, 0, 0]^T$   & CB \\
 2  &  4 & $\frac{1}{\sqrt{2}}[0, 1, -i]^T$   & VB \\
 \hline
\end{tabular}\\\\

Note that the basis is $d_{z^2} = [1,0,0]^T$, $d_{xy} = [0,1,0]^T$, $d_{x^2 - y^2} = [0,0,1]^T$ and the notation VB= valence band (ground state), CB = conduction band (1st excited state) and CB + 1 = Higher energy band (2nd excited state above conduction band). As mentioned in \cite{Liu2013_s}, in the chosen basis, the matrix elements of $L_x$ operator and $L_y$ operator are all zeros and hence $L^2$ operator enjoys exclusive contribution from $L_z$ operator given above. As a result, the eigenvectors of $L^2$ operator are the same as given in the table above but the eigenvalue pair (-2, 2) of $L_z$ maps to the same eigenspace of $L^2$ with eigenvalue =4. In other words $L^2$ has a doubly-degenrate eigenspace of eigenvalue =4 made from eigenvectors VB and CB+1 (see Table above) whereas a non-degenerate eigenspace of eigenvalue =0 with the eigenvector CB (see Table above).

\newpage

\section{Transferability of the Learning to Other systems}\label{Tr_learn_WSe2}

{\color{black} In this section we simulate the possibility of using our algorithm which has trained the network for one system and see if it is possible to 'transfer' the learning to get converged results in a different but closely related system. For unsupervised classical deep learning algorithms once 
the parameters of the neural network is tuned so that the probability distribution of the visible node mimics the unknown distribution of the training data closely, any new sample drawn from the visible node will be representative of a sample generated from the target distribution. No further training is necessary indicating that such models are highly transferable. However for quantum data as has been treated in this report, the meaning of transferability needs to be clarified. In our algorithm, the objective is to train the network to mimic the amplitude and phase distribution of an eigenstate of a given Hamiltonian $H$. Even if we use a fully trained network from some system yet some further amount of training would be necessary for a similar system as the Hamiltonian matrix $\hat{H}$ and the symmetry operator $\hat{O}$ are changing and hence the eigenstates of the new $\hat{H}$ are also slightly different than in the previous system. In that sense, the learning of our algorithm for quantum data is only partially transferable. That being said, a fully trained for a similar system does help and can enhance the rate of convergence and reduce the number of iterations for further re-training. We simulate this possibility using \textbf{`RBM-qasm'} and \textbf{`RBM-cl'} variant for $\rm{WSe}_2$ using the trained network for $\rm{MoS}_2$. Due to the symmetry partitioning of the metal orbitals as guaranteed in \cite{Liu2013_s}, to treat $\rm{WSe}_2$ too one would need $n=2$ qubits. We use $m=2$ as for other systems as well. The gate requirements is exactly the same as for $\rm{MoS}_2$ and $\rm{WS}_2$. We see that in both the variants starting with a trained network for $\rm{MoS}_2$ converged results can be obtained using just 1000-2000 iterations. All the results are displayed in Fig. \ref{Fig:Tr_Wse2}.
}
\newpage

\begin{figure}[!htb]
    \centering
    \includegraphics[width=1.0\textwidth]{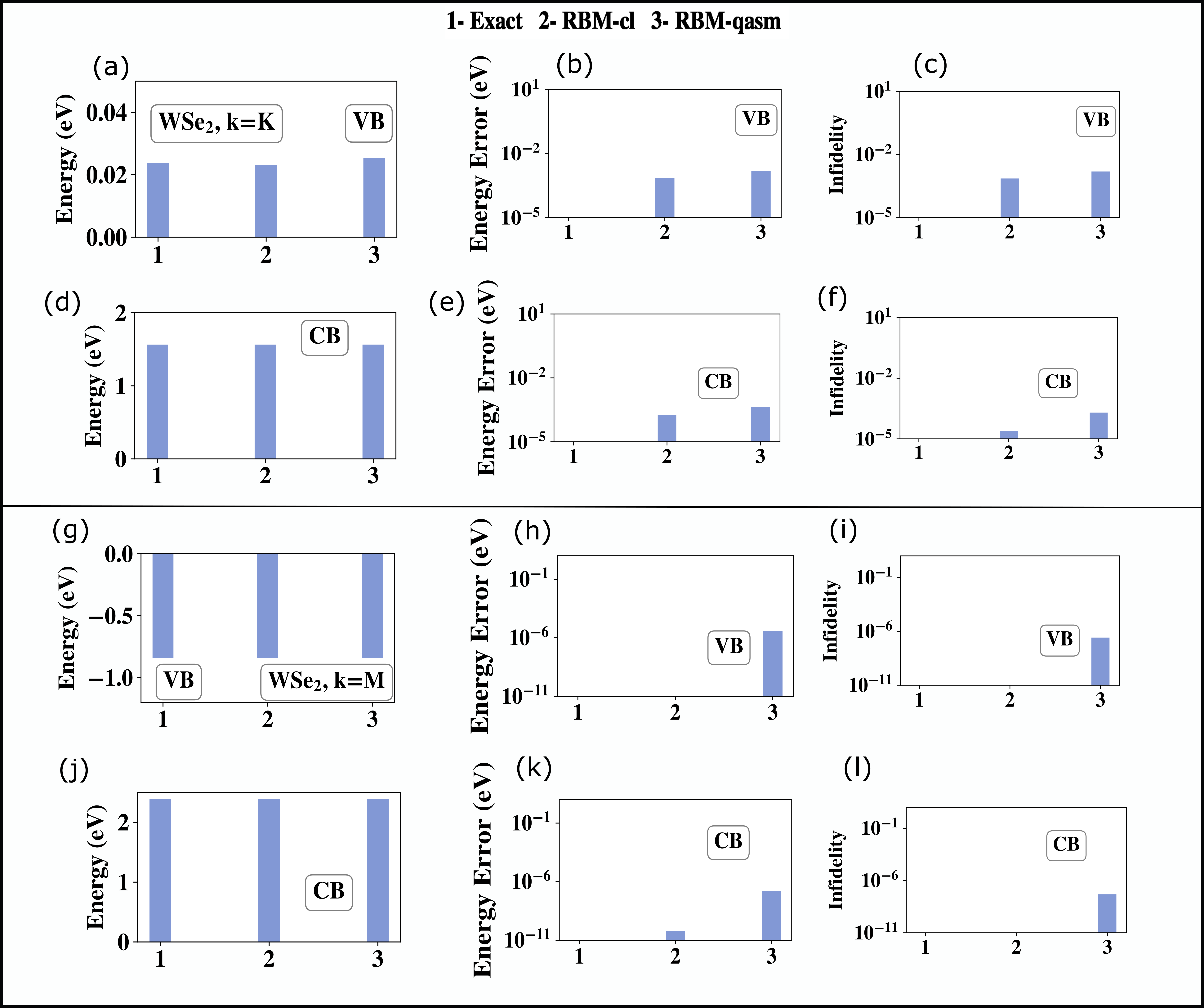}
    \caption{{\color{black}(a) The energy of the valence band (VB) at $(k_x, k_y)=K$ point for $\rm{WSe}_2$ in the three-band model \cite{Liu2013_s} for the exact case and the two flavors of RBM namely \textbf{`RBM-qasm'} and \textbf{`RBM-cl'}. The calculations are done by starting from a network trained with the converged results for the VB of $\rm{MoS}_2$ at the $K$ point. (b) The corresponding energy errors from the calculations in a) (c) The corresponding state infidelities from the calculations in a) (d) Similar result as in a) but for the conduction band (CB) at $(k_x, k_y)=K$ point for $\rm{WSe}_2$. The calculations are done by starting from a network trained with the converged results for the CB of $\rm{MoS}_2$ at the $K$ point (e) The corresponding energy errors from the calculations in d) (f) The corresponding state infidelities from the calculations in d). (g)-(l) are  results for $\rm{WSe}_2$ similar to (a)-(f) but at a different symmetry point i.e. $(k_x, k_y)=M$ point. The calculations in this case are done by starting from a network trained with the converged results for the VB/CB of $\rm{MoS}_2$ at the $M$ point}}
    \label{Fig:Tr_Wse2}
\end{figure}

\newpage
\section{LiH- A molecular example}\label{LiH}

{\color{black}In this section we benchmark the performance of our algorithm using a molecular example $\rm{LiH}$ in STO-6G basis set. The said basis generates 6 molecular orbitals (MOs) for LiH out of which the lower lying two MOs (HOMO and HOMO-1) are completely filled each with 2 electrons corresponding to the two spin orbitals. We freeze the core (HOMO-1) and use the HOMO and the LUMO in the active space\cite{Xia2018_s}. The LUMO+1, LUMO+2, LUMO+3 are a part of the virtual orbital set and are always unoccupied. The Hamiltonian matrix so constructed is a $16\times 16$ matrix . We therefore need $n=4$ neurons in the visible node and 4 qubits in the quantum circuit for these visible neurons. For the hidden neurons we have used $m=4$ as well as $m=6$ to show how the results change by varying hidden node density. The total number of qubits in the quantum circuit for the case $m=4$ is then $4+4+16 = 24$ where 16 ancillary qubits ($mn$) are used for mediating the interacting terms in the RBM distribution. For $m=6$ we consequently need 34 qubits. The total number of gates for the case of $m=4$ is thus 8 single qubit $R_y$ gates and 16 $C-C-R_y$ gates. For the case of $m=6$ the corresponding numbers are 10 single qubit $R_y$ gates and  24 $C-C-R_y$ gates. 

We simulate the ground and excited state potential energy surface as a function of stretching the $\rm{Li}$ and $\rm{H}$ bond length. The results are displayed in Fig. \ref{Fig:LiH_m4_n4} and in Fig. \ref{Fig:LiH_m6_n4} along with the corresponding errors from the exact diagonalization results. We use the 
\textbf{`RBM-qasm'} and \textbf{`RBM-cl'} variant for all benchmarking. While most runs on the \textbf{`RBM-cl'} variant are randomly initialized near the equilibrium bond length warm starting has been used extensively especially near dissociation limit where multi-reference correlation is important. For \textbf{`RBM-qasm'} variant the runs are sometimes warm started with the initial parameter set of a nearby bond length in the \textbf{`RBM-cl'} case. We see that in all cases away from the dissociation limit , errors are fairly low in the range of $10^{-5}-10^{-3}$ eV whereas it is $\leqslant 0.0022$ a.u. near the terminal bond lengths $r_{\rm{LiH}}$ studied in the plots ( $0.7\: \AA\:\: \leq r_{\rm{LiH}} \:\: \leq 2.8\: \AA$ )

In \cite{Xia2018_s} (see Fig. 4(b) and Fig. 4(d), the ground state potential energy surface for $\rm{LiH}$ is discussed with and without the effect of warm-starting (which the reference called `Transfer Learning'). Fig. 4(d) with warm-starting shows all points to have energy errors much lesser than in Fig. 4(b)) which is also found to be the case for the calculations in this report. The key differences between the results in Fig. \ref{Fig:LiH_m4_n4} and in \cite{Xia2018_s} needs to be emphasized at this point. \cite{Xia2018_s} simulated the results at the \textbf{`RBM-cl'} variant only as the circuit was not directly implementable on a quantum simulator (like Qiskit's $qasm  \textunderscore simulator$ as used in this report) or a NISQ device. Also \cite{Xia2018_s} used STO-3G as the basis set and number of hidden nodes $m=8$ unlike here. Besides the phase node included one neuron 
in \cite{Xia2018_s}. We see that for the ground state all errors in energy were $\le 0.001$ a.u. in \cite{Xia2018_s} which is likely due to the higher number of hidden nodes $m=8$ used which makes the network more expressive. We also changed the hidden node density and simulated the ground and the excited state for $\rm{LiH}$ in Fig. \ref{Fig:LiH_m6_n4} with $m=6$ hidden neurons apart from $m=4$. One must note that only \textbf{`RBM-cl'} variant is used in this case akin to \cite{Xia2018_s}. This is because to simulate the system at all bond lengths using $(n=4, m=6)$ neurons one would need 34 qubits in the quantum circuit which is beyond the current standards of $qasm  \textunderscore simulator$ in Qiskit Aer backend. We see the errors in the ground state calculation are usually lower in this case akin to \cite{Xia2018_s} but for excited states the trend is less clear. Possibly to lower the energy errors further for excited states one needs higher hidden node density than what is considered here. Although not used for Fig. \ref{Fig:LiH_m4_n4} or Fig. \ref{Fig:LiH_m6_n4}, optimizers like ADAM, RMSProp etc which are known to recover a neural network from locally trapped minima may also help.} 
\newpage
\begin{figure}[!htb]
    \centering
    \includegraphics[width=1.0\textwidth]{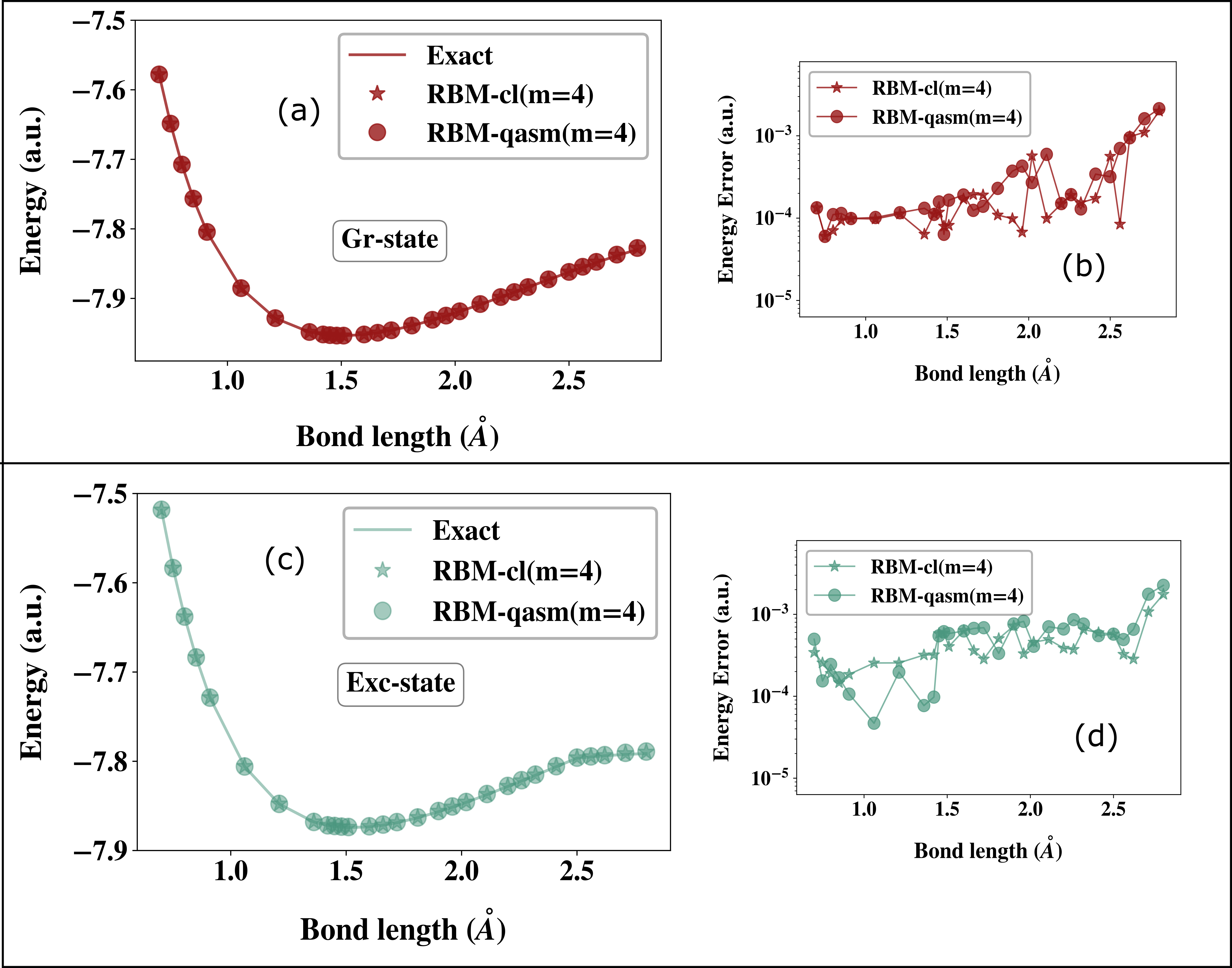}
    \caption{{\color{black}The dissociation curve for the ground state of $\rm{LiH}$ in 
    \textbf{`RBM-cl'} and 
    \textbf{`RBM-qasm'} variant overlayed against the exact value.
    (b) The error in energies from a) from the exact value.
    c)The dissociation curve for the excited state of $\rm{LiH}$ in 
    \textbf{`RBM-cl'} and 
    \textbf{`RBM-qasm'} variant overlayed against the exact value.
    (d) The error in energies from c) from the exact value.
    In all of the results in this panel we use $n=4$ and $m=4$
    }}
    \label{Fig:LiH_m4_n4}
\end{figure}
\newpage
\begin{figure}[!htb]
    \centering
    \includegraphics[width=1.0\textwidth]{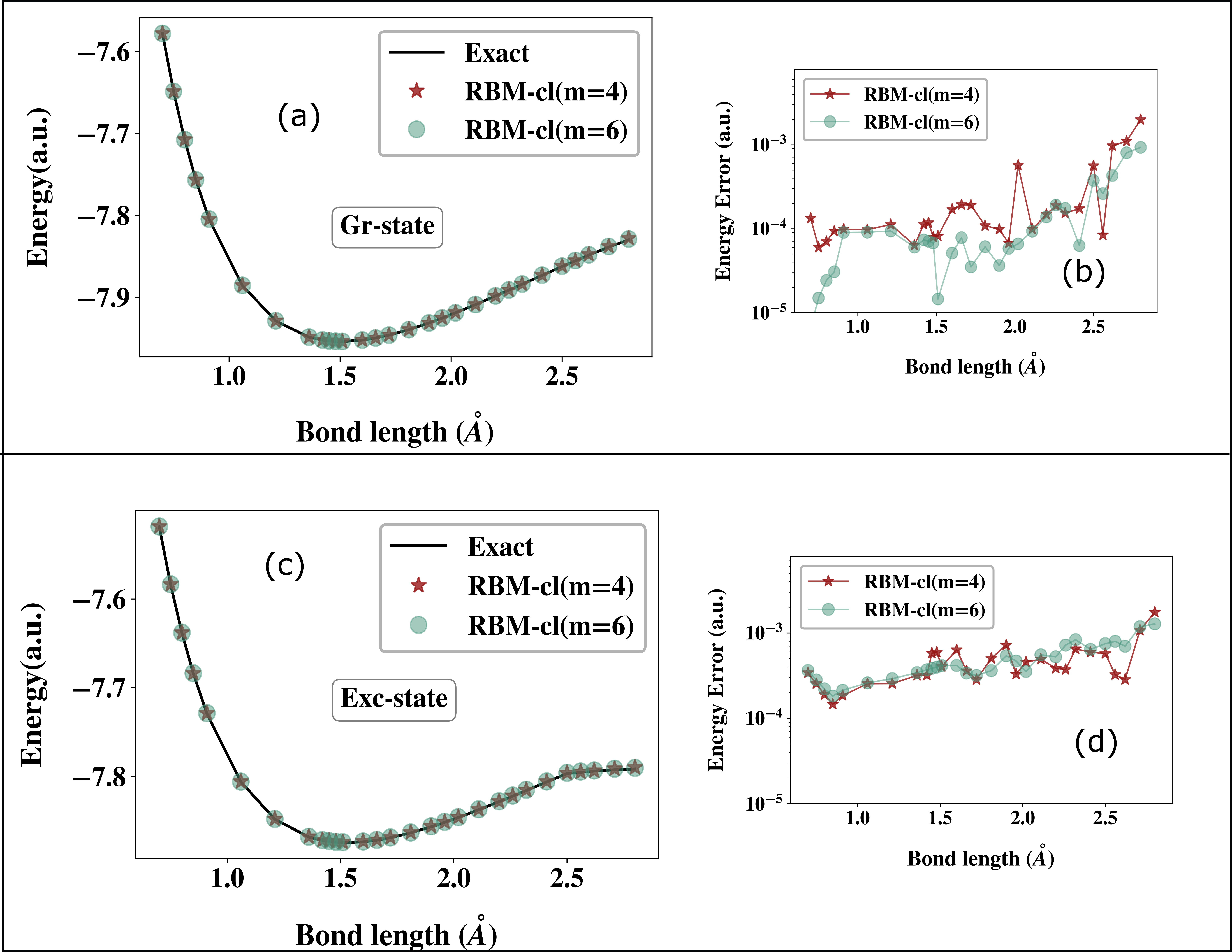}
    \caption{{\color{black}The dissociation curve for the ground state of $\rm{LiH}$ in 
    \textbf{`RBM-cl'}  variant overlayed against the exact value for two different $m$.
    (b) The error in energies from a) from the exact value for both $m$.
    c)The dissociation curve for the excited state of $\rm{LiH}$ in 
    \textbf{`RBM-cl'} variant overlayed against the exact value for two different $m$.
    (d) The error in energies from c) from the exact value for both $m$.
    In all of the results in this panel we use $n=4$ and $m=4$ and compare it with $n=4$ and $m=6$}}
    \label{Fig:LiH_m6_n4}
\end{figure}
\newpage


\end{document}